\documentclass[%
reprint,
nofootinbib,
amsmath,amssymb,
aps,
amsthm,
pra,
]{revtex4-1}

\usepackage{ifthen}
\newboolean{arxiv}
\setboolean{arxiv}{true}

\usepackage{theorem} 
\usepackage{ascmac}
\usepackage{bbm}
\usepackage{bm} 
\usepackage{braket}
\usepackage{dcolumn} 
\usepackage{enumerate}
\usepackage{enumitem}
\usepackage{framed}
\usepackage{graphicx} 
\ifthenelse{\boolean{arxiv}}{\usepackage{hyperref}}{\usepackage[dvipdfmx]{hyperref}}
\usepackage{longtable}
\usepackage{mathtools}
\usepackage{multirow}
\ifthenelse{\boolean{arxiv}}{}{\usepackage{pxjahyper}}
\usepackage{txfonts} 
\usepackage{xparse}

\usepackage{color}
\usepackage[most]{tcolorbox}

\ifthenelse{\boolean{arxiv}}{}{\mathtoolsset{showonlyrefs}}

\newcommand{\hypercolor}{blue}
\hypersetup{
  colorlinks,
  citecolor=\hypercolor,
  linkcolor=\hypercolor,
  urlcolor=\hypercolor,
  bookmarksopen,
  bookmarksopenlevel=4,
  bookmarksnumbered
}

\definecolor{memo}{RGB}{128,0,255}
\definecolor{gray}{RGB}{128,128,128}

\newcommand{\NumSymSharp}{1}
\newcommand{\NumCompletelyMixed}{2}
\newcommand{\NumFilter}{3}
\newcommand{\NumProcEq}{4}
\newcommand{\NumClassical}{5C}
\newcommand{\NumFullQuantum}{5Q}
\newcommand{\poslink}[2]{#2}

\newcommand{\pagetarget}[2]{%
 \phantomsection%
 \label{#1}%
 \hypertarget{#1}{#2}%
}

\newboolean{figure}
\setboolean{figure}{true}

\usepackage{bm}

\theorembodyfont{\upshape}
\newtheorem{thm}{Theorem}
\newtheorem{proposition}[thm]{Proposition}

\newtheorem{postulateno}{Postulate}

\newtheorem{lemma}[thm]{Lemma}

\newcounter{proof}

\ExplSyntaxOn
\NewDocumentEnvironment{proof}{o}
 {
  \par\medskip
  \noindent
  \textbf{Proof~}
 }
 {\QED\par\smallskip}
\ExplSyntaxOff

\newcounter{postulate}
\renewcommand{\thepostulate}{\arabic{postulate}}
\ExplSyntaxOn
\NewDocumentEnvironment{postulate}{oo}
 {
  \refstepcounter{postulate}
  \begin{postulateno}
  \textbf{\hspace{-0.5em}\IfNoValueTF{#2}{\thepostulate}{#2} ~(\IfNoValueTF{#1}{}{#1})}
 }
 {
  \end{postulateno}
 }
\ExplSyntaxOff

\newcounter{ex}

\ExplSyntaxOn
\NewDocumentEnvironment{ex}{o}
 {
  \refstepcounter{ex}
  \par\medskip
  \noindent
  \textbf{\IfNoValueTF{#1}{Example}{Example~of~#1}~}
 }
 {\par\smallskip}
\ExplSyntaxOff

\newcommand{\mB}{\mathcal{B}}
\newcommand{\mC}{\mathcal{C}}

\newcommand{\mF}{\mathcal{F}}

\newcommand{\mK}{\mathcal{K}}
\newcommand{\mL}{\mathcal{L}}
\newcommand{\mM}{\mathcal{M}}

\newcommand{\mP}{\mathcal{P}}

\newcommand{\mS}{\mathcal{S}}
\newcommand{\mT}{\mathcal{T}}

\newcommand{\mX}{\mathcal{X}}
\newcommand{\mY}{\mathcal{Y}}

\newcommand{\ident}{\hat{1}}

\newcommand{\Real}{\mathbf{R}}
\newcommand{\Complex}{\mathbf{C}}

\newcommand{\QED}{\hspace*{0pt}\hfill $\blacksquare$}

\newcommand{\Tr}{{\rm Tr}}
\newcommand{\rank}{{\rm rank}}
\newcommand{\supp}{{\rm supp}}

\def\gauss_sym#1{{\lfloor #1 \rfloor}}

\newcommand{\logeq}{\mathrel{\vcentcolon\Leftrightarrow}}
\renewcommand{\Pr}{\mathrm{Pr}}
\renewcommand{\ident}{\mathbbm{1}}
\newcommand{\id}{\mathrm{id}}
\newcommand{\gdis}{\raisebox{-.1em}{\includegraphics[scale=0.5]{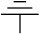}}}
\newcommand{\cross}{\times}

\renewcommand{\i}{\mathbf{i}}

\newcommand{\N}{\mathrm{N}}
\renewcommand{\P}{\mathrm{P}}
\newcommand{\F}{\mathrm{F}}
\newcommand{\System}{\mathbf{Syst}}
\newcommand{\Proc}{\mathbf{Proc}}
\newcommand{\ProcF}{\mathbf{Proc}^\F}
\newcommand{\ProcD}{\mathbf{Proc}^\mathrm{D}}
\newcommand{\St}{\mathbf{St}}
\newcommand{\StP}{\mathbf{St}^\P}
\newcommand{\StF}{\mathbf{St}^\F}
\newcommand{\StN}{\mathbf{St}^\N}
\newcommand{\StNP}{\mathbf{St}^{\N\P}}
\newcommand{\Eff}{\mathbf{Eff}}
\newcommand{\EffF}{\mathbf{Eff}^\F}
\newcommand{\EffM}{\mathbf{Eff}^\mathrm{M}}
\newcommand{\Scalar}{\mathbf{Scalar}}
\newcommand{\ScalarF}{\mathbf{Scalar}^\F}
\newcommand{\PDS}{\mathbf{PDS}}
\newcommand{\MPDS}{\mathbf{MPDS}}
\newcommand{\Test}{\mathbf{Test}}
\newcommand{\Meas}{\mathbf{Meas}}
\newcommand{\EJA}{\mathbf{E}}
\renewcommand{\Vec}{\mathbf{V}}
\newcommand{\Mat}{\mM}
\newcommand{\StMat}{\mathbf{M}}
\newcommand{\tStMat}{\tilde{\StMat}}

\newcommand{\Face}{\mF}
\renewcommand{\ker}{\mK}

\newcommand{\NA}{{N_A}}
\newcommand{\NB}{{N_B}}

\newcommand{\tphi}{\tilde{\phi}}
\newcommand{\tvarphi}{\tilde{\varphi}}

\newcommand{\tE}{\tilde{E}}
\newcommand{\tF}{\tilde{F}}

\newcommand{\tf}{\tilde{f}}

\newcommand{\Realpp}{\Real_{++}}

\newcommand{\Field}{\mathbf{F}}
\newcommand{\Hermite}{\mathbf{H}}
\newcommand{\Oct}{\mathbf{O}}
\newcommand{\Spin}{\mathbf{Spin}}
\newcommand{\tr}{\mathrm{tr}}

\renewcommand{\ol}{\overline}
\newcommand{\eqlocal}{\overset{\mathrm{local}}{=}}

\newcommand{\tPhi}{\tilde{\Phi}}
\newcommand{\Endash}{\text{\textendash}}
\newcommand{\termdef}[1]{\textit{#1}}

\setlist[enumerate]{label=(\arabic*), leftmargin=3em, itemsep=0pt, parsep=0pt, labelwidth=5em}
\setlist[description]{labelindent=2em, leftmargin=3em, itemsep=0pt, parsep=0pt, labelwidth=5em}
\setlist[itemize]{leftmargin=2em, itemsep=0pt, parsep=0pt, labelwidth=4em}

\newcommand{\craket}[1]{\mathinner{({#1})}}
\newcommand{\cra}[1]{{\mathinner{({#1}|}}}
\newcommand{\cket}[1]{{\mathinner{|{#1})}}}
\let\protect\relax
{\catcode`\|=\active
  \xdef\Craket{\protect\expandafter\noexpand\csname Craket \endcsname}
  \expandafter\gdef\csname Craket \endcsname#1{\begingroup
     \ifx\SavedDoubleVert\relax
       \let\SavedDoubleVert\|\let\|\BraDoubleVert
     \fi
     \mathcode`\|32768\let|\BraVert
     \left({#1}\right)\endgroup}
}

\newcommand{\Memo}[1]{}
\newcommand{\Disappear}[1]{}

\begin{document}

\preprint{APS/123-QED}

\title{Derivation of quantum theory with superselection rules}

\affiliation{%
 Quantum Information Science Research Center, Quantum ICT Research Institute, Tamagawa University,
 Machida, Tokyo 194-8610, Japan
}%

\author{Kenji Nakahira}
\affiliation{%
 Quantum Information Science Research Center, Quantum ICT Research Institute, Tamagawa University,
 Machida, Tokyo 194-8610, Japan
}%

\date{\today}

\begin{abstract}
 We reconstruct finite-dimensional quantum theory with superselection rules,
 which can describe hybrid quantum-classical systems,
 from four purely operational postulates:
 symmetric sharpness, complete mixing, filtering, and local equality.
 It has been shown that each of the classical and fully quantum theories is singled out
 by an additional postulate.
\end{abstract}

\maketitle



\section{Introduction}

Although quantum theory has enjoyed great success as an approach for explaining
the behavior of the microscopic world,
we still lack a deep and intuitive understanding of its principles.
This is in contrast to special and general relativity, which are based on simple physical principles.
In 1932, von Neumann presented an abstract mathematical formulation of quantum theory
based on complex Hilbert spaces \cite{Neu-1932}.
However, as he pointed out ``
I do not believe in Hilbert space anymore'' \cite{Bir-1961},
it is thus noteworthy that this formulation is far from a clear understanding of
the physical structure of quantum theory.
For a proper understanding of the reason why the complex Hilbert space
(or $C^*$-algebraic) formalism is relevant in describing the microscopic world,
Birkoff and von Neumann\cite{Bir-Neu-1936}, Zierler\cite{Zie-1961}, Mackey\cite{Mac-1963},
Jauch and Piron\cite{Jau-Pir-1963}, Ludwig\cite{Lud-1983},
and many other researchers have investigated the reconstruction of
the mathematical structure of quantum theory
(simply referred to as the derivation of quantum theory)
from physically meaningful principles.

Around 2000, a research program was launched by Fuchs and Brassard (e.g., \cite{Fuc-2001}),
the goal of which was to reconstruct quantum theory from a few simple information-theoretic principles.
Since then, an increasing number of studies have adopted an information-theoretic approach
\cite{Har-2001,Cli-Bub-Hal-2003,Bar-2007,Dar-2007,Dak-Bru-2009,Mas-Mul-2011,Chi-Dar-Per-2011,Har-2011,
Wil-2012-axioms,Zao-2012,Bar-Mul-Udu-2014,Har-2016,Tul-2018,
Goy-2008,Rau-2009,Fuc-2011,Rau-2011,Fiv-2012,Hoh-Wev-2017,Kru-Bar-Bar-Mul-2017,
Wet-2018,Wet-2018-sequential,Wet-2019}.
Many of these studies are based on the operational probabilistic theory (OPT)
or other similar theories such as the generalized probabilistic theory
\cite{Har-2001,Cli-Bub-Hal-2003,Bar-2007,Dar-2007,Dak-Bru-2009,Mas-Mul-2011,Chi-Dar-Per-2011,Har-2011,
Wil-2012-axioms,Zao-2012,Bar-Mul-Udu-2014,Har-2016,Tul-2018}.
An OPT offers a powerful approach for ensuring a deeper understanding of the operational
and information-theoretic aspects of quantum processes.
In 2001, Hardy proposed a novel approach for deriving ordinary quantum theory,
i.e., quantum theory without superselection rules
(which we denote as fully quantum theory),
using five informational axioms \cite{Har-2001}.
However, some of these axioms are expressed in mathematical terms that
have no clear physical underpinning.
Subsequently, Dakic and Bruckner \cite{Dak-Bru-2009} and Masanes and M\"{u}ller \cite{Mas-Mul-2011}
attempted reconstruction approaches based on more sophisticated axioms.
Chiribella et al. \cite{Chi-Dar-Per-2011} and Hardy \cite{Har-2011}
succeeded in reconstructing fully quantum theory using purely operational postulates.
To avoid technical difficulties resulting from infinite dimensions,
many studies have focused on finite-dimensional quantum theory,
which still exhibits all the essential quantum phenomena.

In quantum information science, effective use of classical systems
in combination with fully quantum systems can be crucial.
For instance, the control of a fully quantum system using classical information,
such as the outcomes of measurements, is crucial in many fields
(e.g., local operations and classical communication,
quantum teleportation, and one-way quantum computing).
Furthermore, fully quantum theory is arguably not self-contained
as the outcome of a measuring process is essentially classical in nature.
To deal with classical and quantum information in a unified way
requires the application of quantum theory to handle hybrid quantum-classical systems.
Several studies have been conducted on hybrid quantum-classical systems
and we will mention only a few examples of such studies herein:
the interaction between quantum and classical systems
(e.g., \cite{Fra-Bel-And-Com-2012,Leg-Fra-Soa-Hor-2015}),
formulation of hybrid quantum-classical dynamics
(e.g., \cite{Hal-Reg-2005,Elz-2012,Bar-Car-Gar-Gom-2012}),
quantumness of correlations in quantum states (e.g., \cite{Hen-Ved-2001,Pia-Hor-Hor-2008}),
and quantum coherence (e.g., \cite{Chi-Str-Ran-Ber-2016,Str-Ade-Ple-2017}).
However, it is noteworthy that all the results mentioned in the previous paragraph
cannot be applied to a quantum theory having hybrid quantum-classical systems.
Therefore, in this paper, we will derive quantum theory with superselection rules,
which can describe classical systems, fully quantum systems, and hybrid quantum-classical systems,
from purely operational postulates only.

Some properties that can be used as postulates (or axioms) to single out fully quantum theory
are not satisfied in classical systems:
an instance of this is the so-called purification postulate \cite{Chi-Dar-Per-2010,Chi-Dar-Per-2011,Tul-2018},
which claims that every state has a purification.
Moreover, there exist postulates that can be used to derive
both fully quantum theory and classical probability theory (which we refer to as classical theory)
but are not satisfied in hybrid quantum-classical systems;
instances of such postulates are as follows:
(i) there exists a deterministic reversible process
between any pair of normalized pure states \cite{Har-2001,Dak-Bru-2009,Mas-Mul-2011,Zao-2012};
(ii) each system is characterized by a natural number that is referred to as dimension or capacity
\cite{Har-2001,Mas-Mul-2011};
and (iii) there exists a deterministic reversible process producing any given permutation
of any given maximal set of perfectly distinguishable normalized pure states
\cite{Har-2011,Bar-Mul-Udu-2014}.
The derivation of quantum theory with superselection rules from an OPT
is presented by Barnum and Wilce \cite{Bar-Wil-2014},
Selby et al. \cite{Sel-Sca-Coe-2018}, Wilce \cite{Wil-2018}, and Jia \cite{Jia-2018}.
However, in Refs.~\cite{Bar-Wil-2014,Jia-2018}, mathematical assumptions that are
not stated in operational terms were adopted as postulates.
Also, some postulates adopted in Refs.~\cite{Sel-Sca-Coe-2018,Wil-2018}
(e.g., symmetric purification or the existence of a conjugate system)
seem to be difficult to intuitively comprehend,
at least for readers that are unfamiliar with quantum theory%
\footnote{Certainly, Wilce noted that there remains some mystery as to the proper interpretation of
the conjugate system \cite{Wil-2018}.
Based on some postulates, he showed that each effect space is a symmetric cone;
however, he did not provide the complete derivation of quantum theory.
Also, in Ref.~\cite{Sel-Sca-Coe-2018}, there seems to be a gap in the proof of Theorem~4.14,
as presented in the footnote \ref{fn:Sel} of this paper.}.
In contrast, the set of our postulates is stated in operational terms
and provides an intuitive interpretation;
in particular, each of our four postulates is easy to intuitively understand
in the context of classical theory.
Moreover, we show that each of the fully quantum and classical theories can be singled out
using an additional operational postulate.
In this paper, we consider only finite-dimensional systems.

The paper is organized as follows.
In Sec.~\ref{sec:OPT}, we present a review of the framework of OPTs.
In Sec.~\ref{sec:Pustulates}, we present a set of four operational postulates
and overview the derivation of quantum theory with superselection rules
(which we simply refer to as quantum theory subsequently).
In Secs.~\ref{sec:derive_PDS} and \ref{sec:derive_symmetry},
we consider an OPT satisfying the first three postulates
and show that each state space is the cone of squares of a Euclidean Jordan algebra (EJA).
In Sec.~\ref{sec:derive_EJA}, we present a review of the fundamental properties
of EJAs that will be used in Sec.~\ref{sec:derive_Quantum}.
In Sec.~\ref{sec:derive_Quantum},
we discuss an OPT with the four postulates and
show that quantum theory is the only theory consistent with these postulates.

\section{Operational probabilistic theory (OPT)} \label{sec:OPT}

In this section, we introduce the framework of OPTs.
The framework can be explained in several ways, leading to essentially almost the same formalism.
Some basic facts about an OPT are reviewed.
The proofs of some of the results are not presented in this paper;
for the details, we refer the reader to Refs.~\cite{Har-2001,Cli-Bub-Hal-2003,Bar-2007,Dar-2007,
Dak-Bru-2009,Mas-Mul-2011,Chi-Dar-Per-2011,Har-2011,Wil-2012-axioms,Zao-2012,
Bar-Mul-Udu-2014,Har-2016,Tul-2018,Chi-Dar-Per-2010}.
To provide an intuitive grasp of an operational interpretation, we will use
diagrammatic representations, which is motivated by the work of Coecke, Abramsky,
and others (see, e.g., \cite{Coe-2003,Abr-Coe-2004,Coe-Kis-2017}).
Any diagrammatic representation can be faithfully and rigorously described in a mathematical formula.

\subsection{Operational theory}

First, we introduce an operational theory.
This theory describes the compositional structure of physical processes.

\subsubsection{Processes, states, and effects}

An operational theory consists of a collection of \termdef{systems} and
a collection of \termdef{processes}.
Systems could represent a physical system such as a photon.
Processes could represent a particular behavior of a physical device such as a beam splitter.
Systems are labeled by capital letters from the beginning of the alphabet ($A, B, \ldots$).
$\System$ is defined as the set of all systems.
Each process has an input system and an output system;
a process having an input system $A$ and an output system $B$ is called
a process from $A$ to $B$.
$\Proc_{A \to B}$ is defined as the set of all processes from $A$ to $B$.

We introduce an `empty' system, called a \termdef{trivial system} and denoted by $I$.
A process from $I$ to $A$ (which is a process with no input) is called a \termdef{state preparation},
or simply a \termdef{state}, of $A$ and is denoted like $\cket{\rho}$ in analogy to Dirac's bra-ket notation
(note that $\cket{\rho}$ could be a mixed state, which will be defined later).
$\St_A \coloneqq \Proc_{I \to A}$ is called the \termdef{state space} of system $A$.
Similarly, a process from $A$ to $I$ (i.e., a process with no output) is
called an \termdef{effect} of $A$ and is denoted like $\cra{e}$.
$\Eff_A \coloneqq \Proc_{A \to I}$ is called the \termdef{effect space} of system $A$.
An effect represents an event associated with a particular outcome of a measurement.
A process from $I$ to $I$ (i.e., a process with no input and no output)
is called a \termdef{scalar}.
Let $\Scalar \coloneqq \Proc_{I \to I}$.

In diagrammatic terms,
a process $f \in \Proc_{A \to B}$, a state $\cket{\rho} \in \St_A$,
an effect $\cra{e} \in \Eff_A$, and a scalar $p \in \Scalar$ are depicted as
\begin{alignat}{1}
 \includegraphics[scale=1.0]{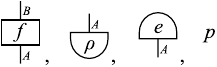} ~\raisebox{.2em}{.}
\end{alignat}
Systems are represented by labeled wires (labels are often omitted).
The trivial system $I$ is represented by `no wire'.
Processes are represented by boxes that have an input wire at the bottom
and an output wire at the top.
For a scalar, the box will be omitted.
Diagrammatic representations can be considered as something like data flow diagrams
with time increasing from the bottom to the top.

\begin{ex}[quantum theory]
 The state space $\St_A$ and the effect space $\Eff_A$ are
 isomorphic to $\bigoplus_{i=1}^k \mS_+(\Complex^{n_i})$,
 where $\Complex$ is the set of all complex numbers
 and $\mS_+(\Complex^n)$ is the set of all complex positive semidefinite matrices of order $n$.
 The natural numbers $k, n_1, \ldots, n_k$ are determined by the system $A$.
 $\NA \coloneqq \sum_{i=1}^k n_i$ is called the \termdef{rank} of $A$.
 System $A$ is called \termdef{classical} if $k = \NA$
 (i.e., $\St_A \cong \Eff_A \cong \bigoplus_{i=1}^\NA \mS_+(\Complex) \cong \Real_+^\NA$)
 holds,
 where $\mX \cong \mY$ denotes that $\mX$ is isomorphic to $\mY$,
 and $\Real_+$ is the set of all nonnegative real numbers.
 Classical theory is a special case of quantum theory in which every system is classical.
 System $A$ is called \termdef{fully quantum} if $k = 1$
 (i.e., $\St_A \cong \Eff_A \cong \mS_+(\Complex^\NA)$) holds.
 Fully quantum theory is a special case of quantum theory in which every system is fully quantum.
 Note that a system $A$ with $\NA = 1$ is classical and fully quantum.
 The state space of the trivial system, $\St_I = \Scalar$,
 is isomorphic to $\mS_+(\Complex) \cong \Real_+$.
 When $\St_A \cong \bigoplus_{i=1}^{k_A} \mS_+(\Complex^{m_i})$ and
 $\St_B \cong \bigoplus_{j=1}^{k_B} \mS_+(\Complex^{n_j})$ hold,
 $\Proc_{A \to B}$ is isomorphic to the space of all completely positive (CP)
 maps from $\bigoplus_{i=1}^{k_A} \mS(\Complex^{m_i})$ to
 $\bigoplus_{j=1}^{k_B} \mS(\Complex^{n_j})$,
 where $\mS(\Complex^n)$ is the set of all complex Hermitian matrices of order $n$.
 In the examples of quantum theories, we will identify a process with its corresponding CP map;
 in particular, we will identify a state (or effect) with the corresponding
 positive semidefinite matrix.
\end{ex}

\subsubsection{Sequential / parallel composition of processes}

Two processes can be composed sequentially if the output system of one is
equal to the input system of the other.
The sequential composition of $f \in \Proc_{A \to B}$ and $g \in \Proc_{B \to C}$
is a process from $A$ to $C$, which is denoted as $g \circ f \in \Proc_{A \to C}$.
The composition, $\cra{e} \circ \cket{\rho} \in \Scalar$,
of a state $\cket{\rho} \in \St_A$ and an effect $\cra{e} \in \Eff_A$
is also denoted by $\cra{e} \cket{\rho}$ or $\craket{e|\rho}$.
$g \circ f$ and $\craket{e|\rho}$ are respectively depicted as
\begin{alignat}{1}
 \includegraphics[scale=1.0]{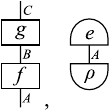} ~\raisebox{.2em}{.}
 \label{eq:process_circ}
\end{alignat}
Whenever we write $g \circ f$ for two processes $f$ and $g$,
we always assume that the output system of $f$ and the input system of $g$ are equal.
Sequential composition is associative, i.e.,
$h \circ (g \circ f) = (h \circ g) \circ f$ holds for
any $f \in \Proc_{A \to B}$, $g \in \Proc_{B \to C}$, and $h \in \Proc_{C \to D}$,
where $f = f'$ denotes that $f$ and $f'$ are indistinguishable in an operational theory.

Two systems $A$ and $B$ can be composed in parallel to yield a new system, denoted $A \otimes B$.
The composition is associative, i.e.,
$A \otimes (B \otimes C) = (A \otimes B) \otimes C$ holds for any $A,B,C \in \System$.
$I \otimes A$ and $A \otimes I$ are identified with $A$ itself.
The parallel composition of $f \in \Proc_{A \to B}$ and $g \in \Proc_{C \to D}$,
denoted as $f \otimes g$, is a process from $A \otimes C$ to $B \otimes D$.
Diagrammatically, $f \otimes g$ is depicted as a pair of processes arranged side by side:
\begin{alignat}{1}
 \includegraphics[scale=1.0]{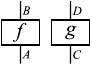} ~\raisebox{.5em}{.}
 \label{eq:process_otimes}
\end{alignat}
The parallel composition of processes is also associative,
i.e., $f \otimes (g \otimes h) = (f \otimes g) \otimes h$ holds
for processes $f$, $g$, and $h$.
A collection of connected processes will be called a \termdef{diagram}.
Assume that
\begin{alignat}{1}
 (g_1 \otimes g_2) \circ (f_1 \otimes f_2) &= (g_1 \circ f_1) \otimes (g_2 \circ f_2)
 \label{eq:process_product2}
\end{alignat}
holds for four any processes $f_1$, $f_2$, $g_1$, and $g_2$,
which is diagrammatically represented as
\begin{alignat}{1}
 \includegraphics[scale=1.0]{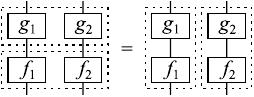} ~\raisebox{.5em}{,}
 \label{eq:process_product2_diagram}
\end{alignat}
where the auxiliary lines (dashed lines) are only intended to guide the eye.

\begin{ex}[quantum theory]
 The sequential composition, $g \circ f$, of two processes $f \in \Proc_{A \to B}$
 and $g \in \Proc_{B \to C}$ is
 the CP map defined as $(g \circ f)[\cket{\rho}] \coloneqq g[f[\cket{\rho}]]$
 for any $\cket{\rho} \in \St_A$.
 In particular, $\craket{e|\rho} = \Tr[\cra{e} \cdot \cket{\rho}]$ holds,
 where $\cdot$ indicates the matrix product.
 When $\St_A \cong \bigoplus_{i=1}^{k_A} \mS_+(\Complex^{m_i})$ and
 $\St_B \cong \bigoplus_{j=1}^{k_B} \mS_+(\Complex^{n_j})$ hold,
 $\St_{A \otimes B} \cong \bigoplus_{i=1}^{k_A} \bigoplus_{j=1}^{k_B} \mS_+(\Complex^{m_in_j})$ holds.
 The parallel composition, $\cket{\rho} \otimes \cket{\sigma}$, of two states
 $\cket{\rho}$ and $\cket{\sigma}$ is
 the tensor product of the matrices $\cket{\rho}$ and $\cket{\sigma}$.
 The parallel composition, $f \otimes h$, of two processes $f \in \Proc_{A \to B}$
 and $h \in \Proc_{C \to D}$ is
 the CP map defined as $(f \otimes h)[\cket{\rho} \otimes \cket{\sigma}] \coloneqq
 f[\cket{\rho}] \otimes h[\cket{\sigma}]$
 for any $\cket{\rho} \in \St_A$ and $\cket{\sigma} \in \St_C$.

 We present two examples of systems that are neither classical nor fully quantum.
 The first one is hybrid quantum-classical systems.
 Consider a classical system $A$ and a fully quantum one $B$ with $\NA,\NB > 1$.
 Since $\St_A \cong \bigoplus_{i=1}^\NA \mS_+(\Complex)$ and
 $\St_B \cong \mS_+(\Complex^\NB)$ hold,
 the state space of the composite system $A \otimes B$ satisfies
 $\St_{A \otimes B} \cong \bigoplus_{i=1}^\NA \mS_+(\Complex^\NB)$.
 The second one is a system $C$ with superselection rules for the total number of particles,
 In such a system, coherent superpositions between states with a different number of particles
 cannot be observed.
 Assume that $C$ consists of at most $K$ particles.
 The state space $\St_C$ can be decomposed $\St_C = \bigoplus_{k=0}^K \St_C^{(k)}$,
 where $\St_C^{(k)}$ is the subspace with the number of particles equal to $k$.
 If each subspace $\St_C^{(k)}$ can be thought of as a fully quantum system,
 then $\St_C^{(k)} \cong \mS_+(\Complex^{n_k})$ holds for some natural number $n_k$.
 In this case, $\St_C$ is isomorphic to $\bigoplus_{k=0}^K \mS_+(\Complex^{n_k})$.
\end{ex}

\subsubsection{Identity processes and swap processes}

Assume that, for each system $A$, there exists a `do nothing' process on $A$,
which is called an \termdef{identity process} and denoted by $\id_A \in \Proc_{A \to A}$
or simply $\id$.
Diagrammatically, $\id_A$ is depicted as
\begin{alignat}{1}
 \includegraphics[scale=1.0]{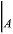} ~\raisebox{.5em}{.}
 \label{eq:process_id_def}
\end{alignat}
$\id_I$ is depicted as the `empty space'.
Identity processes satisfy the following conditions
for any systems $A$ and $B$ and any process $f \in \Proc_{A \to B}$:
\begin{enumerate}
 \item $f \circ \id_A = f = \id_B \circ f$.
 \item $\id_A \otimes \id_B = \id_{A \otimes B}$.
 \item $f \otimes \id_I = f = \id_I \otimes f$.
\end{enumerate}
It follows that there is a unique identity process for each system.
The above property~(1) is depicted as:
\begin{alignat}{1}
 \includegraphics[scale=1.0]{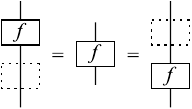} ~\raisebox{.0em}{,}
 \label{eq:process_id_f}
\end{alignat}
where the auxiliary boxes indicate the identity processes.
It is easily seen that, for any processes $f$ and $g$,
$(\id \circ f) \otimes (g \circ \id) = f \otimes g = (f \circ \id) \otimes (\id \circ g)$
holds, which is diagrammatically represented as
\begin{alignat}{1}
 \includegraphics[scale=1.0]{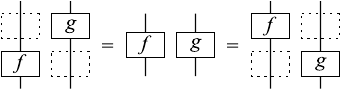} ~\raisebox{.5em}{.}
 \label{eq:process_fg}
\end{alignat}
Intuitively, this means that the vertical shifts of processes do not affect diagrams.

We also assume that, for any systems $A$ and $B$,
there exists a process $\cross_{A \otimes B} \in \Proc_{A \otimes B \to B \otimes A}$,
called a \termdef{swap process} and diagrammatically depicted as
\begin{alignat}{1}
 \includegraphics[scale=1.0]{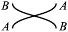}
 ~\raisebox{.0em}{,}
 \label{eq:cross_def}
\end{alignat}
such that
\begin{alignat}{1}
 \includegraphics[scale=1.0]{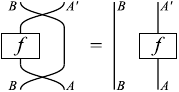} \label{eq:cross_f}
\end{alignat}
(i.e., $\cross_{A',B} \circ (f \otimes \id_B) \circ \cross_{B,A} = \id_B \otimes f$)
holds for any systems $A$, $A'$, and $B$ and any process $f \in \Proc_{A \to A'}$%
\footnote{More strictly, we also assume that $\cross_{A,I} = \id_A$ and
$\cross_{A,B \otimes C} = (\id_B \otimes \cross_{A,C}) \circ (\cross_{A,B} \otimes \id_C)$
hold for any systems $A$, $B$, and $C$.}.

\begin{ex}[quantum theory]
 The identity process is the identity map.
 The swap process $\cross_{A,B}$ is the process satisfying
 $\cross_{A,B} \circ [\cket{\rho} \otimes \cket{\sigma}] =
 \cket{\sigma} \otimes \cket{\rho}$
 for any $\cket{\rho} \in \St_A$ and $\cket{\sigma} \in \St_B$.
\end{ex}

Note that although no knowledge of category theory is required to read this paper,
this theory provides a suitable framework to describe an operational theory.
Specifically, an operational theory can be considered as a category with
systems as objects and processes as morphisms.
In particular, this category is a strict symmetric monoidal category with the tensor unit $I$
and the tensor product $\otimes$.

\subsubsection{Properties of scalars}

We introduce the following diagram
consisting of two processes $u_1 \in \Proc_{C \to A \otimes E}$ and
$u_2 \in \Proc_{B \otimes E \to D}$
\begin{alignat}{1}
 \includegraphics[scale=1.0]{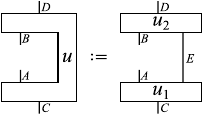} ~\raisebox{.5em}{.}
 \label{eq:u}
\end{alignat}
This diagram, denoted as $u:\Proc_{A \to B} \to \Proc_{C \to D}$,
maps $f \in \Proc_{A \to B}$ to
$u(f) \coloneqq u_2 \circ (f \otimes \id_E) \circ u_1 \in \Proc_{C \to D}$.
Any process from $C$ to $D$ that includes $f \in \Proc_{A \to B}$
is expressed in the form $u(f) \in \Proc_{C \to D}$
with some diagram $u:\Proc_{A \to B} \to \Proc_{C \to D}$;
for example,
\begin{alignat}{1}
 \includegraphics[scale=1.0]{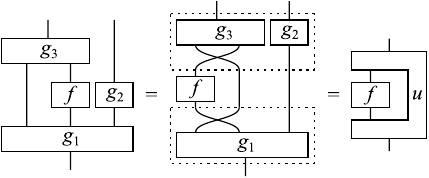} ~\raisebox{1.0em}{,}
 \label{eq:uf_example}
\end{alignat}
where $u_1$ and $u_2$ are, respectively,
the processes enclosed by the lower and upper auxiliary boxes.

For a scalar $a$ and a process $f$, $a \otimes f$ is denoted by $af$ or $a \cdot f$.
A process expressed in the form $af$ with a scalar $a$ and a process $f$ will be
referred to as \termdef{scalar multiplication} of $f$.
Two processes $f, g \in \Proc_{A \to B}$ will be referred to as \termdef{proportional},
denoted by $f \propto g$, if there exists a scalar $a$ satisfying
either $f = a g$ or $a f = g$.

It follows that, for any scalar $a$, process $f \in \Proc_{A \to B}$, and diagram
$u:\Proc_{A \to B} \to \Proc_{C \to D}$,
$u(a f) = a \cdot u(f)$ holds%
\footnote{Proof: We obtain $u(a f) = u_2 \circ (a f \otimes \id_E) \circ u_1
 = (\id_I \otimes u_2) \circ (a \otimes f \otimes \id_E) \circ (\id_I \otimes u_1)
 = (\id_I \circ a \circ \id_I) \otimes [u_2 \circ (f \otimes \id_E) \circ u_1]
 = a \cdot u(f)$,
 where the third equality follows from Eq.~\eqref{eq:process_product2_diagram}.},
which is diagrammatically represented as
\begin{alignat}{1}
 \includegraphics[scale=1.0]{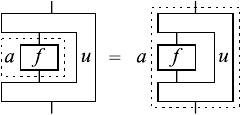} ~\raisebox{.5em}{.}
 \label{eq:af_u0}
\end{alignat}
This implies that any scalar can freely move around a diagram;
for example, $[(af) \circ g] \otimes h = [f \circ (ag)] \otimes h = (f \circ g) \otimes (ah)
= a [(f \circ g) \otimes h]$ holds for any processes $f$, $g$, $h$ and any scalar $a$.
Consider the particular case in which $A = B = C = D = E = I$
and $u_2 = \id_I$ and let $b \coloneqq u_1 \in \Scalar$; then,
$u:\Scalar \to \Scalar$ satisfies $u(a) = a \circ b$ for any $a \in \Scalar$.
Also, from $a = a \cdot \id_I$ and Eq.~\eqref{eq:af_u0},
$u(a) = a \cdot u(\id_I) = a b$ holds.
Thus, for any two scalars $a$ and $b$, $a \circ b = a b$ holds,
i.e., the sequential and parallel compositions of two scalars are equal.

\begin{ex}[quantum theory]
 For any $a,b \in \Scalar = \Real_+$,
 $a \cdot b = a \otimes b = a \circ b$ is the multiplication of two nonnegative real numbers
 $a$ and $b$.
 $a f$ with $a \in \Scalar$ and $f \in \Proc_{A \to B}$ is
 the CP map that satisfies $(af)[\cket{\rho}] = a \cdot f[\cket{\rho}]$
 for any $\cket{\rho} \in \St_A$.
\end{ex}

\subsection{OPT} \label{subsec:OPT_OPT}

Next, we review an OPT.
This theory is an operational theory that assigns probabilities for processes.

\subsubsection{Tests and feasible processes}

A \termdef{test} is a mathematical model that represents the behavior of a physical device.
A test having an input system $A$ and an output system $B$
consists of $k$ processes $f_1,\ldots,f_k \in \Proc_{A \to B}$
that represent mutually exclusive and collectively exhaustive events.
$f \in \Proc_{A \to B}$ is called \termdef{deterministic} if $\{ f \}$ is a test.
Intuitively, a deterministic process represents an event that always happens.
$\Test_{A \to B}$ and $\ProcD_{A \to B}$ are, respectively, defined as the sets of
all tests and deterministic processes from $A$ to $B$.
Each element of $\Meas_A \coloneqq \Test_{A \to I}$ is called
a \termdef{measurement} of $A$.
Assume that, for each system $A$, there exist a deterministic state of $A$
and a deterministic effect of $A$.

As well as processes, tests can be composed sequentially and in parallel.
The sequential composition of $\{ f_j \in \Proc_{A \to B} \}_{j=1}^J \in \Test_{A \to B}$
and
$\{ g_l \in \Proc_{B \to C} \}_{l=1}^L \in \Test_{B \to C}$
is $\{ g_l \circ f_j \}_{(j,l)=(1,1)}^{(J,L)} \in \Test_{A \to C}$.
The parallel composition of $\{ f_j \in \Proc_{A \to B} \}_{j=1}^J \in \Test_{A \to B}$
and $\{ h_m \in \Proc_{C \to D} \}_{m=1}^M \in \Test_{C \to D}$ is
$\{ f_j \otimes h_m \}_{(j,m)=(1,1)}^{(J,M)} \in \Test_{A \otimes C \to B \otimes D}$.
Since the sequential composition, $\{ g \circ f \}$, of two tests $\{ f \}$ and $\{ g \}$
with $f \in \ProcD_{A \to B}$ and $g \in \ProcD_{B \to C}$ is
a test, the sequential composition of two deterministic processes is deterministic.
Also, the parallel composition of two deterministic processes is deterministic.

We will call a process $f$ \termdef{feasible} if
there exists a test including $f$.
Intuitively, a feasible process represents a randomly occurring event.
Let $\ProcF_{A \to B}$ be the set of all feasible processes from $A$ to $B$;
then, $\ProcD_{A \to B} \subseteq \ProcF_{A \to B} \subseteq \Proc_{A \to B}$
obviously holds.
Also, we define $\StF_A \coloneqq \ProcF_{I \to A}$,
$\EffF_A \coloneqq \ProcF_{A \to I}$,
and $\ScalarF \coloneqq \ProcF_{I \to I}$.

\begin{ex}[quantum theory]
 A set of processes (i.e., CP maps), $\{ f_j \}_{j=1}^k$, is a test if and only if
 the CP map $\sum_{j=1}^k f_j$ is trace-preserving (TP).
 A set of effects $\Pi \coloneqq \{ \cra{e_j} \in \Eff_A \}_{j=1}^k$
 is a measurement if and only if $\sum_{j=1}^k \cra{e_j} = \ident_\NA$ holds,
 where $\ident_n$ is the identity matrix of order $n$.
 This means that $\Pi$ is a positive operator-valued measure.
 $f \in \ProcD_{A \to B}$ means that $f$ is a TP-CP map.
 Also, $f \in \ProcF_{A \to B}$ means that $f$ is a trace non-increasing CP map.
 In particular, $\cket{\rho} \in \StF_A$ is equivalent to $\Tr~\cket{\rho} \le 1$,
 and $\cra{e} \in \EffF_A$ is equivalent to $\cra{e} \le \ident_\NA$,
 where $H \le H'$ (or $H' \ge H$) with Hermitian matrices $H$ and $H'$ denotes that
 $H' - H$ is positive semidefinite.
\end{ex}

\subsubsection{Assigning Probabilities}

Assume that, for each feasible state $\cket{\rho}$,
a probability for $\cket{\rho}$ to happen is assigned
and is dependent only on $\cket{\rho}$.
Let us denote this probability by $\Pr[\cket{\rho}]$,
where $\Pr$ is a function from $\StF_A$ to $[0,1]$
(where $[0,1]$ denotes the set of all real numbers between $0$ and $1$, inclusive).
Also, assume that each $p \in \ScalarF$ is identified with $\Pr(p)$
and that $\ScalarF = [0,1]$ holds.
This means that any feasible scalar can be interpreted as a probability.
A feasible state $\cket{\rho}$ satisfying $\Pr[\cket{\rho}] = 1$ is called
a \termdef{normalized state}.
Let $\StN_A$ be the set of all normalized states of $A$.

Assume that, for any $A,B,E \in \System$,
$f \in \ProcF_{A \to B}$, and $\cket{\sigma} \in \StF_{A \otimes E}$,
the probability of the joint occurrence of the state $\cket{\sigma}$ and the process $f$
is $\Pr[f, \cket{\sigma}] \coloneqq \Pr[(f \otimes \id_E) \circ \cket{\sigma}]$.
Also, assume that, for any two feasible scalars $p$ and $q$,
the probability of the joint occurrence of $p$ and $q$
(i.e., $\Pr(p,q) = \Pr(p q) = p q$)
is the product of probabilities $p$ and $q$,
which means that $p q$ is the product of two real numbers $p$ and $q$.
This can be interpreted to indicate that an event that $p$ happens and one that $q$ happens
to occur independently.
The case $q = \id_I$ yields $p \cdot \id_I = p$, which gives $\id_I = 1$.

It is natural to assume that
$f_1,\ldots,f_k \in \Proc_{A \to B}$ represent
mutually exclusive and collectively exhaustive events
if, for any feasible state $\cket{\sigma} \in \StF_{A \otimes E}$,
the sum of the probabilities of the joint occurrence of $\cket{\sigma}$ and $f_j$
is equal to the probability of occurrence of $\cket{\sigma}$.
This assumption implies
\begin{alignat}{1}
 \lefteqn{ \{ f_j \}_{j=1}^k \in \Test_{A \to B} } \nonumber \\
 &\Leftrightarrow\quad
 \forall E \in \System, \cket{\sigma} \in \StF_{A \otimes E}, \;
 \sum_{j=1}^k \Pr[f_j,\cket{\sigma}] = \Pr[\cket{\sigma}].
 \label{eq:process_test_nas0}
\end{alignat}

\begin{ex}[quantum theory]
 $\Pr[\cket{\rho}] = \Tr~\cket{\rho}$ and
 $\Pr[f,\cket{\sigma}] = \Tr[(f \otimes \id_E)[\cket{\sigma}]]$ hold.
\end{ex}

\subsubsection{Deterministic effects}

We denote one of the deterministic effects of a system $A$ as $\cra{\gdis_A}$
or simply $\cra{\gdis}$, which is depicted as
\begin{alignat}{1}
 \includegraphics[scale=1.0]{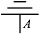} ~\raisebox{.3em}{.}
 \label{eq:discard}
\end{alignat}
From Eq.~\eqref{eq:process_test_nas0},
$\craket{\gdis_A|\rho} = \Pr[\cket{\rho}]$ obviously holds
for any $\cket{\rho} \in \StF_A$.
This can be interpreted that any deterministic effect always happens.
In particular, $\cra{\gdis_I} = 1$ holds.
For any systems $A$ and $B$,
we choose $\cra{\gdis_{A \otimes B}} \coloneqq \cra{\gdis_A} \otimes \cra{\gdis_B}$,
which is depicted as
\begin{alignat}{1}
 \includegraphics[scale=1.0]{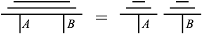} ~\raisebox{.3em}{.}
 \label{eq:discard_AB}
\end{alignat}

\begin{ex}[quantum theory]
 $\cra{\gdis_A}$ is the identity matrix $\ident_\NA$.
\end{ex}

\subsubsection{Unfeasible processes}

In this paper, for the sake of mathematical convenience, we assume $\Scalar = \Real_+$.
A real number larger than 1 is an unfeasible scalar.
Unfeasible scalars are not quite intuitive
since they cannot be interpreted as probabilities; however,
they allow for a simple mathematical analysis,
as will be seen throughout this paper.
For any $a \in \Scalar$ and $f \in \ProcF_{A \to B} \subseteq \Proc_{A \to B}$,
$a f \in \Proc_{A \to B}$ obviously holds.
Assume that any unfeasible process is expressed as
scalar multiplication of a feasible process, i.e.,
\begin{alignat}{1}
 \Proc_{A \to B} &= \{ a f : a \in \Scalar, f \in \ProcF_{A \to B} \}
 \label{eq:procecc_Proc_def}
\end{alignat}
holds.
Such $\Proc_{A \to B}$ can be interpreted as the smallest conceivable process space
that is consistent with an operational theory satisfying $\Scalar = \Real_+$.

Recall that $\Pr[\cket{\rho'}] = \craket{\gdis_A|\rho'}$ holds
for any $\cket{\rho'} \in \StF_A$.
We here extend the domain of the function $\Pr$
by letting $\Pr[\cket{\rho}] \coloneqq \craket{\gdis_A|\rho}$ for any $\cket{\rho} \in \St_A$.
Note that $\Pr[\cket{\rho}]$ can be larger than 1.
Also, let $\Pr[f,\cket{\sigma}] \coloneqq \Pr[(f \otimes \id_E) \circ \cket{\sigma}]
= \cra{\gdis} \circ (f \otimes \id_E) \circ \cket{\sigma}$
for any $f \in \Proc_{A \to B}$ and $\cket{\sigma} \in \St_{A \otimes E}$.

\subsubsection{Properties of tests and feasible processes}

Since any $\cket{\sigma} \in \St_{A \otimes E}$ is expressed as
$\cket{\sigma} = a \cket{\sigma'}$ with $a \in \Scalar$ and $\cket{\sigma'} \in \StF_{A \otimes E}$,
$\Pr[\cket{\sigma}] = \Pr[a \cket{\sigma'}] = a \cdot \Pr[\cket{\sigma'}]$ and
$\Pr[f,\cket{\sigma}] = \Pr[f,a \cket{\sigma'}] = a \cdot \Pr[f,\cket{\sigma'}]$ obviously hold.
Thus, Eq.~\eqref{eq:process_test_nas0} also holds
if $\cket{\sigma} \in \StF_{A \otimes E}$
is replaced by $\cket{\sigma} \in \St_{A \otimes E}$.
Moreover,
\begin{alignat}{2}
 \Pr[f,\cket{\sigma}] &\le \Pr[\cket{\sigma}], &\quad
 &\forall E \in \System, \cket{\sigma} \in \St_{A \otimes E}
 \label{eq:process_feasible_Pr}
\end{alignat}
holds for any $f \in \ProcF_{A \to B}$.
Considering the particular case of $B = E = I$,
we have that, for any $\cra{e} \in \EffF_A$,
\begin{alignat}{2}
 \craket{e|\sigma} &\le \craket{\gdis|\sigma}, &\quad
 &\forall \cket{\sigma} \in \St_A.
 \label{eq:effect_feasible_Pr}
\end{alignat}

We have that, for any $\cket{\rho} \in \St_A$,
\footnote{Proof: Since $\Pr[\cket{\rho},\cket{\sigma}]
= \craket{\gdis_A|\rho} \craket{\gdis_E|\sigma}$ holds
for any $\cket{\rho} \in \St_A$ and $\cket{\sigma} \in \St_{I \otimes E} = \St_E$,
Eq.~\eqref{eq:process_test_nas0} gives that
$\cket{\rho} \in \ProcD_{I \to A}$ is equivalent to $\craket{\gdis|\rho} = 1$.
Also, from the definition of a normalized state,
$\craket{\gdis|\rho} = 1$ is equivalent to $\cket{\rho} \in \StN_A$
(note that if $\craket{\gdis|\rho} = 1$ holds, then $\cket{\rho} \in \ProcD_{I \to A}$ holds,
and thus $\cket{\rho}$ is feasible).}
\begin{alignat}{2}
 \cket{\rho} \in \ProcD_{I \to A} &\quad\Leftrightarrow\quad
 \craket{\gdis|\rho} = 1 &&\quad\Leftrightarrow\quad
 \cket{\rho} \in \StN_A.
 \label{eq:process_state_determine}
\end{alignat}
Using Eq.~\eqref{eq:process_state_determine},
one can easily see that Eq.~\eqref{eq:process_test_nas0} can be replaced with%
\footnote{Proof: For any $f \in \Proc_{A \to B}$ and
$\cket{\sigma} \in \St_{A \otimes E}$,
$\Pr[f,\cket{\sigma}]
= [\cra{\gdis_A} \otimes \cra{\gdis_E}] \circ (f \otimes \id_E) \circ \cket{\sigma}
= \cra{\gdis_A} \circ f \circ \cket{\sigma_A} = \Pr[f,\cket{\sigma_A}]$ holds,
where $\cket{\sigma_A} \coloneqq [\id_A \otimes \cra{\gdis_E}] \circ \cket{\sigma} \in \St_A$.
Also, $\Pr[\cket{\sigma}] = \craket{\gdis_{A \otimes E}|\sigma} = \craket{\gdis_A|\sigma_A}
= \Pr[\cket{\sigma_A}]$ holds.
Substituting these relations into Eq.~\eqref{eq:process_test_nas0} gives Eq.~\eqref{eq:process_test_nas}.}
\begin{alignat}{1}
 \{ f_j \}_{j=1}^k \in \Test_{A \to B}
 &\quad\Leftrightarrow\quad
 \cket{\rho} \in \St_A, \;
 \sum_{j=1}^k \Pr[f_j,\cket{\rho}] = \Pr[\cket{\rho}].
 \label{eq:process_test_nas}
\end{alignat}
We also have
\begin{alignat}{2}
 \cket{\rho} \in \StF_A &\quad\Leftrightarrow\quad
 \craket{\gdis|\rho} \le 1,
 \label{eq:process_feasible_nas_state}
\end{alignat}
where $\Leftarrow$ follows from the fact that
$\{ \cket{\rho}, [1 - \craket{\gdis|\rho}] \cket{\sigma} \}$
with $\cket{\sigma} \in \StN_A$ is a test.

\subsubsection{Equality and local equality of processes} \label{subsubsec:OPT_OPT_eq}

Recall that any scalar that includes $f \in \Proc_{A \to B}$
is expressed in the form $u(f)$ with some diagram $u:\Proc_{A \to B} \to \Scalar$.
In an OPT,
two processes $f, f' \in \Proc_{A \to B}$ are said to be equal and denoted by $f = f'$ if
\begin{alignat}{1}
 u(f) = u(f'), &\quad \forall u:\Proc_{A \to B} \to \Scalar
 \label{eq:proc_unique}
\end{alignat}
holds.
$f = f'$ means that $f$ and $f'$ are indistinguishable from a
probabilistic point of view.
It follows that $f = f'$ holds if and only if
\begin{alignat}{1}
 u(f) = u(f'), &\quad \forall C,D \in \System, \; u:\Proc_{A \to B} \to \Proc_{C \to D}
 \label{eq:proc_unique_CD}
\end{alignat}
holds%
\footnote{Proof: The case $C = D = I$ immediately yields the ``if'' part.
To prove the ``only if'' part, we assume that $u(f) \neq u(f')$ holds
for some diagram $u:\Proc_{A \to B} \to \Proc_{C \to D}$
and prove $f \neq f'$.
From the definition of equality of processes,
there exists a diagram $u':\Proc_{C \to D} \to \Scalar$
such that $u'[u(f)] \neq u'[u(f')]$.
Let $v:\Proc_{A \to B} \to \Scalar$ be the diagram
satisfying $v(g) \coloneqq u'[u(g)]$ $~(\forall g \in \Proc_{A \to B})$;
then, $v(f) \neq v(f')$ holds.
Thus, $f \neq f'$ holds.}.

A process $\emptyset \in \Proc_{A \to B}$ that satisfies
$u(\emptyset) = 0$ for any $u:\Proc_{A \to B} \to \Scalar$
is called the \termdef{zero process}.
Note that, for any systems $A$ and $B$, there is a unique zero process from $A$ to $B$.
It follows from Eq.~\eqref{eq:af_u0} that
$0 \cdot f = \emptyset$ holds for any $f \in \Proc_{A \to B}$
since $u(0 \cdot f) = 0 \cdot u(f) = 0 = u(\emptyset)$ holds for any $u:\Proc_{A \to B} \to \Scalar$.
The zero process from $I$ to $A$ is called the \termdef{zero state} and denoted by $\cket{\emptyset} \in \St_A$.
Similarly, the zero process from $A$ to $I$ is called the \termdef{zero effect}
and denoted by $\cra{\emptyset} \in \Eff_A$.

Let us discuss the equality of states.
For any diagram $u:\St_B \to \Scalar$ and $\cket{\rho} \in \St_B$,
$u_1$ of Eq.~\eqref{eq:u} is in $\St_{I \otimes E} = \St_E$,
and thus
\begin{alignat}{1}
 \includegraphics[scale=1.0]{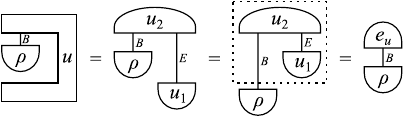}
 \label{eq:state_effect_unique}
\end{alignat}
holds, where the second equality follows from Eq.~\eqref{eq:process_fg},
and the effect enclosed by the auxiliary box,
$\cra{u_2} \circ [\id_B \otimes \cket{u_1}] \in \Eff_B$, is denoted by
$\cra{e_u}$.
Thus, a diagram $u:\St_B \to \Scalar$ can be expressed by the effect $\cra{e_u} \in \Eff_B$.
It follows from this and Eq.~\eqref{eq:proc_unique} that
two states $\cket{\rho}, \cket{\rho'} \in \St_B$ are equal if
\begin{alignat}{1}
 \craket{e|\rho} = \craket{e|\rho'}, &\quad \forall \cra{e} \in \Eff_B
 \label{eq:st_unique}
\end{alignat}
holds.
It also follows that $\cket{\rho} \in \St_A$ is the zero state if and only if
$\craket{\gdis|\rho} = 0$ holds%
\footnote{Proof: The ``only if'' part is obvious.
For the ``if'' part, assume $\craket{\gdis|\rho} = 0$.
From Eq.~\eqref{eq:effect_feasible_Pr},
any $\cra{e'} \in \EffF_A$ satisfies $0 \le \craket{e'|\rho} \le \craket{\gdis|\rho} = 0$,
i.e., $\craket{e'|\rho} = 0$.
For any $\cra{e} \in \Eff_A$,
which is in the form $\cra{e} = a \cra{e'}$
with $a \in \Scalar$ and $\cra{e'} \in \EffF_A$,
we have $\craket{e|\rho} = a \craket{e'|\rho} = 0 = \craket{e|\emptyset}$.
Thus, Eq.~\eqref{eq:st_unique} gives $\cket{\rho} = \cket{\emptyset}$.}.
Any $\cket{\rho} \in \St_A$ is expressed in the form%
\footnote{Proof: When $\cket{\rho} = \cket{\emptyset}$ holds,
from $\craket{\gdis|\rho} = 0$,
Eq.~\eqref{eq:rho_normalize} holds for any $\cket{\rho_0} \in \StN_A$.
Consider the case $\cket{\rho} \neq \cket{\emptyset}$.
Let $a \coloneqq \craket{\gdis|\rho} \neq 0$ and
$\cket{\rho_0} \coloneqq a^{-1} \cket{\rho}$;
then, Eq.~\eqref{eq:rho_normalize} obviously holds.
Moreover, from $\craket{\gdis|\rho_0} = a^{-1} \craket{\gdis|\rho} = 1$,
Eq.~\eqref{eq:process_state_determine} gives $\cket{\rho_0} \in \StN_A$.}
\begin{alignat}{1}
 \includegraphics[scale=1.0]{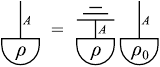}
 \label{eq:rho_normalize}
\end{alignat}
with a certain $\cket{\rho_0} \in \StN_A$.
It follows from Eq.~\eqref{eq:rho_normalize} that,
for any nonzero state $\cket{\rho}$,
$\craket{\gdis|\rho}^{-1} \cket{\rho}$ is a normalized state.
Intuitively, a nonzero feasible state $\cket{\rho}$ represents
the process of preparing the normalized state $\craket{\gdis|\rho}^{-1} \cket{\rho}$
with probability $\craket{\gdis|\rho}$.
We can discuss the equality of effects as well as that of states.
Specifically, two effects $\cra{e}, \cra{e'} \in \Eff_A$ are equal if
\begin{alignat}{1}
 \craket{e|\rho} = \craket{e'|\rho}, &\quad \forall \cket{\rho} \in \St_A
 \label{eq:eff_unique}
\end{alignat}
holds.
It is easy to verify that $\cra{\gdis_A}$ is the unique deterministic effect of $A$.
Indeed, assume $\cra{e} \in \Eff_A$ is deterministic;
then, $\craket{e|\rho} = \craket{\gdis_A|\rho}$ holds for any $\cket{\rho} \in \St_A$,
and thus $\cra{e} = \cra{\gdis_A}$ holds%
\footnote{Some papers (e.g., \cite{Chi-Dar-Per-2010,Chi-Dar-Per-2011})
adopt an OPT in which each system does not necessarily have a unique deterministic effect.
Such an OPT has non-fixed causal structure.
In this paper, we restrict our attention to fixed causal structure.}.

Two processes $f,f' \in \Proc_{A \to B}$ are said to be \termdef{locally equal}
and denoted by $f \eqlocal f'$ if
\begin{alignat}{1}
 \cra{e} \circ f \circ \cket{\rho} = \cra{e} \circ f' \circ \cket{\rho},
 &\quad \forall \cra{e} \in \Eff_B, \cket{\rho} \in \St_A
\end{alignat}
holds.
We have that, from Eq.~\eqref{eq:st_unique},
\begin{alignat}{1}
 f \eqlocal f' &\quad\Leftrightarrow\quad
 f \circ \cket{\rho} = f' \circ \cket{\rho}, \quad \forall \cket{\rho} \in \St_A.
 \label{eq:eqlocal_frho}
\end{alignat}
$f = f'$ is a sufficient condition for $f \eqlocal f'$, but not necessary.
Note that, in the particular case of $A = I$ or $B = I$,
the equality and local equality are equivalent,
as already shown in Eqs.~\eqref{eq:st_unique} and \eqref{eq:eff_unique}.
We can easily obtain
\begin{alignat}{1}
 f = f' &\quad\Leftrightarrow\quad
 f \otimes \id_E \eqlocal f' \otimes \id_E, \quad \forall E \in \System.
 \label{eq:eq_ffp}
\end{alignat}

A process $f \in \Proc_{A \to B}$ is called \termdef{reversible} if
there exists $\tf \in \Proc_{B \to A}$, called an \termdef{inverse} of $f$, such that
$\tf \circ f \eqlocal \id_A$ and $f \circ \tf \eqlocal \id_B$.
$\tf$ is also reversible.
Note that, for a reversible process $f$, there may be more than one inverse of $f$.

\begin{ex}[quantum theory]
 In quantum theory, two processes $f,f' \in \Proc_{A \to B}$ are equal
 if and only if they are locally equal.
 The ``only if'' part is obvious, so we prove only ``if'' part.
 Assume $f \eqlocal f'$.
 From Eq.~\eqref{eq:eqlocal_frho},
 $f[\cket{\rho}] = f'[\cket{\rho}]$ holds for any $\cket{\rho} \in \St_A$.
 It is well known that any $\cket{\sigma} \in \St_{A \otimes E}$ can be expressed in the form
 $\cket{\sigma} = \sum_{i=1}^l c_i \cket{\rho_i} \otimes \cket{\rho'_i}$,
 where $c_1,\ldots,c_l \in \Real$, $\cket{\rho_1},\ldots,\cket{\rho_l} \in \St_A$,
 and $\cket{\rho'_1},\ldots,\cket{\rho'_l} \in \St_E$ hold.
 Thus, we have
 $(f \otimes \id_E)[\cket{\sigma}] = \sum_{i=1}^l c_i f[\cket{\rho_i}] \otimes \cket{\rho'_i}
 = \sum_{i=1}^l c_i f'[\cket{\rho_i}] \otimes \cket{\rho'_i} = (f' \otimes \id_E)[\cket{\sigma}]$.
 Therefore, $f \otimes \id_E \eqlocal f' \otimes \id_E$ holds for any system $E$,
 which gives $f = f'$ from Eq.~\eqref{eq:eq_ffp}.

 A process $f \in \Proc_{A \to B}$ is reversible if and only if
 $\NA = \NB$ holds and
 $f$ is written in the form $f[\cket{\rho}] = E \cdot \cket{\rho} \cdot E^\dagger$
 with an invertible matrix $E$ of order $\NA$,
 where $^\dagger$ denotes the conjugate transpose.
 Its inverse, $\tf$, is written as
 $\tf[\cket{\rho}] = E^{-1} \cdot \cket{\rho} \cdot (E^{-1})^\dagger$.
\end{ex}


\subsubsection{Sum of processes}

Assume that, for any two feasible processes $g_1,g_2 \in \ProcF_{A \to B}$
and any $p \in \ScalarF$,
there exists a feasible process $h \in \ProcF_{A \to B}$ satisfying
\begin{alignat}{1}
 u(h) &= p \cdot u(g_1) + (1-p) \cdot u(g_2), \quad \forall u:\Proc_{A \to B} \to \Scalar.
 \nonumber \\
 \label{eq:process_weighted_sum}
\end{alignat}
Such a process $h$, denoted by $p g_1 + (1-p) g_2$,
can be interpreted as a probabilistic mixture of
$g_1$ and $g_2$ with probabilities $p$ and $1-p$.
For any two processes $f_1,f_2 \in \Proc_{A \to B}$,
$f_\mathrm{sum} \in \Proc_{A \to B}$ is
called the \termdef{sum of $f_1$ and $f_2$} and denoted by $f_1 + f_2$
if it satisfies
\begin{alignat}{2}
 u(f_\mathrm{sum}) &= u(f_1) + u(f_2), &\quad &\forall u:\Proc_{A \to B} \to \Scalar.
 \label{eq:process_sum2}
\end{alignat}
Since the right-hand side of Eq.~\eqref{eq:process_weighted_sum} equals to $u(p g_1) + u[(1-p)g_2]$,
the feasible process $p g_1 + (1-p) g_2$ equals to the sum of
two processes $p g_1$ and $(1-p)g_2$.
It is easily seen that, for any $f_1,f_2 \in \Proc_{A \to B}$,
the process $f_1 + f_2$ always exists%
\footnote{Proof: For each $i \in \{1,2\}$, $f_i$ is expressed in the form $f_i = a_i f'_i$
with $a_i \in \Scalar$ and $f'_i \in \ProcF_{A \to B}$.
We assume, without loss of generality, that $a_1,a_2 > 0$.
Substituting
$g_1 = f'_1$, $g_2 = f'_2$, and $p = \frac{a_1}{a}$ (where $a \coloneqq a_1 + a_2$)
into Eq.~\eqref{eq:process_weighted_sum}
gives $h \coloneqq \frac{a_1}{a} f'_1 + \frac{a_2}{a} f'_2 \in \ProcF_{A \to B}$.
Since $u(a h) = a \cdot u(h) = a \cdot [\frac{a_1}{a} u(f'_1) + \frac{a_2}{a} u(f'_2)]
= u(f_1) + u(f_2)$ holds, we have $f_1 + f_2 = a \cdot h \in \Proc_{A \to B}$.},
i.e., the sum of two processes is always a process.
In particular, the sum of two scalars $p_1, p_2 \in \Scalar$ (i.e., $p_1 + p_2$) is
equal to the sum of the real numbers.

It is easily seen that
any $f_1,\ldots,f_k \in \Proc_{A \to B}$ and $a_1,\ldots,a_k \in \Scalar$ satisfy
\begin{alignat}{2}
 u\left( \sum_{j=1}^k a_j f_j \right) &= \sum_{j=1}^k a_j u(f_j)
 &\quad &\forall u:\Proc_{A \to B} \to \Proc_{C \to D},
\end{alignat}
which is diagrammatically depicted as
\begin{alignat}{1}
 \includegraphics[scale=1.0]{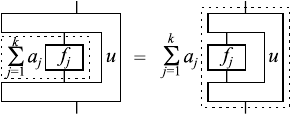} ~\raisebox{.5em}{.}
 \label{eq:fsum_u}
\end{alignat}
This implies that the diagram $u$ distributes over addition.

We consider a set of process $F \coloneqq \{ f_j \in \Proc_{A \to B} \}_{j=1}^k$.
It follows from Eq.~\eqref{eq:process_test_nas0} that
$F \in \Test_{A \to B}$ and
$\left\{ \sum_{j=1}^k f_j \right\} \in \Test_{A \to B}$ are equal,
which yields
\begin{alignat}{1}
 \{ f_j \}_{j=1}^k \in \Test_{A \to B}
 &\quad\Leftrightarrow\quad
 \sum_{j=1}^k f_j \in \ProcD_{A \to B}.
 \label{eq:process_test_sum}
\end{alignat}
In the particular case of $B = I$, we have
\begin{alignat}{1}
 \{ \cra{e_j} \}_{j=1}^k \in \Meas_A
 &\quad\Leftrightarrow\quad
 \sum_{j=1}^k \cra{e_j} = \cra{\gdis}.
 \label{eq:process_eff_sum}
\end{alignat}

\begin{ex}[quantum theory]
 The sum of processes is equal to the sum of CP maps.
 In particular, the sum of states (or effects) is the sum of matrices.
\end{ex}

\subsubsection{Process spaces are convex cones}

Each process space $\Proc_{A \to B}$
can be regarded as a subset of a real vector space.
Let $\Vec_{A \to B}$ be the set of all formal sums
of the form $\sum_i a_i f_i$ with $a_i \in \Real$ and $f_i \in \Proc_{A \to B}$,
where $\sum_i a_i f_i$ is the element satisfying
$u(\sum_i a_i f_i) = \sum_i a_i u(f_i)$ for any diagram $u:\Proc_{A \to B} \to \Scalar$.
It follows that $\Vec_{A \to B}$ is the vector space
spanned by the process space $\Proc_{A \to B}$.
In this paper, we only consider the case where $\Vec_{A \to B}$ is finite-dimensional
for any systems $A$ and $B$.
Let us define $\Vec_A \coloneqq \Vec_{I \to A}$ and $\Vec_A^* \coloneqq \Vec_{A \to I}$.
We will refer to an element of $\Vec_{A \to B}$ as an \termdef{extended process}
from $A$ to $B$.
Similarly, we will refer to elements of $\Vec_A$ and $\Vec_A^*$
as an \termdef{extended state} of $A$ and an \termdef{extended effect} of $A$, respectively.
Since each extended effect can be described as a linear functional on an extended state,
$\Vec_A^*$ can be regarded as the dual vector space of $\Vec_A$.
We use an overline, $\ol{\phantom{x}}$, to denote extended processes such as $\ol{f}$.
Any extended process $\ol{f} \in \Vec_{A \to B}$,
which can be expressed in the form $\ol{f} = \sum_i a_i f_i$
with $a_i \in \Real$ and $f_i \in \Proc_{A \to B}$,
satisfies $\ol{f} = f_+ - f_-$, where $f_+ \coloneqq \sum_{\{i:a_i > 0\}} a_i f_i \in \Proc_{A \to B}$
and $f_- \coloneqq \sum_{\{i:a_i < 0\}} (-a_i) f_i \in \Proc_{A \to B}$.
As well as processes, extended processes can be composed sequentially and in parallel.
Specifically, for any $\ol{f} \coloneqq \sum_i a_i f_i \in \Vec_{A \to B}$,
$\ol{g} \coloneqq \sum_j b_j g_j \in \Vec_{B \to C}$,
and $\ol{h} \coloneqq \sum_k c_k h_k \in \Vec_{C \to D}$
with $a_i, b_j, c_k \in \Real$, $f_i \in \Proc_{A \to B}$, $g_j \in \Proc_{B \to C}$,
and $h_k \in \Proc_{C \to D}$,
$\ol{g} \circ \ol{f} = \sum_i \sum_j a_i b_j (g_j \circ f_i)$ and
$\ol{f} \otimes \ol{h} = \sum_i \sum_k a_i c_k (f_i \otimes h_k)$ hold.
Diagrammatically, we have
\begin{alignat}{1}
 \includegraphics[scale=1.0]{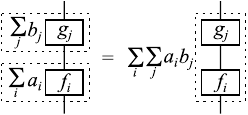}
 \label{eq:process_ex_sequential}
\end{alignat}
and
\begin{alignat}{1}
 \includegraphics[scale=1.0]{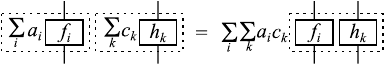} ~\raisebox{.5em}{.}
 \label{eq:process_ex_parallel}
\end{alignat}

Since $a f + b g \in \Proc_{A \to B}$ holds for
any $a, b \in \Scalar = \Real_+$ and $f, g \in \Proc_{A \to B}$,
$\Proc_{A \to B}$ is a convex cone.
Moreover, $\Proc_{A \to B}$ is salient, i.e.,
$\Proc_{A \to B}$ does not contain both $\ol{f}$ and $-\ol{f}$
for any $\ol{f} \neq \emptyset$.

We here recall some basic properties of convex cones.
For a salient convex cone $\mC$ with $\emptyset \in \mC$ in a real vector space $\Vec$,
the partial ordering on $\Vec$ is defined as follows:
for $x, y \in \Vec$,
\begin{alignat}{1}
 x \le y ~(\mbox{or}~ y \ge x) &\quad\logeq\quad y - x \in \mC.
 \label{eq:le}
\end{alignat}
For $x \in \mC$, the set defined as
\begin{alignat}{1}
 \Face_x &\coloneqq \{ y \in \mC :
 \exists \delta \in \Realpp, \delta y \le x \}
 \label{eq:Face}
\end{alignat}
is called the \termdef{face} of $x$,
where $\Realpp$ is the set of all positive real numbers.
$x \in \mC$ is called an \termdef{interior point} if $\Face_x = \mC$ holds.
$x \in \mC$ is called \termdef{atomic} if
$\Face_x = \{ p x : p \in \Real_+ \}$ holds.

Since each process space $\Proc_{A \to B}$ is a convex cone,
the partial ordering on $\Vec_{A \to B}$ and the face of
$f \in \Proc_{A \to B}$ can be, respectively, defined by
Eqs.~\eqref{eq:le} and \eqref{eq:Face} with $\mC = \Proc_{A \to B}$.
One can easily verify that $u(\ol{f}) \le u(\ol{g})$ holds
for any $u:\Proc_{A \to B} \to \Proc_{C \to D}$ and
$\ol{f},\ol{g} \in \Vec_{A \to B}$ satisfying $\ol{f} \le \ol{g}$.
From the definition of faces, we have that, for any $f,g \in \Proc_{A \to B}$,
\begin{alignat}{1}
 f \in \Face_g &\quad\Leftrightarrow\quad
 \exists~\delta \in \Realpp, \; \delta f \le g.
\end{alignat}
We call a state or an effect \termdef{pure} if it is atomic,
\termdef{mixed} if it is not atomic,
and \termdef{completely mixed} if it is an interior point.
$\StP_A$ is defined as the set of all pure states of system $A$.
Let $\StNP_A \coloneqq \StN_A \cap \StP_A$,
which is the set of all normalized pure states of $A$.

Recall that an effect $\cra{e}$ is feasible if and only if
there exists a measurement including $\cra{e}$.
Using Eq.~\eqref{eq:process_eff_sum}, we have
\begin{alignat}{1}
 \cra{e} \in \EffF_A &\quad\Leftrightarrow\quad
 \cra{e} \le \cra{\gdis_A}.
 \label{eq:process_eff_feasible}
\end{alignat}
From Eq.~\eqref{eq:process_eff_sum}, the following relation also holds:
\begin{alignat}{1}
 \cra{e} \in \EffF_A
 &\quad\Leftrightarrow\quad \{ \cra{e}, \cra{\gdis} - \cra{e} \} \in \Meas_A \nonumber \\
 &\quad\Leftrightarrow\quad \cra{\gdis} - \cra{e} \in \EffF_A.
 \label{eq:process_EffF_nas}
\end{alignat}
It follows that $\cra{\gdis}$ is completely mixed.
Indeed, for any $\cra{e} \in \Eff_A$, there exists $\delta \in \Realpp$
such that $\delta \cra{e} \in \EffF_A$,
which implies from Eq.~\eqref{eq:process_eff_feasible} that $\delta \cra{e} \le \cra{\gdis}$.

\begin{lemma} \label{lemma:FaceZero}
 If $\cket{\rho} \in \St_A$ and $\cra{e} \in \Eff_A$ satisfy $\craket{e|\rho} = 0$,
 then $\craket{e|\sigma} = 0$ holds for any $\cket{\sigma} \in \Face_\cket{\rho}$.
\end{lemma}
\begin{proof}
 From the definition of $\Face_\cket{\rho}$, there exists $\delta \in \Realpp$
 such that $\delta \cket{\sigma} \le \cket{\rho}$.
 This gives $0 \le \craket{e|\sigma} \le \delta^{-1} \craket{e|\rho} = 0$.
\end{proof}

\begin{lemma} \label{lemma:FaceOne}
 If $\cket{\rho} \in \St_A$ and $\cra{e} \in \EffF_A$ satisfy
 $\craket{e|\rho} = \craket{\gdis|\rho}$,
 then $\craket{e|\sigma} = \craket{\gdis|\sigma}$ holds for any
 $\cket{\sigma} \in \Face_\cket{\rho}$.
\end{lemma}
\begin{proof}
 From Eq.~\eqref{eq:process_EffF_nas},
 $\cra{e'} \coloneqq \cra{\gdis} - \cra{e} \in \EffF_A$ holds.
 From $\craket{e'|\rho} = 0$ and Lemma~\ref{lemma:FaceZero},
 we have $\craket{e'|\sigma} = 0$.
 Therefore, $\craket{e|\sigma} = \craket{\gdis|\sigma}$ holds.
\end{proof}

It is easily seen that $\StN_A$ is a convex set.
This paper assumes that $\StN_A$ is closed;
in this case, $\St_A$ is a closed convex cone.
From Carath\'{e}odory's theorem,
any $\cket{\rho'} \in \StN_A$ can be expressed in the form
$\cket{\rho'} = \sum_{i=1}^l p_i \cket{\psi_i}$,
where $\cket{\psi_1},\ldots,\cket{\psi_l} \in \StNP_A$,
$p_1,\ldots,p_l \in \Realpp$, and $\sum_{i=1}^l p_i = 1$ hold.
From Eq.~\eqref{eq:rho_normalize}, any $\cket{\rho} \in \St_A$ is proportional to
some normalized state and thus is expressed in the form
$\cket{\rho} = \sum_{i=1}^l c_i \cket{\psi_i}$
with $\cket{\psi_1},\ldots,\cket{\psi_l} \in \StNP_A$ and $c_1,\ldots,c_l \in \Realpp$.

\begin{ex}[quantum theory]
 Recall that $\St_A \cong \bigoplus_{i=1}^{k_A} \mS_+(\Complex^{m_i})$
 and $\St_B \cong \bigoplus_{j=1}^{k_B} \mS_+(\Complex^{n_j})$ hold
 for some natural numbers $k_A,m_1,\ldots,m_{k_A},k_B,n_1,\ldots,n_{k_B}$;
 in this case, $\Vec_{A \to B}$ is isomorphic to the space of all linear maps
 from $\bigoplus_{i=1}^{k_A} \mS(\Complex^{m_i})$ to
 $\bigoplus_{j=1}^{k_B} \mS(\Complex^{n_j})$
 (which are often called Hermitian-preserving maps).
 $\Vec_A \cong \Vec_A^* \cong \bigoplus_{i=1}^{k_A} \mS(\Complex^{m_i})$ also holds.
 For extended processes $\ol{f}, \ol{g} \in \Vec_{A \to B}$,
 $\ol{f} \le \ol{g}$ holds if and only if $\ol{g} - \ol{f}$ is a CP map.
 In particular, for any $\cket{\ol{x}}, \cket{\ol{y}} \in \Vec_A$,
 $\cket{\ol{x}} \le \cket{\ol{y}}$ holds if and only if
 $\cket{\ol{y}} - \cket{\ol{x}}$ is a positive semidefinite matrix.
 The same holds for effects.
 For two states $\cket{\rho},\cket{\sigma} \in \St_A$,
 $\cket{\sigma} \in \Face_\cket{\rho}$ and
 $\supp~\cket{\sigma} \subseteq \supp~\cket{\rho}$ are equivalent,
 where $\supp$ denotes the support of a matrix.
 A state $\cket{\rho}$ is completely mixed if and only if the matrix $\cket{\rho}$ has
 full rank.
 A state $\cket{\psi}$ is pure if and only if $\rank~\cket{\psi} \le 1$ holds.
 \end{ex}

\subsubsection{Perfectly distinguishable states}

A set of states $\Phi \coloneqq \{ \cket{\rho_i} \in \St_A \}_{i=1}^k$ is said to be
\termdef{perfectly distinguishable} if there exists a measurement
$\Pi \coloneqq \{ \cra{e_i} \}_{i=1}^m \in \Meas_A$ $~(m \ge k)$ such that
\begin{alignat}{2}
 \craket{e_i|\rho_i} &= \craket{\gdis|\rho_i}, &\quad \forall i &\in \{1,\ldots,k\},
 \label{eq:distinguishable_scalar}
\end{alignat}
in which case, we say that $\Pi$ \termdef{perfectly distinguishes} between $\Phi$.
Equation~\eqref{eq:distinguishable_scalar} implies
$\craket{e_j|\rho_i} = \delta_{i,j}$ for any $i \in \{1,\ldots,k\}$ and
$j \in \{1,\ldots,m\}$, where $\delta_{i,j}$ is the Kronecker delta.
The \termdef{kernel} of $\cket{\sigma} \in \St_A$, denoted by $\ker_\cket{\sigma}$, is defined as
the set of all states of $A$ that is perfectly distinguishable from $\cket{\sigma}$.
$\cket{\rho} \in \ker_\cket{\sigma}$ and $\cket{\sigma} \in \ker_\cket{\rho}$
are obviously equivalent.

Let us refer to a set of perfectly distinguishable normalized pure states as a \termdef{PDS}%
\footnote{PDS stands for a Perfectly Distinguishable Set of normalized pure states.}.
We also refer to a PDS $\Phi \coloneqq \{ \cket{\phi_i} \in \StNP_A \}_{i=1}^n$
as an \termdef{MPDS} (which stands for a maximal PDS)
if there exists no normalized pure state $\cket{\psi}$ such that
$\{ \cket{\phi_1},\ldots,\cket{\phi_n},\cket{\psi} \}$ are perfectly distinguishable.
Let $\PDS_A$ and $\MPDS_A$ be, respectively, the sets of all PDSs and MPDSs of $A$.
The maximum number of elements of an MPDS for $A$ is called the \termdef{rank} of $A$
and denoted by $\NA$.
Since there exist normalized states for any system $A$, $\NA \ge 1$ holds.
For an MPDS $\Phi$, a measurement $\Pi \in \Meas_A$ with $|\Pi| = |\Phi|$
that perfectly distinguishes between $\Phi$ is called \termdef{maximal},
where $|\mX|$ is the number of elements in a set $\mX$.
An effect $\cra{e} \in \Eff_A$ is called \termdef{maximal}
if there exists a maximal measurement including $\cra{e}$.
Let $\EffM_A$ be the set of all maximal effects of $A$.
$\EffM_A \subseteq \EffF_A \subseteq \Eff_A$ obviously holds.

\begin{ex}[quantum theory]
 A set of states $\{ \cket{\rho_i} \in \St_A \}_{i=1}^k$ is perfectly distinguishable
 if and only if
 $\cket{\rho_i} \cdot \cket{\rho_j} = \cket{\emptyset}$ holds
 for any distinct $i,j \in \{1,\ldots,k\}$,
 where $\cket{\emptyset}$ is the zero square matrix of order $\NA$.
 It follows that $\cket{\rho} \in \ker_\cket{\sigma}$ and
 $\cket{\rho} \cdot \cket{\sigma} = \cket{\emptyset}$ are equivalent.
 $\Psi \coloneqq \{ \cket{\varphi_i} \in \StNP_A \}_{i=1}^k$ is a PDS if and only if
 $k \le \NA$ and $\cket{\varphi_i} = \ket{\varphi_i} \bra{\varphi_i}$
 $~(\forall i \in \{1,\ldots,k\})$ hold
 for some orthonormal basis (ONB) $\{ \ket{\varphi_i} \}_{i=1}^\NA$ of $\Complex^\NA$.
 In particular, $\Psi$ is an MPDS if and only if $k = \NA$ holds.
 For any $\Phi \coloneqq \{ \cket{\phi_i} \}_{i=1}^\NA \in \MPDS_A$,
 there exists a unique maximal measurement $\{ \cra{e_i} \}_{i=1}^\NA$
 that perfectly distinguishes between $\Phi$,
 where $\cra{e_i}$ and $\cket{\phi_i}$ are the same matrix.
 An effect is maximal if and only if
 it is in the form $\ket{\psi}\bra{\psi}$ with some unit vector $\ket{\psi} \in \Complex^\NA$.
\end{ex}

\section{Four postulates} \label{sec:Pustulates}

In this section, we present the postulates used in our derivation.
For each postulate, examples in classical and quantum theories are also presented.
It is noteworthy that all these postulates are satisfied in quantum theory;
conversely, quantum theory is uniquely singled out from these postulates.

\subsection{Postulates} \label{subsec:Postulates}

\subsubsection*{\NumSymSharp. Symmetric sharpness} \label{subsubsec:derive_SymSharp}

Our first postulate concerns the duality between normalized pure states and
maximal effects.
\begin{postulate}[\pagetarget{pos:SymSharp}{Symmetric sharpness}][\NumSymSharp] \label{postulate:SymSharp}
 To every normalized pure state $\cket{\phi}$,
 there corresponds one and only one maximal effect, denoted by $\cra{\phi^\dagger}$,
 giving unit probability (i.e., $\craket{\phi^\dagger|\phi} = 1$).
 Furthermore, for any two normalized pure states $\cket{\varphi}$ and $\cket{\psi}$
 and their corresponding maximal effects $\cra{\varphi^\dagger}$ and $\cket{\psi^\dagger}$,
 respectively,
 the probability of the joint occurrence of $\cket{\varphi}$ and $\cra{\psi^\dagger}$ is
 equal to the probability of the joint occurrence of $\cket{\psi}$ and $\cra{\varphi^\dagger}$
 (i.e., $\craket{\psi^\dagger|\varphi} = \craket{\varphi^\dagger|\psi}$).
\end{postulate}

In any OPT, for any $\cra{e} \in \EffM_A$,
there exists $\cket{\phi} \in \StNP_A$ such that $\craket{e|\phi} = 1$.
The \poslink{pos:SymSharp}{symmetric sharpness postulate} states that
each $\cra{e} \in \EffM_A$ corresponds to one and only one $\cket{\phi} \in \StNP_A$
such that $\craket{e|\phi} = 1$.
This postulate also states that the probability of the joint occurrence of
any normalized pure state and any maximal effect is invariant
under the exchange of the normalized pure state and
its corresponding maximal effect.

\begin{ex}[classical theory]
 As previously mentioned, $\St_A \cong \Real_+^\NA$ holds.
 Regarding the examples of classical theory, without loss of generality,
 we will identify a state of $A$ with the corresponding element of $\Real_+^\NA$,
 which is an $\NA$-dimensional nonnegative column vector.
 Each system $A$ has exactly $\NA$ normalized pure states
 $\cket{1} \coloneqq \ket{1},\ldots,\cket{\NA} \coloneqq \ket{\NA}$,
 where $\{ \ket{n} \}_{n=1}^\NA$ is the standard ONB of $\Real^\NA$.
 Each state $\cket{\rho}$ of $A$ can be expressed in the form
 $\cket{\rho} = \sum_{i=1}^\NA p_i \ket{i}$ with $p_1,\ldots,p_\NA \in \Real_+$.
 We will also identify a process from $A$ to $B$ with its corresponding
 $\NB \times \NA$ nonnegative matrix.
 Particularly, an effect of $A$ can be expressed as an $\NA$-dimensional
 nonnegative row vector.
 $\cra{i^\dagger} \coloneqq \bra{i}$ is the unique maximal effect
 satisfying $\craket{i^\dagger|i} = \braket{i|i} = 1$.
 $\craket{i^\dagger|j} = \delta_{i,j} = \craket{j^\dagger|i}$ also holds.
\end{ex}

\begin{ex}[quantum theory]
 For any $\cket{\phi} \in \StNP_A$,
 which can be expressed in the form $\cket{\phi} = \ket{\phi} \bra{\phi}$ with some
 unit vector $\ket{\phi} \in \Complex^\NA$,
 $\cra{\phi^\dagger} \coloneqq \ket{\phi} \bra{\phi}$ represents the unique maximal effect
 that satisfies $\craket{\phi^\dagger|\phi} = |\braket{\phi|\phi}|^2 = 1$.
 $\craket{\psi^\dagger|\varphi} = |\braket{\psi|\varphi}|^2 = \craket{\varphi^\dagger|\psi}$
 also holds for any $\cket{\psi} = \ket{\psi} \bra{\psi} \in \StNP_A$
 and $\cket{\varphi} = \ket{\varphi} \bra{\varphi} \in \StNP_A$.
\end{ex}

\subsubsection*{\NumCompletelyMixed. Complete mixing} \label{subsubsec:derive_CompletelyMixed}

The second postulate provides a sufficient condition for a state to be completely mixed.
\begin{postulate}[\pagetarget{pos:CompletelyMixed}{Complete mixing}][\NumCompletelyMixed]
 \label{postulate:CompletelyMixed} 
 Every state for which there exists no maximal effect giving zero probability
 is completely mixed.
\end{postulate}

In any OPT, the converse of this postulate is also true,
i.e., for every completely mixed state $\cket{\rho}$, there exists no maximal effect
giving zero probability.
Indeed, for any $\cra{e} \in \EffM_A$,
there exists $\cket{\psi} \in \StNP_A$ satisfying $\craket{e|\psi} = 1$.
From $\cket{\psi} \in \St_A = \Face_\cket{\rho}$,
$\delta \cket{\psi} \le \cket{\rho}$ holds for some $\delta \in \Realpp$,
which gives $\craket{e|\rho} \ge \delta \craket{e|\psi} = \delta > 0$.

\begin{ex}[classical theory]
 Every $\cket{\rho} \in \St_A$ is in the form $\cket{\rho} = \sum_{i=1}^\NA p_i \ket{i}$
 with $p_1,\ldots,p_\NA \in \Real_+$.
 If $\craket{j^\dagger|\rho} = p_j > 0$ holds for each $\cra{j^\dagger} = \bra{j} \in \EffM_A$,
 then $\cket{\rho}$ is completely mixed.
\end{ex}

\begin{ex}[quantum theory]
 The arguments for classical theory can readily be extended to quantum theory.
 Every $\cket{\rho} \in \St_A$ can be expressed in the form
 $\cket{\rho} = \sum_{i=1}^\NA p_i \ket{\phi_i}\bra{\phi_i}$
 with $p_1,\ldots,p_\NA \in \Real_+$ and an ONB $\{ \ket{\phi_i} \}_{i=1}^\NA$ of $\Complex^\NA$.
 If $\craket{\phi_j^\dagger|\rho} = p_j > 0$ holds
 for each $\cra{\phi_j^\dagger} = \ket{\phi_j}\bra{\phi_j} \in \EffM_A$,
 then $\cket{\rho}$ is completely mixed.
\end{ex}

\subsubsection*{\NumFilter. Filtering} \label{subsubsec:derive_Filter}

The third postulate entails the existence of what we refer to as filters.
Herein, we will first define filters before stating the postulate.
For any normalized pure state $\cket{\phi} \in \StNP_A$ and any $c \in [0,1]$,
a \termdef{filter} $F_\cket{\phi}^c$ is defined as a process in $\ProcF_{A \to A}$
that satisfies
\begin{alignat}{1}
 \includegraphics[scale=1.0]{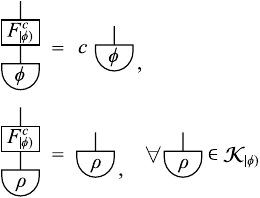} ~\raisebox{1.6em}{;}
 \label{eq:derive_filter}
\end{alignat}
that is, $F_\cket{\phi}^c$ transforms $\cket{\phi}$ to $c \cket{\phi}$
and leaves any state that is perfectly distinguishable from $\cket{\phi}$.
Furthermore, suppose that $F_\cket{\phi}^c$ is reversible if $c > 0$ holds.
We say that $\cket{\phi} \in \StNP_A$ \termdef{can be arbitrarily filtered} if
there exists a filter $F_\cket{\phi}^c$ for every $c \in [0,1]$.

\begin{postulate}[\pagetarget{pos:Filter}{Filtering}][\NumFilter] \label{postulate:Filter}
 Any normalized pure state can be arbitrarily filtered.
\end{postulate}

\begin{ex}[classical theory]
 For any $\cket{j} \in \StNP_A$ and $c \in [0,1]$,
 the filter $F_\cket{j}^c \in \ProcF_{A \to A}$ can be expressed as the following matrix:
 \begin{alignat}{2}
  F_\cket{j}^c &= c \ket{j}\bra{j} + \Upsilon_j,
  \label{eq:FilterClassical}
 \end{alignat}
 where $\Upsilon_j \coloneqq \ident_\NA - \ket{j}\bra{j}$.
 For any $\cket{\sigma} \in \St_A$, which is in the form $\cket{\sigma} = \sum_{i=1}^\NA p_i \ket{i}$,
 we have $F_\cket{j}^c \circ \cket{\sigma} = \sum_{i=1}^\NA p_i \ket{i} - (1 - c) p_j \ket{j}$.
 Since $\{ F_\cket{j}^c, \ident_\NA - F_\cket{j}^c \}$ is a test, $F_\cket{j}^c$ is feasible.
 It is easy to verify that $F_\cket{j}^c$ satisfies Eq.~\eqref{eq:derive_filter}.
 Indeed, $F_\cket{j}^c \circ \cket{j} = c \cket{j}$ obviously holds.
 Also, for any $\cket{\rho} \in \ker_\cket{j}$,
 since $\cket{\rho}$ can be expressed in the form
 $\cket{\rho} = \sum_{i=1}^\NA q_i \ket{i}$ with $q_j = 0$,
 we have $F_\cket{j}^c \circ \cket{\rho} = \cket{\rho}$.
 Let $\tF_\cket{\phi}^c \coloneqq \ket{j}\bra{j} + c \Upsilon_j$; then,
 $\tF_\cket{\phi}^c \circ F_\cket{\phi}^c = c \cdot \id_A
 = F_\cket{\phi}^c \circ \tF_\cket{\phi}^c$ obviously holds
 (note that $\id_A = \ident_\NA$ holds).
 Thus, for any $c > 0$, $F_\cket{\phi}^c$ is reversible and its inverse is
 $c^{-1} \tF_\cket{\phi}^c$.
\end{ex}

\begin{ex}[quantum theory]
 For any $\cket{\phi} = \ket{\phi}\bra{\phi} \in \StNP_A$ and $c \in [0,1]$,
 the CP map $F_\cket{\phi}^c$ defined as
 \begin{alignat}{1}
  F_\cket{\phi}^c \circ \cket{\rho} &= E_\cket{\phi}^c \cdot \cket{\rho} \cdot [E_\cket{\phi}^c]^\dagger,
  \quad \cket{\rho} \in \St_A, \nonumber \\
  E_\cket{\phi}^c &\coloneqq \sqrt{c} \ket{\phi} \bra{\phi}
  + \Upsilon'_{\ket{\phi}}
  \label{eq:FilterQuantum}
 \end{alignat}
 is a filter, where $\Upsilon'_{\ket{\phi}} \coloneqq \ident_\NA - \ket{\phi} \bra{\phi}$.
 Since $[E_\cket{\phi}^c]^\dagger E_\cket{\phi}^c \le \ident_\NA$ holds,
 $F_\cket{\phi}^c$ is trace non-increasing and thus feasible.
 Again, Eq.~\eqref{eq:derive_filter} can easily be verified.
 Also, let $\tF_\cket{\phi}^c$ be the following CP map:
 \begin{alignat}{1}
  \tF_\cket{\phi}^c \circ \cket{\rho} &= \tE_\cket{\phi}^c \cdot \cket{\rho} \cdot [\tE_\cket{\phi}^c]^\dagger,
  \quad \cket{\rho} \in \St_A, \nonumber \\
  \tE_\cket{\phi}^c &\coloneqq \ket{\phi} \bra{\phi} + \sqrt{c} \Upsilon'_{\ket{\phi}};
  \label{eq:FilterQuantumInv}
 \end{alignat}
 then, $\tF_\cket{\phi}^c \circ F_\cket{\phi}^c = c \cdot \id_A
 = F_\cket{\phi}^c \circ \tF_\cket{\phi}^c$ holds.
 Thus, for any $c > 0$, $F_\cket{\phi}^c$ is reversible
 and its inverse is $c^{-1} \tF_\cket{\phi}^c$.
\end{ex}

\subsubsection*{\NumProcEq. Local equality} \label{subsubsec:derive_ProcEq}

The fourth postulate requires that two processes are equal
if they are locally equal.
\begin{postulate}[\pagetarget{pos:ProcEq}{Local equality}][\NumProcEq]
 \label{postulate:ProcEq} 
 Any two locally equal processes are equal.
\end{postulate}
Recall that, in any OPT, the converse is true,
i.e., any two equal processes are locally equal.
Thus, this postulate implies that, for any $f, f' \in \Proc_{A \to B}$,
$f \eqlocal f'$ and $f = f'$ are equivalent.

We here consider the following scenario.
Suppose that two processes $f,f' \in \Proc_{A \to B}$ are equal and
we want to prove it, where $A$ and $B$ refer to systems in our laboratory.
From the definition of equality of processes,
we need to show that two scalars
$\cra{e} \circ (f \circ \id_E) \circ \cket{\rho}$ and
$\cra{e} \circ (f' \circ \id_E) \circ \cket{\rho}$ are equal
for any $E \in \System$, $\cket{\rho} \in \St_{A \otimes E}$,
and $\cra{e} \in \Eff_{B \otimes E}$.
If the \poslink{pos:ProcEq}{local equality postulate} does not hold, then
since $E$ could be an extremely large system (such as the universe),
it may be practically impossible to prove $f = f'$.
If the \poslink{pos:ProcEq}{local equality postulate} holds, then we only need to show
that $\cra{e} \circ f \circ \cket{\rho} = \cra{e} \circ f' \circ \cket{\rho}$
holds for any $\cket{\rho} \in \St_A$ and $\cra{e} \in \Eff_B$,
which can be accomplished only in the laboratory.
Thus, we can say that this postulate significantly reduces the amount of information required to
identify a process.

\begin{ex}[classical theory]
 Any process $f \in \Proc_{A \to B}$ can be expressed in the form
 $f = \sum_{i=1}^\NB \sum_{j=1}^\NA f_{i,j} \ket{i_B} \bra{j_A}$
 with the standard ONBs $\{ \cket{i_B} \coloneqq \ket{i_B} \}_{i=1}^\NB$ of $\Real^\NB$
 and $\{ \cket{j_A} \coloneqq \ket{j_A} \}_{j=1}^\NA$ of $\Real^\NA$,
 where $f_{i,j} \coloneqq \braket{i_B|f|j_A} =
 \cra{i_B^\dagger} \circ f \circ \cket{j_A} \in \Real_+$.
 Hence, if two processes $f,f' \in \Proc_{A \to B}$ satisfy
 $\cra{i_B^\dagger} \circ f \circ \cket{j_A} = \cra{i_B^\dagger} \circ f' \circ \cket{j_A}$
 for any $i \in \{1,\ldots,\NB\}$ and $j \in \{1,\ldots,\NA\}$,
 then $f = f'$ obviously holds.
\end{ex}

\begin{ex}[quantum theory]
 As we have already shown in Subsubsec.~\ref{subsubsec:OPT_OPT_eq},
 any quantum theory enjoys the \poslink{pos:ProcEq}{local equality postulate}.
\end{ex}

The above four postulates can be thought of as approaches describing
the behavior of feasible processes.
(Recall that any process is proportional to some feasible process.)
Indeed, it can certainly be stated that
the \poslink{pos:CompletelyMixed}{complete mixing postulate} is equivalent to
the statement that every feasible state for which there exists no maximal effect
giving zero probability is completely mixed,
and that the \poslink{pos:ProcEq}{local equality postulate} is equivalent to
the statement that any two feasible and locally equal processes are equal.
Regarding the \poslink{pos:Filter}{filtering postulate},
we note that a filter $F_\cket{\phi}^c \in \ProcF_{A \to A}$ is reversible
if and only if there exist a feasible process $\tF_\cket{\phi}^c \in \ProcF_{A \to A}$
and a nonzero feasible scalar $p$ such that
$\tF_\cket{\phi}^c \circ F_\cket{\phi}^c \eqlocal p \cdot \id_A
\eqlocal F_\cket{\phi}^c \circ \tF_\cket{\phi}^c$.

It is noteworthy that, roughly speaking, each of our four postulates is essentially the same as,
or similar to, that used in previous studies.
Specifically, in studies that focus on reconstructing quantum theory from operational postulates,
postulates on symmetric sharpness \cite{Sel-Sca-Coe-2018}
(or sharpness \cite{Har-2011,Wil-2012-axioms,Wil-2018}),
complete mixing \cite{Chi-Dar-Per-2011,Zao-2012},
and filtering \cite{Wil-2012-axioms,Wil-2012,Wil-2018} have been used.
Also, several studies have used the so-called local tomography
(or local discriminability) postulate
\cite{Har-2001,Dak-Bru-2009,Mas-Mul-2011,Chi-Dar-Per-2011,Har-2011,Zao-2012,Sel-Sca-Coe-2018},
which is highly related to the \poslink{pos:ProcEq}{local equality postulate}.
Indeed, it has been shown that the \poslink{pos:ProcEq}{local equality postulate} holds
for any OPT having the local tomography postulate \cite{Chi-Dar-Per-2010}.
We will show later that, conversely, the local tomography postulate holds
for any OPT having the \poslink{pos:ProcEq}{local equality postulate},
which means that these two postulates are equivalent.
Although each of the four postulates is not new in itself, combining them
provides a new way to reconstruct quantum theory.

\subsection{Overview of derivation} \label{subsec:derive_abstract_stream}

In this paper, we show that an OPT having the four proposed postulates
satisfies the following properties:
\begin{enumerate}[label=(\Alph*)]
 \item \pagetarget{property:A}{}
       For each system $A$, the state space $\St_A$ is isomorphic
       to a direct sum of spaces of complex positive semidefinite matrices,
       i.e., $\St_A \cong \bigoplus_{i=1}^k \mS_+(\Complex^{n_i})$,
       where $k, n_1, \ldots, n_k$ refer to some natural numbers that satisfy
       $\sum_{i=1}^k n_i = \NA$ (Theorem~\ref{thm:Complex}).
 \item \pagetarget{property:B}{}
       For any systems $A$ and $B$ with
       $\St_A \cong \bigoplus_{i=1}^{k_A} \mS_+(\Complex^{m_i})$
       and $\St_B \cong \bigoplus_{j=1}^{k_B} \mS_+(\Complex^{n_j})$,
       $\St_{A \otimes B} \cong \bigoplus_{i=1}^{k_A}
       \bigoplus_{j=1}^{k_B} \mS_+(\Complex^{m_in_j})$ holds
       (Theorem~\ref{thm:ComplexOtimes}).
 \item \pagetarget{property:C}{}
       Let $\StMat_A^+ \coloneqq \{ \Mat^A_\cket{\sigma} : \cket{\sigma} \in \St_A \}
       \cong \St_A$,
       where $\Mat^A_\cket{\sigma}$ refers to a certain full and faithful matrix
       representation of $\cket{\sigma}$.
       Also, let $\StMat_A$ be the real vector space spanned by $\StMat_A^+$
       and $\mathbf{CP}_{A \to B}$ be the set of all CP maps from $\StMat_A$ to $\StMat_B$.
       Then, for any systems $A$ and $B$, each process space $\Proc_{A \to B}$ is isomorphic to
       $\mathbf{CP}_{A \to B}$ as convex cones.
       Also, there exists an isomorphism
       $\mL:\Proc_{A \to B} \ni f \mapsto \mL_f \in \mathbf{CP}_{A \to B}$ such that
       (i) $\Mat^B_{f \circ \cket{\rho}} = \mL_f[\Mat^A_\cket{\rho}]$
       for any $f \in \Proc_{A \to B}$ and $\cket{\rho} \in \St_A$;
       (ii) $f \in \Proc_{A \to B}$ is deterministic if and only if $\mL_f$ is TP;
       and (iii) $f \in \Proc_{A \to B}$ is feasible if and only if $\mL_f$
       is trace non-increasing (Theorems~\ref{thm:Proc} and \ref{thm:ProcFeasible}).
\end{enumerate}
It is evident that quantum theory is uniquely singled out from the above three properties.
We present a brief comment on these properties.
The key step in the derivation of quantum theory is to prove Property~(\hyperlink{property:A}{A}),
which characterizes the state space for each system.
Properties~(\hyperlink{property:B}{B}) and (\hyperlink{property:C}{C}) can be easily derived
from our postulates and Property~(\hyperlink{property:A}{A}).
Property~(\hyperlink{property:B}{B}) represents the relation between the state space of
a composite system and those of its subsystems.
Property~(\hyperlink{property:C}{C}) implies that each process is fully and faithfully represented
by a CP map.
This property also characterizes deterministic processes and feasible processes.

We now present an overview of our approach to the derivation of quantum theory.
The first half of our approach focuses on analyzing the behavior of an individual system
using the first three postulates presented in the previous subsection
and shows that the following property holds:
\begin{enumerate}[label=($\diamondsuit$)]
 \item \pagetarget{property:diamondsuit}{}
       For each system $A$, $\St_A$ is a symmetric cone and $\Eff_A$ is its dual cone.
\end{enumerate}
We will give the definition of a symmetric cone and a dual cone
in Subsec.~\ref{subsec:derive_abstract_symmetric}.
It is a well-known fact that a symmetric cone can be decomposed as a direct sum of five types of
irreducible symmetric cones, one of which is
the set of complex positive semidefinite matrices, $\mS_+(\Complex^n)$,
of a certain order $n$ \cite{Jor-Neu-Wig-1934} (see also Theorem~\ref{thm:Neumann}).
The second half of our approach assumes that
Property~(\hyperlink{property:diamondsuit}{$\diamondsuit$}) and
the \poslink{pos:ProcEq}{local equality postulate} hold.
By analyzing the behavior of a composite system, we can conclude that $\St_A$ is isomorphic to
a direct sum of spaces of complex positive semidefinite matrices,
which results in Property~(\hyperlink{property:A}{A}).
Properties~(\hyperlink{property:B}{B}) and (\hyperlink{property:C}{C}) can also be derived.
Our approach consists of the following three steps
(where the first half is further divided into two steps):
\begin{enumerate}[label=Step~\arabic*):,leftmargin=4em]
 \item Derive some basic properties of PDSs (Sec.~\ref{sec:derive_PDS}).
 \item Derive Property~(\hyperlink{property:diamondsuit}{$\diamondsuit$})
       (Sec.~\ref{sec:derive_symmetry}).
 \item Derive Properties~(\hyperlink{property:A}{A})--(\hyperlink{property:C}{C}),
       i.e., single out quantum theory (Sec.~\ref{sec:derive_Quantum}).
\end{enumerate}
A schematic flow chart of our derivation of quantum theory is shown in Fig.~\ref{fig:flow}.
\begin{figure}[ht]
 \centering
 \includegraphics[scale=0.8]{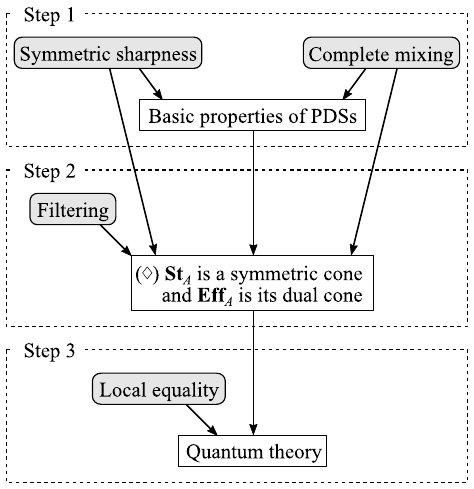}
 \caption{A flow chart of our derivation of quantum theory.}
 \label{fig:flow}
\end{figure}

We will briefly explain each step.

\subsubsection*{Step~1}

Firstly, we consider an OPT with
the \poslink{pos:SymSharp}{symmetric sharpness} and
\poslink{pos:CompletelyMixed}{complete mixing} postulates
and derive some basic properties of PDSs.
For instance, we show that, to every MPDS, there corresponds one and only one maximal measurement
and that each PDS has certain symmetries.

\subsubsection*{Step~2}

Secondly, we use the \poslink{pos:Filter}{filtering postulate}
in addition to the first two postulates.
In this step, we first show that every state has a spectral decomposition.
Using this result, we then show that every state space is a symmetric cone
and that the effect space of $A$ is the dual cone of the state space of $A$.

It has been shown that any symmetric cone is the cone of squares of a certain
EJA \cite{Koe-1957,Vin-1960}.
Since the state space $\St_A$ is a symmetric cone,
$\St_A$ is the cone of squares of some EJA $\EJA_A$.
We will give the definition of EJAs and describe their basic properties in Sec.~\ref{sec:derive_EJA}.

\subsubsection*{Step~3}

Lastly, we derive Properties~(\hyperlink{property:A}{A})--(\hyperlink{property:C}{C}) from
Property~(\hyperlink{property:diamondsuit}{$\diamondsuit$}) and
the \poslink{pos:ProcEq}{local equality postulate}.
This implies that quantum theory can be derived from the four postulates.
As will be presented later,
a necessary and sufficient condition that a state space $\St_A$ is isomorphic to
$\bigoplus_{i=1}^k \mS_+(\Complex^{n_i})$ (i.e., Property~(\hyperlink{property:A}{A}) holds)
is that the corresponding EJA $\EJA_A$ is isomorphic to
$\bigoplus_{i=1}^k \mS(\Complex^{n_i})$.
To derive Property~(\hyperlink{property:A}{A}),
we show $\EJA_A \cong \bigoplus_{i=1}^k \mS(\Complex^{n_i})$
using the correspondence between the dimensions and the ranks of EJAs.
Subsequently, we derive Properties~(\hyperlink{property:B}{B}) and (\hyperlink{property:C}{C}).

The results in this step indicate that
an OPT satisfying Property~(\hyperlink{property:diamondsuit}{$\diamondsuit$})
and the \poslink{pos:ProcEq}{local equality postulate} is quantum theory.
Also, as we already mentioned in Subsec.~\ref{subsec:Postulates},
quantum theory satisfies the four postulates that we propose.
Thus, one can easily see that, in an OPT $\mT$, the following statements are all equivalent:
\begin{enumerate}[label=(\alph*)]
 \item $\mT$ satisfies the four postulates described in Subsec.~\ref{subsec:Postulates}.
 \item $\mT$ satisfies Property~(\hyperlink{property:diamondsuit}{$\diamondsuit$}) and the
       \poslink{pos:ProcEq}{local equality postulate}.
 \item $\mT$ is quantum theory.
\end{enumerate}

\section{Basic properties of a PDS} \label{sec:derive_PDS}

In this section, we present an OPT satisfying
the \poslink{pos:SymSharp}{symmetric sharpness}
and \poslink{pos:CompletelyMixed}{complete mixing} postulates
and derive some basic properties of PDSs.
We list here the main results of our study in this section:
\begin{enumerate}
 \item To every MPDS, there corresponds one and only one maximal measurement
       (Lemma~\ref{lemma:MPDSMeas}).
 \item Every MPDS of system $A$ has exactly $\NA$ elements (Lemma~\ref{lemma:Rank}).
       Also, the sum of all elements is the same for every MPDS of $A$
       (Lemma~\ref{lemma:Chi}).
 \item For every PDS $\Phi \coloneqq \{ \cket{\phi_i} \}_{i=1}^k$,
       the kernel of $\cket{\chi_\Phi} \coloneqq \sum_{i=1}^k \cket{\phi_i}$
       is equal to the set of all states $\cket{\rho}$ satisfying
       $\craket{\chi_\Phi^\dagger|\rho} = 0$,
       where $\cra{\chi_\Phi^\dagger} \coloneqq \sum_{i=1}^k \cra{\phi_i^\dagger}$
       (Lemma~\ref{lemma:PerpPDS}).
       Also, for every PDS $\Phi$, the face of $\cket{\chi_\Phi}$
       is equal to the set of all states $\cket{\rho}$ satisfying
       $\craket{\chi_\Phi^\dagger|\rho} = \craket{\gdis|\rho}$
       (Lemma~\ref{lemma:FacePhi}).
\end{enumerate}

\subsection{Results about symmetric sharpness}

From the \poslink{pos:SymSharp}{symmetric sharpness postulate},
each maximal effect corresponds to one and only one normalized pure state.
Let us define $\dagger$ as the map $\StNP_A \ni \cket{\phi} \mapsto
\cket{\phi}^\dagger \coloneqq \cra{\phi^\dagger} \in \EffM_A$,
where $\cra{\phi^\dagger}$ satisfies $\craket{\phi^\dagger|\phi} = 1$.
By a slight abuse of notation, we will use the same symbol $\dagger$
for the inverse map $\EffM_A \ni \cra{e} \mapsto
\cra{e}^\dagger \coloneqq \cket{e^\dagger} \in \StNP_A$,
where $\cket{e^\dagger}$ satisfies $\craket{e|e^\dagger} = 1$.
We can represent as
\begin{alignat}{1}
 \includegraphics[scale=1.0]{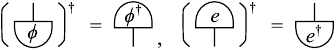} ~\raisebox{.2em}{.}
\end{alignat}

The following lemma ensures that, to every MPDS, there corresponds one and only one
maximal measurement.
\begin{lemma} \label{lemma:MPDSMeas}
 For any MPDS $\Phi \coloneqq \{ \cket{\phi_i} \}_{i=1}^n$,
 there exists a unique maximal measurement
 $\Pi \coloneqq \{ \cra{\phi_i^\dagger} \}_{i=1}^n$ that perfectly distinguishes between $\Phi$.
 Furthermore, $\Phi$ is the unique MPDS that is perfectly distinguished by $\Pi$.
\end{lemma}
\begin{proof}
 There exists a maximal measurement, denoted by $\Pi' \coloneqq \{ \cra{e_i} \}_{i=1}^n$,
 that perfectly distinguishes between $\Phi$.
 From the \poslink{pos:SymSharp}{symmetric sharpness postulate},
 $\cra{\phi_i^\dagger}$ is the unique maximal effect that satisfies
 $\craket{\phi_i^\dagger|\phi_i} = 1$,
 which indicates $\cra{e_i} = \cra{\phi_i^\dagger}$.
 Thus, we have $\Pi' = \Pi$.
 Moreover, since $\cket{\phi_i}$ is the unique normalized pure state that satisfies
 $\craket{\phi_i^\dagger|\phi_i} = 1$,
 $\Phi$ is obviously the unique MPDS that is perfectly distinguished by $\Pi$.
\end{proof}

\begin{lemma} \label{lemma:Perp}
 We have that, for $\cket{\phi} \in \StNP_A$ and $\cket{\rho} \in \St_A$,
 \begin{alignat}{1}
  \cket{\rho} \in \ker_\cket{\phi} &\quad\Leftrightarrow\quad \craket{\phi^\dagger|\rho} = 0.
 \end{alignat}
\end{lemma}
\begin{proof}
 The case $\cket{\rho} = \cket{\emptyset}$ is obvious;
 suppose $\cket{\rho} \neq \cket{\emptyset}$.

 $\Rightarrow$:
 Let $\{ \cra{e_\phi}, \cra{e_\rho} \}$ be a measurement that perfectly distinguishes
 between $\{ \cket{\phi}, \cket{\rho} \}$.
 Arbitrarily choose $\cket{\varphi} \in \StNP_A \cap \Face_{\cket{\rho}}$.
 From $\craket{e_\rho|\rho} = \craket{\gdis|\rho}$ and Lemma~\ref{lemma:FaceOne},
 we have $\craket{e_\rho|\varphi} = \craket{\gdis|\varphi} = 1$.
 Hence, the set of the pure states $\{ \cket{\phi}, \cket{\varphi} \}$ is
 perfectly distinguished by the measurement $\{ \cra{e_\phi}, \cra{e_\rho} \}$ and thus is a PDS.
 (Note that this implies $\NA \ge 2$.)
 Consider an MPDS $\{ \cket{\phi}, \cket{\varphi}, \cket{\psi_1},\ldots,\cket{\psi_k} \}$
 that includes $\cket{\phi}$ and $\cket{\varphi}$, where $k$ is a nonnegative integer.
 From Lemma~\ref{lemma:MPDSMeas},
 $\{ \cra{\phi^\dagger}, \cra{\varphi^\dagger}, \cra{\psi_1^\dagger}, \ldots, \cra{\psi_k^\dagger} \}$
 is a measurement that perfectly distinguishes between this MPDS,
 which indicates $\craket{\phi^\dagger|\varphi} = 0$.
 Since $\craket{\phi^\dagger|\varphi} = 0$ holds for any normalized pure state
 $\cket{\varphi} \in \Face_{\cket{\rho}}$,
 we have $\craket{\phi^\dagger|\rho} = 0$.

 $\Leftarrow$:
 Let $\Pi \coloneqq \{ \cra{\phi^\dagger}, \cra{\gdis} - \cra{\phi^\dagger} \}$.
 Since $\cra{\phi^\dagger}$ is feasible, $\Pi$ is a measurement (see Eq.~\eqref{eq:process_EffF_nas}).
 From $\craket{\phi^\dagger|\phi} = 1$ and
 $[\cra{\gdis} - \cra{\phi^\dagger}] \cket{\rho} = \craket{\gdis|\rho}$,
 $\cket{\phi}$ and $\cket{\rho}$ are perfectly distinguished by $\Pi$.
\end{proof}


\begin{lemma} \label{lemma:SharpMix}
 Consider $\{ c_i \in \Real_+ \}_{i=1}^l$, $\{ d_j \in \Real_+ \}_{j=1}^m$,
 $\{ \cket{\psi_i} \in \StNP_A \}_{i=1}^l$, and $\{ \cket{\phi_j} \in \StNP_A \}_{j=1}^m$.
 Then, we have
 \begin{alignat}{1}
  \includegraphics[scale=1.0]{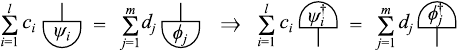} ~\raisebox{.5em}{.}
  \nonumber \\
  \label{eq:derive_sharpmix}
 \end{alignat}
\end{lemma}
\begin{proof}
 From the \poslink{pos:SymSharp}{symmetric sharpness postulate},
 we have that, for any $\cket{\varphi} \in \StNP_A$,
 \begin{alignat}{1}
  \includegraphics[scale=1.0]{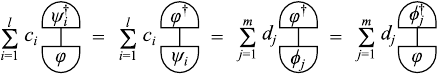} ~\raisebox{.5em}{.}
  \nonumber \\
  \label{eq:derive_sharpmix_proof}
 \end{alignat}
 Since any state can be written as a weighted sum of
 normalized pure states,
 Eq.~\eqref{eq:derive_sharpmix_proof} gives
 $\sum_{i=1}^l c_i \craket{\psi_i^\dagger|\rho} = \sum_{j=1}^m d_j \craket{\phi_j^\dagger|\rho}$
 for any $\cket{\rho} \in \St_A$.
 Therefore, $\sum_{i=1}^l c_i \cra{\psi_i^\dagger} = \sum_{j=1}^m d_j \cra{\phi_j^\dagger}$
 holds.
\end{proof}

For every $\cket{\rho} \in \St_A$, which can be expressed in the form
$\cket{\rho} = \sum_{i=1}^l c_i \cket{\psi_i}$ with $c_1,\ldots,c_l \in \Realpp$
and $\cket{\psi_1},\ldots,\cket{\psi_l} \in \StNP_A$, $\cra{\rho^\dagger} \in \Eff_A$ is defined as
\begin{alignat}{1}
 \includegraphics[scale=1.0]{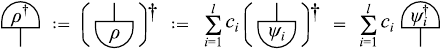} ~\raisebox{.5em}{.}
 \label{eq:derive_dagger_rho}
\end{alignat}
$\cra{\emptyset^\dagger} = \cra{\emptyset}$ obviously holds.
Lemma~\ref{lemma:SharpMix} guarantees that $\cra{\rho^\dagger}$ is
uniquely determined, regardless of how $\cket{\rho}$ is decomposed.

Lemma~\ref{lemma:SharpMix} can be immediately generalized to extended states.
Indeed, it is clear that Eq.~\eqref{eq:derive_sharpmix} holds
even if $c_i$ and $d_j$ are any real numbers.
Thus, for any extended state $\cket{\ol{v}} \in \Vec_A$, which is expressed in the form
$\cket{\ol{v}} = \sum_{i=1}^l c_i \cket{\psi_i}$ with $c_1,\ldots,c_l \in \Real$
and $\cket{\psi_1},\ldots,\cket{\psi_l} \in \StNP_A$,
we can define
$\cra{\ol{v}^\dagger} \coloneqq \cket{\ol{v}}^\dagger \coloneqq
\sum_{i=1}^l c_i \cra{\psi_i^\dagger}$.
Clearly, $\cra{\ol{v}^\dagger}$ is uniquely determined,
regardless of how $\cket{\ol{v}}$ is decomposed.
One can easily verify that the map $\dagger:\Vec_A \ni \cket{\ol{v}} \mapsto
\cket{\ol{v}}^\dagger \in \Vec_A^*$ is linear.
We should note that the converse of Eq.~\eqref{eq:derive_sharpmix} does not
necessarily hold.

For every PDS $\Phi \coloneqq \{ \cket{\phi_i} \}_{i=1}^k \in \PDS_A$,
$\cket{\chi_\Phi}$ is defined as
\begin{alignat}{1}
 \includegraphics[scale=1.0]{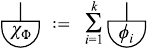} ~\raisebox{.5em}{.}
 \label{eq:derive_chi}
\end{alignat}
Since $\craket{\gdis|\chi_\Phi} = k$ holds, $\cket{\chi_\Phi}$ is not feasible
if $k$ is larger than 1.

\begin{lemma} \label{lemma:ChiOne}
 For any $\Phi \in \MPDS_A$ and $\cra{e} \in \EffM_A$, we have
 \begin{alignat}{1}
  \includegraphics[scale=1.0]{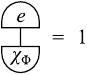} ~\raisebox{1.0em}{.}
 \end{alignat}
\end{lemma}
\begin{proof}
 Let $\Phi \coloneqq \{ \cket{\phi_i} \}_{i=1}^n$;
 then, we have
 \begin{alignat}{1}
  \includegraphics[scale=1.0]{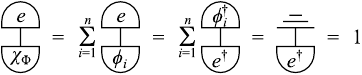} ~\raisebox{1.0em}{.}
 \end{alignat}
 The second equality follows from the \poslink{pos:SymSharp}{symmetric sharpness postulate}.
 The third equality follows from $\sum_{i=1}^n \cra{\phi_i^\dagger} = \cra{\gdis}$,
 which is obtained by the fact that, from Lemma~\ref{lemma:MPDSMeas},
 $\{ \cra{\phi_i^\dagger} \}_{i=1}^n$ is a maximal measurement.
\end{proof}

\begin{lemma} \label{lemma:Rank}
 Every MPDS of system $A$ has exactly $\NA$ elements.
\end{lemma}
\begin{proof}
 There exists an MPDS $\Phi \coloneqq \{ \cket{\phi_i} \}_{i=1}^\NA \in \MPDS_A$
 with $|\Phi| = \NA$.
 From Lemma~\ref{lemma:ChiOne},
 for any $\{ \cket{\psi_i} \}_{i=1}^n \in \MPDS_A$, we have
 \begin{alignat}{1}
  \includegraphics[scale=1.0]{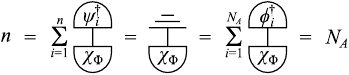} ~\raisebox{1.0em}{.}
 \end{alignat}
\end{proof}


\begin{lemma} \label{lemma:PDSChi}
 For any $\Phi \in \PDS_A$, both $\cra{\chi_\Phi^\dagger}$ and
 $\cra{\gdis_A} - \cra{\chi_\Phi^\dagger}$ are feasible effects.
\end{lemma}
\begin{proof}
 Let $\Phi \coloneqq \{ \cket{\phi_i} \}_{i=1}^k$.
 Consider an MPDS $\Phi_\mathrm{ex} \coloneqq \{ \cket{\phi_i} \}_{i=1}^\NA \supseteq \Phi$.
 Since $\{ \cra{\phi_i^\dagger} \}_{i=1}^\NA$ is a measurement from Lemma~\ref{lemma:MPDSMeas},
 $\cra{\chi_\Phi^\dagger}$ is a feasible effect.
 Also, from Eq.~\eqref{eq:process_EffF_nas}, $\cra{\gdis_A} - \cra{\chi_\Phi^\dagger}$ is
 a feasible effect.
\end{proof}

The face of a PDS $\Phi \in \PDS_A$ is defined as
$\Face_\Phi \coloneqq \Face_\cket{\chi_\Phi}$.
Also, the kernel of $\Phi$ is defined as
$\ker_\Phi \coloneqq \ker_\cket{\chi_\Phi}$.
They can also be expressed by
\begin{alignat}{2}
 \Face_\Phi &= \{ \cket{\rho} \in \St_A : \exists \delta \in \Realpp,
 \delta \cket{\rho} \le \cket{\chi_\Phi} \}, \nonumber \\
 \ker_\Phi &= \{ \cket{\rho} \in \St_A : \cket{\rho}
 ~\mbox{is perfectly distinguishable from}~ \cket{\chi_\Phi} \}.
 \label{eq:derive_face_ker}
\end{alignat}
Note that $\Face_\Phi$ and $\ker_\Phi$ depend only on $\cket{\chi_\Phi}$,
which implies that,
for any $\Phi, \Psi \in \PDS_A$ satisfying $\cket{\chi_\Phi} = \cket{\chi_\Psi}$,
$\Face_\Phi = \Face_\Psi$ and $\ker_\Phi = \ker_\Psi$ hold.

Lemma~\ref{lemma:Perp} can be generalized as follows:
\begin{lemma} \label{lemma:PerpPDS}
 We have that, for $\Phi \in \PDS_A$ and $\cket{\rho} \in \St_A$,
 \begin{alignat}{1}
  \cket{\rho} \in \ker_\Phi &\quad\Leftrightarrow\quad \craket{\chi_\Phi^\dagger|\rho} = 0.
 \end{alignat}
\end{lemma}
\begin{proof}
 Let $\Phi \coloneqq \{ \cket{\phi_i} \}_{i=1}^k$.
 The case $k = 0$ is obvious since $\cket{\chi_\Phi} = \cket{\emptyset}$ and
 $\ker_\Phi = \St_A$ hold.
 Suppose $k \ge 1$.

 $\Rightarrow$:
 Let $\Pi \coloneqq \{ \cra{e_\Phi}, \cra{e_\rho} \}$ be a measurement
 that perfectly distinguishes between $\{ \cket{\chi_\Phi}, \cket{\rho} \}$.
 For each $i \in \{1,\ldots,k\}$,
 we have $\cket{\phi_i} \in \Face_\Phi$, and thus
 $\craket{e_\Phi|\phi_i} = \craket{\gdis|\phi_i} = 1$ holds from Lemma~\ref{lemma:FaceOne}.
 This implies that $\Pi$ perfectly distinguishes between $\{ \cket{\phi_i}, \cket{\rho} \}$,
 i.e., from Lemma~\ref{lemma:Perp}, $\craket{\phi_i^\dagger|\rho} = 0$ holds.
 Therefore, $\craket{\chi_\Phi^\dagger|\rho} = \sum_{i=1}^k \craket{\phi_i^\dagger|\rho} = 0$
 holds.

 $\Leftarrow$:
 Since $\cra{\chi_\Phi^\dagger}$ is feasible from Lemma~\ref{lemma:PDSChi},
 $\Pi' \coloneqq \{ \cra{\chi_\Phi^\dagger}, \cra{\gdis} - \cra{\chi_\Phi^\dagger} \}$
 is a measurement.
 We have
 \begin{alignat}{3}
  \craket{\chi_\Phi^\dagger|\chi_\Phi} &= \sum_{i=1}^k \sum_{j=1}^k \craket{\phi_i^\dagger|\phi_j}
  &&= k &&= \craket{\gdis|\chi_\Phi}, \nonumber \\
  [\cra{\gdis} - \cra{\chi_\Phi^\dagger}] \cket{\rho} &= \craket{\gdis|\rho},
 \end{alignat}
 where the second equality in the first line follows from
 $\craket{\phi_i^\dagger|\phi_j} = \delta_{i,j}$.
 Therefore, $\cket{\chi_\Phi}$ and $\cket{\rho}$ are perfectly distinguished by $\Pi'$.
\end{proof}
Lemma~\ref{lemma:Perp} is the special case of this lemma where $\Phi = \{ \cket{\phi} \}$.
From Lemma~\ref{lemma:PerpPDS}, $\ker_\Phi$ is rewritten as
$\ker_\Phi = \left\{ \cket{\rho} \in \St_A : \craket{\chi_\Phi^\dagger|\rho} = 0 \right\}$.
It is easily seen that $\craket{\rho^\dagger|\sigma} = 0$ holds
for any $\cket{\rho} \in \ker_\Phi$ and $\cket{\sigma} \in \Face_\Phi$.

\begin{lemma} \label{lemma:PDSJoin}
 For any $\Phi_1, \Phi_2 \in \PDS_A$ with $\craket{\chi_{\Phi_1}^\dagger|\chi_{\Phi_2}} = 0$,
 $\Phi_1 \cup \Phi_2 \in \PDS_A$ holds.
\end{lemma}
\begin{proof}
 If $|\Phi_1| = 0$ or $|\Phi_2| = 0$ holds,
 then the lemma is obvious.
 Suppose $|\Phi_1| \ge 1$ and $|\Phi_2| \ge 1$.
 Let $\Phi_1 \coloneqq \{ \cket{\phi_i} \}_{i=1}^k$ and
 $\Phi_2 \coloneqq \{ \cket{\phi_i} \}_{i=k+1}^m$ with $0 < k < m$.
 Also, let $\Psi_l \coloneqq \{ \cket{\phi_i} \}_{i=1}^{k+l}$
 for each $l \in \{ 0,1,\ldots,m-k \}$.
 To prove that $\Psi_{m-k} = \Phi_1 \cup \Phi_2$ is a PDS,
 we proceed by induction on $l$.
 $\Psi_0 = \Phi_1 \in \PDS_A$ obviously holds.
 Suppose $\Psi_l \in \PDS_A$ with $0 \le l < m - k$.
 Let $t \coloneqq k + l + 1$.
 Since $\cket{\phi_t} \in \Face_{\Phi_2}$ holds from $\cket{\phi_t} \in \Phi_2$,
 Lemma~\ref{lemma:FaceZero} gives $\craket{\chi_{\Phi_1}^\dagger|\phi_t} = 0$.
 Also, from $\Phi_2 \in \PDS_A$,
 $\craket{\phi_{k+i}^\dagger|\phi_t} = 0$ holds for any $i \in \{1,\ldots,l\}$.
 Thus, we have
 \begin{alignat}{1}
  \includegraphics[scale=1.0]{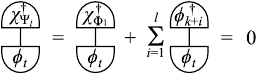} ~\raisebox{1.0em}{.}
 \end{alignat}
 Let $\cra{e} \coloneqq \cra{\gdis} - \cra{\chi_{\Psi_l}^\dagger}$;
 then, from Lemma~\ref{lemma:PDSChi}, $\cra{e} \in \EffF_A$ holds.
 This implies that
 $\Pi \coloneqq \{ \cra{\phi_1^\dagger}, \ldots, \cra{\phi_{k+l}^\dagger}, \cra{e} \}$
 is a measurement.
 We also have
 \begin{alignat}{1}
  \includegraphics[scale=1.0]{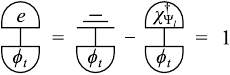} ~\raisebox{1.0em}{.}
 \end{alignat}
 Thus, $\Psi_{l+1}$ is perfectly distinguished by $\Pi$ and thus a PDS.
 Therefore, $\Psi_{m-k}$ is a PDS.
\end{proof}

\subsection{Results about symmetric sharpness and complete mixing}

We here discuss an OPT satisfying
the \poslink{pos:SymSharp}{symmetric sharpness} and
\poslink{pos:CompletelyMixed}{complete mixing} postulates.

\begin{lemma} \label{lemma:Chi}
 $\cket{\chi_\Phi} = \cket{\chi_\Psi}$ holds for any $\Phi, \Psi \in \MPDS_A$.
\end{lemma}
\begin{proof}
 Let $p$ be the maximum value of $p' \in \Real_+$ satisfying
 $p' \cket{\chi_\Phi} \le \cket{\chi_\Psi}$.
 Also, let
 \begin{alignat}{1}
  \includegraphics[scale=1.0]{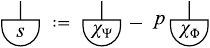} ~\raisebox{.5em}{.}
  \label{fig:derive_chi_proof_s}
 \end{alignat}
 Since $\cket{\chi_\Phi} \not\in \Face_\cket{s}$ holds from the definition of $p$,
 $\cket{s}$ is not completely mixed.
 Thus, from the \poslink{pos:CompletelyMixed}{complete mixing postulate},
 there exists a maximal effect $\cra{e}$ satisfying $\craket{e|s} = 0$.
 Therefore, we have
 \begin{alignat}{1}
  \includegraphics[scale=1.0]{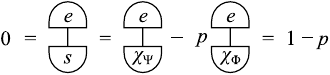} ~\raisebox{1.3em}{;}
 \end{alignat}
 i.e., $p = 1$,
 where the last equality follows from
 $\craket{e|\chi_\Psi} = \craket{e|\chi_\Phi} = 1$
 by Lemma~\ref{lemma:ChiOne}.
 Substituting $p = 1$ into Eq.~\eqref{fig:derive_chi_proof_s} gives
 \begin{alignat}{1}
  \includegraphics[scale=1.0]{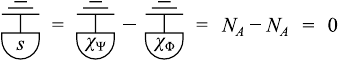} ~\raisebox{1.0em}{,}
 \end{alignat}
 where the second equality follows from Lemma~\ref{lemma:Rank}.
 Hence, we have $\cket{s} = \cket{\emptyset}$,
 which means $\cket{\chi_\Phi} = \cket{\chi_\Psi}$.
\end{proof}
By lemma~\ref{lemma:Chi}, we know that
$\cket{\chi_\Phi}$ with $\Phi \in \MPDS_A$ depends only on $A \in \System$
and not on $\Phi$.
Let us denote such $\cket{\chi_\Phi}$ by $\cket{\chi_A}$ or simply $\cket{\chi}$
and refer to it as the \termdef{invariant state} of $A$.
Clearly, $\cra{\chi_A^\dagger} = \cra{\gdis_A}$ holds.

\begin{ex}[quantum theory]
 $\cket{\chi_A}$ is the identity matrix of order $\NA$.
\end{ex}

\begin{lemma} \label{lemma:ChiCM}
 The invariant state is completely mixed.
\end{lemma}
\begin{proof}
 Any $\cket{\rho} \in \St_A$ has the form $\cket{\rho} = \sum_{i=1}^l c_i \cket{\psi_i}$
 with $c_1,\ldots,c_l \in \Realpp$ and $\cket{\psi_1},\ldots,\cket{\psi_l} \in \StNP_A$.
 For each $i \in \{1,\ldots,l\}$, consider an MPDS $\Phi$ including $\cket{\psi_i}$;
 then, it follows that $\cket{\psi_i} \le \cket{\chi_\Phi} = \cket{\chi_A}$ holds,
 which leads to $\cket{\rho} \in \Face_\cket{\chi_A}$.
 Therefore, $\cket{\chi_A}$ is completely mixed.
\end{proof}

We will say that a set of $k$ PDSs $\{ \Phi_i \in \PDS_A \}_{i=1}^k$
\termdef{can compose an MPDS} if
$\Phi_1,\ldots,\Phi_k$ are disjoint and $\bigcup_{i=1}^k \Phi_i$ is an MPDS.
In particular, we will call two PDSs $\Phi$ and $\Psi$ \termdef{complementary}
if $\{ \Phi, \Psi \}$ can compose an MPDS.
We see at once that, for any $\Phi \in \PDS_A$,
there exists $\Psi \in \PDS_A$ such that $\Phi$ and $\Psi$ are complementary.

\begin{lemma} \label{lemma:PDSkMPDS}
 Consider a set of $k$ PDSs $\tPhi \coloneqq \{ \Phi_i \in \PDS_A \}_{i=1}^k$.
 Then, $\tPhi$ can compose an MPDS if and only if
 $\sum_{i=1}^k \cket{\chi_{\Phi_i}} = \cket{\chi}$ holds.
\end{lemma}
By this lemma, $\Phi, \Psi \in \PDS_A$ are complementary if and only if
$\cket{\chi_\Phi} + \cket{\chi_\Psi} = \cket{\chi}$ holds.
\begin{proof}
 The case $k \le 1$ is obvious; suppose $k > 1$.

 ``Only if'':
 Since $\Phi_1,\ldots,\Phi_k$ are disjoint and
 $\Phi \coloneqq \bigcup_{i=1}^k \Phi_i$ is an MPDS,
 we have $\sum_{i=1}^k \cket{\chi_{\Phi_i}} = \cket{\chi_\Phi} = \cket{\chi}$.

 ``If'':
 From $\sum_{i=1}^k \cket{\chi_{\Phi_i}} = \cket{\chi}$, we have that,
 for any $j \in \{1,\ldots,k\}$,
 \begin{alignat}{3}
  \sum_{i=1}^k \craket{\chi_{\Phi_j}^\dagger|\chi_{\Phi_i}}
  &= \craket{\chi_{\Phi_j}^\dagger|\chi}
  &&= |\Phi_j| &&= \craket{\chi_{\Phi_j}^\dagger|\chi_{\Phi_j}},
 \end{alignat}
 where the second equality follows from Lemma~\ref{lemma:ChiOne}.
 Thus, $\craket{\chi_{\Phi_j}^\dagger|\chi_{\Phi_i}} = 0$ holds
 for any distinct $i,j \in \{1,\ldots,k\}$.
 Therefore, by applying Lemma~\ref{lemma:PDSJoin} recursively,
 we can see that $\tPhi$ can compose an MPDS.
\end{proof}

\begin{lemma} \label{lemma:PDSComplementKer}
 If two PDSs $\Phi$ and $\Psi$ satisfy $\cket{\chi_\Phi} \in \ker_\Psi$,
 then there exists a PDS that includes $\Phi$ and is complementary to $\Psi$.
\end{lemma}
\begin{proof}
 Using Lemma~\ref{lemma:PerpPDS}, we have $\craket{\chi_{\Psi}^\dagger|\chi_\Phi} = 0$.
 From this and Lemma~\ref{lemma:PDSJoin}, $\Phi \cup \Psi$ is a PDS
 such that $\cket{\chi_{\Phi \cup \Psi}} = \cket{\chi_\Phi} + \cket{\chi_\Psi}$.
 Let $\Phi'$ be a PDS complementary to $\Phi \cup \Psi$;
 then, $\craket{\chi_{\Phi'}^\dagger|\chi_\Phi} = 0$ holds,
 and thus, from Lemma~\ref{lemma:PDSJoin}, $\Phi \cup \Phi'$ is a PDS including $\Phi$.
 Since $\cket{\chi_{\Phi \cup \Phi'}} + \cket{\chi_\Psi} =
 \cket{\chi_\Phi} + \cket{\chi_\Psi} + \cket{\chi_{\Phi'}}
 = \cket{\chi_{\Phi \cup \Psi}} + \cket{\chi_{\Phi'}} = \cket{\chi}$ holds
 from Lemma~\ref{lemma:PDSkMPDS},
 $\Phi \cup \Phi'$ is complementary to $\Psi$.
\end{proof}

\begin{lemma} \label{lemma:PDSComplementLe}
 Consider $\Psi \in \PDS_A$.
 Then, any $\cket{\phi} \in \StNP_A \cap \ker_\Psi$ satisfies
 $\cket{\phi} \le \cket{\chi} - \cket{\chi_\Psi}$.
\end{lemma}
\begin{proof}
 From Lemma~\ref{lemma:PDSComplementKer},
 there exists a PDS $\Phi$ that includes $\cket{\phi}$ and
 is complementary to $\Psi$.
 Therefore, $\cket{\phi} \le \cket{\chi_{\Phi}} = \cket{\chi} - \cket{\chi_\Psi}$
 holds, where the equality follows from Lemma~\ref{lemma:PDSkMPDS}.
\end{proof}

\begin{lemma} \label{lemma:FacePhi}
 For $\Phi \in \PDS_A$ and $\cket{\rho} \in \St_A$, we have
 \begin{alignat}{1}
  \cket{\rho} \in \Face_\Phi &\quad\Leftrightarrow\quad
  \craket{\chi_\Phi^\dagger|\rho} = \craket{\gdis|\rho}.
  \label{eq:Face_Phi_distinguish}
 \end{alignat}
\end{lemma}
\begin{proof}
 $\Rightarrow$:
 From $\craket{\chi_\Phi^\dagger|\chi_\Phi} = \craket{\gdis|\chi_\Phi}$ and
 Lemma~\ref{lemma:FaceOne}, we have $\craket{\chi_\Phi^\dagger|\rho} = \craket{\gdis|\rho}$.

 $\Leftarrow$:
 The case $\cket{\rho} = \cket{\emptyset}$ is obvious,
 so assume $\cket{\rho} \neq \cket{\emptyset}$.
 Let $\Psi$ be a PDS complementary to $\Phi$.
 For any $\cket{\psi} \in \StNP_A \cap \Face_\cket{\rho}$,
 from $\craket{\chi_\Psi^\dagger|\rho} = 0$ and Lemma~\ref{lemma:FaceZero},
 $\craket{\chi_\Psi^\dagger|\psi} = 0$ holds.
 Thus, from Lemma~\ref{lemma:PDSComplementLe},
 $\cket{\psi} \le \cket{\chi} - \cket{\chi_\Psi} = \cket{\chi_\Phi}$ holds,
 which yields $\cket{\psi} \in \Face_\Phi$.
 Since $\cket{\rho}$ can be represented by a weighted sum of normalized pure states
 in $\Face_\Phi$, $\cket{\rho} \in \Face_\Phi$ holds.
\end{proof}

\begin{lemma} \label{lemma:KernelPhi}
 For two PDSs $\Phi, \Psi \in \PDS_A$, $\ker_\Phi = \Face_\Psi$ holds if and only if
 $\Phi$ and $\Psi$ are complementary.
\end{lemma}
\begin{proof}
 ``Only if'':
 Since $\cket{\chi_\Psi} \in \Face_\Psi = \ker_\Phi$ holds,
 it follows from Lemma~\ref{lemma:PDSComplementKer} that
 there exists a PDS $\Psi'$ that includes $\Psi$ and is complementary to $\Phi$.
 If $\Psi \neq \Psi'$ holds, then
 there must exist $\cket{\psi} \in \Psi'$ such that $\cket{\psi} \not\in \Psi$.
 Such $\cket{\psi}$ satisfies $\cket{\psi} \in \ker_\Phi$ and $\cket{\psi} \not\in \Face_\Psi$,
 which contradicts $\ker_\Phi = \Face_\Psi$.
 Therefore, $\Psi = \Psi'$ holds, i.e., $\Phi$ and $\Psi$ are complementary.

 ``If'':
 Since $\cket{\chi_\Psi} + \cket{\chi_\Phi} = \cket{\chi}$ holds from Lemma~\ref{lemma:PDSkMPDS},
 $\craket{\chi_\Psi^\dagger|\rho} + \craket{\chi_\Phi^\dagger|\rho} = \craket{\gdis|\rho}$
 holds for any $\cket{\rho} \in \St_A$.
 Thus, we have
 \begin{alignat}{2}
  \cket{\rho} \in \ker_\Phi
  &\quad\Leftrightarrow\quad \craket{\chi_\Phi^\dagger|\rho} = 0 \nonumber \\
  &\quad\Leftrightarrow\quad \craket{\chi_\Psi^\dagger|\rho} = \craket{\gdis|\rho}
  &&\quad\Leftrightarrow\quad \cket{\rho} \in \Face_\Psi,
 \end{alignat}
 which follows from Lemmas~\ref{lemma:PerpPDS} and \ref{lemma:FacePhi}.
 Therefore, $\ker_\Phi = \Face_\Psi$ holds.
\end{proof}
Lemma~\ref{lemma:KernelPhi} still holds if we exchange $\Phi$ and $\Psi$.
Thus, for any two PDSs $\Phi, \Psi \in \PDS_A$,
$\ker_\Phi = \Face_\Psi$ and $\Face_\Phi = \ker_\Psi$ are obviously equivalent.

\section{Symmetric properties of state space} \label{sec:derive_symmetry}

In this section, we show Property~(\hyperlink{property:diamondsuit}{$\diamondsuit$})
in an OPT that satisfies
the \poslink{pos:SymSharp}{symmetric sharpness},
\poslink{pos:CompletelyMixed}{complete mixing},
and \poslink{pos:Filter}{filtering} postulates.
The main results in this section are:
\begin{enumerate}
 \item Every state has a spectral decomposition,
       which means that any $\cket{\rho} \in \St_A$ can be expressed in the form
       $\cket{\rho} = \sum_{i=1}^\NA p_i \cket{\phi_i}$
       with $p_1,\ldots,p_\NA \in \Real_+$ and
       $\{ \cket{\phi_i} \}_{i=1}^\NA \in \MPDS_A$ (Proposition~\ref{pro:Spectral}).
 \item $\St_A$ is a symmetric cone (Theorem~\ref{thm:SymmetricCone}).
 \item $\Eff_A$ is the dual cone of $\St_A$ (Theorem~\ref{thm:Eff}).
\end{enumerate}

\subsection{Symmetric cones} \label{subsec:derive_abstract_symmetric}

We here review the definitions of symmetric cones.
A convex cone $\mC$ in a real vector space $\Vec$ is called a \termdef{symmetric cone}
if $\mC$ is self-dual and homogeneous, whose definitions are given below.

First, we will recall the definition of self-duality.
Let $\mC$ be a convex cone in a real vector space $\Vec$.
\begin{alignat}{1}
 \mC^* &\coloneqq \{ f \in \Vec^* : \forall x \in \mC, f(x) \ge 0 \}
 \label{eq:derive_C_ast}
\end{alignat}
is called the \termdef{dual cone} of $\mC$,
where $\Vec^*$ is the dual vector space of $\Vec$.
One can easily verify that $\mC^*$ is a closed convex cone.
For any inner product of $\Vec$, denoted by $\braket{~,~}$,
there exists an isomorphism $\#:\Vec^* \ni f \mapsto f^\# \in \Vec$
satisfying $f(x) = \braket{f^\#,x}$ for any $x \in \Vec$ and $f \in \Vec^*$.
Let
\begin{alignat}{1}
 \mC^\star &\coloneqq \{ f^\# \in \Vec : f \in \mC^* \};
 \label{eq:derive_C_star}
\end{alignat}
then, $\mC^\star = \{ x \in \Vec : \forall y \in \mC, \braket{x, y} \ge 0 \}$
obviously holds.
$\mC^*$ and $\mC^\star$ are isomorphic as convex cones.
Also, the restriction of $\#$ to $\mC^*$ is an isomorphism from $\mC^*$ to $\mC^\star$.
If $\mC^\star = \mC$ holds for some inner product $\braket{~,~}$,
then $\mC$ is called \termdef{self-dual} with respect to $\braket{~,~}$.
In this case, $\mC \cong \mC^*$ obviously holds.
In particular, the state space $\St_A$ is self-dual if
there exists an inner product $\braket{~,~}$ such that
\begin{alignat}{1}
 \St_A &= \{ \cket{\ol{x}} \in \Vec_A : \forall \cket{\rho} \in \St_A,
 \braket{\ol{x}, \rho} \ge 0 \},
 \label{eq:derive_StA_selfdual}
\end{alignat}
where $\braket{\ol{x},\ol{y}}$ $~(\cket{\ol{x}},\cket{\ol{y}} \in \Vec_A)$
is a simple notation for $\braket{\cket{\ol{x}},\cket{\ol{y}}}$.

Next, we will give the definition of homogeneity.
A convex cone $\mC$ is called \termdef{homogeneous}
if, for any two interior points $x,y \in \mC$, there exists
an automorphism $g_{x,y}$ on $\mC$ such that $g_{x,y}(x) = y$.

\subsection{Projection processes} \label{subsec:derive_symmetry_projection}

For any $\Phi \in \PDS_A$,
a process $P_\Phi \in \Proc_{A \to A}$ is called
a \termdef{projection process} onto $\Face_\Phi$ if
\begin{alignat}{1}
 \includegraphics[scale=1.0]{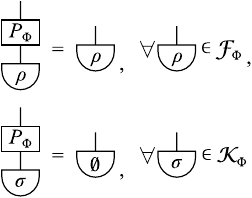}
 \label{eq:derive_Projection}
\end{alignat}
holds.

\begin{ex}[quantum theory]
 For any PDS $\Phi \coloneqq \{ \cket{\phi_i} \}_{i=1}^k$,
 a projection process onto $\Face_\Phi$ is expressed by
 $P_\Phi \circ \cket{\rho} = \cket{\chi_\Phi} \cdot \cket{\rho} \cdot \cket{\chi_\Phi}$.
 One can easily see that $\cket{\chi_\Phi}$ is a projection matrix.
\end{ex}

The following lemma guarantees the existence of projection processes.
\begin{lemma} \label{lemma:Projection}
 For any $\Phi \in \PDS_A$, there exists a projection process $P_\Phi$ onto $\Face_\Phi$.
\end{lemma}
\begin{proof}
 If $\Phi$ is an MPDS, then $\id_A$ is a projection process onto $\Face_\Phi = \St_A$,
 and thus the lemma is obvious.
 Suppose that $\Phi$ is not an MPDS.
 Let $\Phi \coloneqq \{ \cket{\phi_i} \}_{i=1}^k$.
 Also, let $\Psi \coloneqq \{ \cket{\phi_i} \}_{i=k+1}^\NA \in \PDS_A$
 be complementary to $\Phi$.
 It suffices to show that
 $P_\Phi \coloneqq F_{\cket{\phi_\NA}}^0 \circ \cdots
 \circ F_{\cket{\phi_{k+2}}}^0 \circ F_{\cket{\phi_{k+1}}}^0$
 is a projection process onto $\Face_\Phi$,
 where $F_{\cket{\phi_\NA}}^0,\ldots,F_{\cket{\phi_{k+2}}}^0,F_{\cket{\phi_{k+1}}}^0$
 are filters (see Eq.~\eqref{eq:derive_filter}).
 Arbitrarily choose $\cket{\rho} \in \Face_\Phi$ and
 $i \in \{ k+1,\ldots,\NA \}$.
 From $\craket{\phi_i^\dagger|\chi_\Phi} = 0$ and Lemma~\ref{lemma:FaceZero},
 $\craket{\phi_i^\dagger|\rho} = 0$ holds.
 Thus, from Lemma~\ref{lemma:Perp},
 $\cket{\phi_i}$ and $\cket{\rho}$ are perfectly distinguishable.
 Therefore, we have $F_{\cket{\phi_i}}^0 \circ \cket{\rho} = \cket{\rho}$,
 which gives $P_\Phi \circ \cket{\rho} = \cket{\rho}$.
 We also arbitrarily choose $\cket{\sigma} \in \ker_\Phi$.
 From Lemma~\ref{lemma:KernelPhi}, we have $\cket{\sigma} \in \Face_\Psi$,
 i.e., $\delta \cket{\sigma} \le \cket{\chi_\Psi}$ holds for some $\delta \in \Realpp$.
 Thus, $\delta P_\Phi \circ \cket{\sigma} \le P_\Phi \circ \cket{\chi_\Psi} =
 \sum_{j=k+1}^\NA P_\Phi \circ \cket{\phi_j} = \cket{\emptyset}$ obviously holds,
 which gives $P_\Phi \circ \cket{\sigma} = \cket{\emptyset}$.
\end{proof}

\subsection{Spectral decomposition} \label{subsec:derive_symmetry_spectral}

We here show that any state and any extended state have spectral decompositions.
The following lemma is useful for proving the existence of a spectral decomposition
of a state.
\begin{lemma} \label{lemma:StDecomp}
 Consider a PDS $\Phi \in \PDS_A \setminus \MPDS_A$
 (where $\setminus$ denotes the set difference operator).
 Then, for any $\cket{\rho} \in \ker_\Phi$,
 there exist $p \in \Real_+$, $\cket{\phi} \in \StNP_A \cap \ker_\Phi$,
 and $\cket{\rho'} \in \ker_\Phi$ satisfying
 \begin{alignat}{1}
  \includegraphics[scale=1.0]{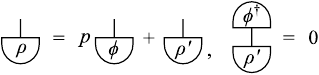} ~\raisebox{1.3em}{.}
  \label{eq:spectral_stdecomp}
 \end{alignat}
 Furthermore, $\Phi' \coloneqq \Phi \cup \{ \cket{\phi} \} \in \PDS_A$
 and $\cket{\rho'} \in \ker_{\Phi'}$ hold.
\end{lemma}
\begin{proof}
 Let $\Psi$ be a PDS complementary to $\Phi$ and
 $p$ be the maximum value of $p' \in \Real_+$ satisfying
 $\cket{\rho} \ge p' \cket{\chi_\Psi}$.
 Also, let $\cket{\sigma} \coloneqq \cket{\rho} - p \cket{\chi_\Psi}$.
 We have
 $0 \le \craket{\chi_\Phi^\dagger|\sigma} \le \craket{\chi_\Phi^\dagger|\rho} = 0$,
 i.e., $\cket{\sigma} \in \ker_\Phi$.

 Firstly, we prove that $\cket{\sigma} + \cket{\chi_\Phi}$ is not completely mixed.
 Assume, by contradiction, that $\cket{\sigma} + \cket{\chi_\Phi}$ is completely mixed;
 then, there exists $\delta \in \Realpp$ such that
 \begin{alignat}{1}
  \includegraphics[scale=1.0]{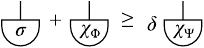} ~\raisebox{.5em}{.}
  \label{eq:spec_sigma}
 \end{alignat}
 Since $\Face_\Phi = \ker_\Psi$ and $\ker_\Phi = \Face_\Psi$ hold
 from Lemma~\ref{lemma:KernelPhi},
 we have $\cket{\sigma} \in \ker_\Phi = \Face_\Psi$,
 $\cket{\chi_\Phi} \in \Face_\Phi = \ker_\Psi$,
 and $\cket{\chi_\Psi} \in \Face_\Psi$.
 Thus, $P_\Psi \circ \cket{\sigma} = \cket{\sigma}$,
 $P_\Psi \circ \cket{\chi_\Phi} = \cket{\emptyset}$,
 and $P_\Psi \circ \cket{\chi_\Psi} = \cket{\chi_\Psi}$ hold,
 where $P_\Psi$ is a projection process onto $\Face_\Psi$.
 This yields
 \begin{alignat}{1}
  \includegraphics[scale=1.0]{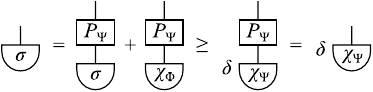} ~\raisebox{1.0em}{,}
 \end{alignat}
 where the inequality follows from Eq.~\eqref{eq:spec_sigma}.
 Therefore, we have
 \begin{alignat}{1}
  \includegraphics[scale=1.0]{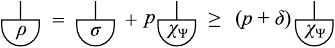} ~\raisebox{.5em}{,}
 \end{alignat}
 which contradicts the definition of $p$.
 Hence, $\cket{\sigma} + \cket{\chi_\Phi}$ is not completely mixed.

 Secondly, we show that
 there exist $\cket{\phi} \in \StNP_A \cap \ker_\Phi$
 and $\cket{\rho'} \in \ker_\Phi$ satisfying Eq.~\eqref{eq:spectral_stdecomp}.
 From the \poslink{pos:CompletelyMixed}{complete mixing postulate},
 there exists $\cket{\phi} \in \StNP_A$ satisfying
 $\cra{\phi^\dagger} [\cket{\sigma} + \cket{\chi_\Phi}] = 0$,
 i.e., $\craket{\phi^\dagger|\sigma} = \craket{\phi^\dagger|\chi_\Phi} = 0$.
 From $\craket{\phi^\dagger|\chi_\Phi} = 0$, we have $\cket{\phi} \in \ker_\Phi$.
 From Lemma~\ref{lemma:PDSComplementLe},
 $\cket{\phi} \le \cket{\chi_\Psi}$ holds,
 and thus $\cket{\rho} = \cket{\sigma} + p \cket{\chi_\Psi} \ge
 p \cket{\chi_\Psi} \ge p \cket{\phi}$ holds.
 Let $\cket{\rho'} \coloneqq \cket{\rho} - p \cket{\phi} \in \St_A$;
 then, $\cket{\rho'} \in \ker_\Phi$ holds from $\cket{\rho} \in \ker_\Phi$.
 Moreover, we have
 \begin{alignat}{1}
  \includegraphics[scale=1.0]{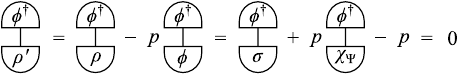} ~\raisebox{1.3em}{,}
 \end{alignat}
 where the last equality follows from the fact that,
 from Lemma~\ref{lemma:ChiOne}, we have $\craket{\phi^\dagger|\chi} = 1$
 and thus $\craket{\phi^\dagger|\chi_\Psi}
 = \craket{\phi^\dagger|\chi} - \craket{\phi^\dagger|\chi_\Phi} = 1$.
 Therefore, Eq.~\eqref{eq:spectral_stdecomp} holds.

 Finally, we show $\Phi' \in \PDS_A$ and $\cket{\rho'} \in \ker_{\Phi'}$.
 From $\craket{\phi^\dagger|\chi_\Phi} = 0$ and Lemma~\ref{lemma:PDSJoin},
 we have $\Phi' \in \PDS_A$.
 Also, since we have $\craket{\chi_{\Phi'}^\dagger|\rho'} = \craket{\chi_\Phi^\dagger|\rho'}
 + \craket{\phi^\dagger|\rho'} = 0$,
 $\cket{\rho'} \in \ker_{\Phi'}$ holds from Lemma~\ref{lemma:PerpPDS}.
\end{proof}

Now, we are in a position to show the following proposition.
\begin{proposition}[Spectral decompositions of states] \label{pro:Spectral}
 For any $\cket{\rho} \in \St_A$,
 there exist $p_1,\ldots,p_\NA \in \Real_+$ and
 $\{ \cket{\phi_i} \}_{i=1}^\NA \in \MPDS_A$ satisfying
 \begin{alignat}{1}
  \includegraphics[scale=1.0]{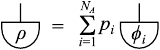} ~\raisebox{.5em}{.}
  \label{eq:spectral}
 \end{alignat}
\end{proposition}
The decomposition shown in Eq.\eqref{eq:spectral} is called a \termdef{spectral decomposition}
of $\cket{\rho}$.
\begin{proof}
 Let $\cket{\rho_1} \coloneqq \cket{\rho}$.
 We show that, for each $k \in \{1,\ldots,\NA\}$,
 $\cket{\rho_k}$ can be expressed in the form
 \begin{alignat}{1}
  \includegraphics[scale=1.0]{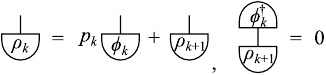} ~\raisebox{1.3em}{.}
  \label{eq:spectral_main}
 \end{alignat}
 We now proceed by induction on $k$.
 Let $\Phi_0$ be the empty set.
 Since $\Phi_0 \in \PDS_A \setminus \MPDS_A$ and
 $\cket{\rho_1} \in \St_A = \ker_{\Phi_0}$ hold,
 Lemma~\ref{lemma:StDecomp} gives that
 there exist $p_1 \in \Real_+$, $\cket{\phi_1} \in \StNP_A$,
 and $\cket{\rho_2} \in \St_A$ satisfying Eq.~\eqref{eq:spectral_main} with $k = 1$.
 $\cket{\rho_2} \in \ker_{\Phi_1}$ also holds, where $\Phi_1 \coloneqq \{ \cket{\phi_1} \}$.
 Consider the case $k \in \{2,\ldots,\NA\}$;
 let $\Phi_k \coloneqq \Phi_{k-1} \cup \{\cket{\phi_k}\} = \{ \cket{\phi_i} \}_{i=1}^k$.
 Since $\Phi_{k-1} \in \PDS_A \setminus \MPDS_A$ and
 $\cket{\rho_k} \in \ker_{\Phi_{k-1}}$ hold,
 from Lemma~\ref{lemma:StDecomp},
 there exist $p_k \in \Real_+$, $\cket{\phi_k} \in \StNP_A \cap \ker_{\Phi_{k-1}}$,
 and $\cket{\rho_{k+1}} \in \ker_{\Phi_k}$ satisfying Eq.~\eqref{eq:spectral_main}.

 From $|\Phi_\NA| = \NA$, $\Phi_\NA$ is obviously an MPDS.
 Thus, from $\cket{\rho_{\NA+1}} \in \ker_{\Phi_\NA} = \{ \cket{\emptyset} \}$,
 $\cket{\rho_{\NA+1}} = \cket{\emptyset}$ holds.
 By recursively applying Eq.~\eqref{eq:spectral_main}, we have
 \begin{alignat}{1}
  \includegraphics[scale=1.0]{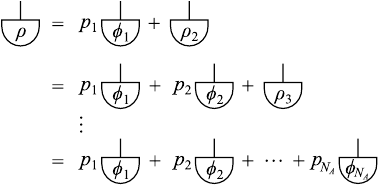} ~\raisebox{.5em}{.}
  \label{eq:spectral_proof_rho_iterate}
 \end{alignat}
\end{proof}

Using Proposition~\ref{pro:Spectral}, we can show that
every extended state has also a spectral decomposition.
\begin{proposition}[Spectral decompositions of extended states]
 \label{pro:SpectralVec} 
 For any $\cket{\ol{v}} \in \Vec_A$,
 there exist $c_1,\ldots,c_\NA \in \Real$ and
 $\{ \cket{\phi_i} \}_{i=1}^\NA \in \MPDS_A$ satisfying
 \begin{alignat}{1}
  \includegraphics[scale=1.0]{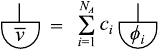} ~\raisebox{.5em}{.}
  \label{eq:spectral_vec}
 \end{alignat}
 \end{proposition}
The decomposition shown in Eq.\eqref{eq:spectral_vec} is called a \termdef{spectral decomposition}
of $\cket{\ol{v}}$.
\begin{proof}
 $\cket{\ol{v}} \in \Vec_A$ can be expressed in the form
 $\cket{\ol{v}} = \cket{v_+} - \cket{v_-}$ for some $\cket{v_+}, \cket{v_-} \in \St_A$.
 Arbitrarily choose $q \in \Real_+$ such that $q \cket{\chi} \ge \cket{v_-}$.
 Note that since, from Lemma~\ref{lemma:ChiCM}, $\cket{\chi}$ is completely mixed,
 such $q$ exists.
 Since $\cket{\rho} \coloneqq \cket{\ol{v}} + q \cket{\chi}
 = \cket{v_+} + [q \cket{\chi} - \cket{v_-}] \ge \cket{v_+}$ holds,
 $\cket{\rho} \in \St_A$ holds.
 From Proposition~\ref{pro:Spectral}, $\cket{\rho}$ has a spectral decomposition
 as in Eq.~\eqref{eq:spectral}
 with $p_1,\ldots,p_\NA \in \Real_+$ and $\{ \cket{\phi_i} \}_{i=1}^\NA \in \MPDS_A$.
 Thus, letting $c_i \coloneqq p_i - q$, we have
 \begin{alignat}{1}
  \includegraphics[scale=1.0]{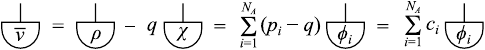} ~\raisebox{.5em}{.}
 \end{alignat}
\end{proof}

\begin{ex}[quantum theory]
 A spectral decomposition of a state (resp. extended state)
 is a spectral decomposition of a positive semidefinite matrix
 (resp. Hermitian matrix).
\end{ex}

\subsection{Self-duality} \label{subsec:SelfDual}

Using Proposition~\ref{pro:Spectral}, we will derive that each state space is self-dual.

\begin{lemma} \label{lemma:InnerProd}
 The binary operation $\braket{~,~}$ on $\Vec_A$ defined by
 \begin{alignat}{3}
  \braket{\ol{v}, \ol{w}}
  &\coloneqq \craket{\ol{v}^\dagger|\ol{w}},
  &\quad \cket{\ol{v}}, \cket{\ol{w}} \in \Vec_A
  \label{eq:derive_inner_prod}
 \end{alignat}
 is an inner product.
\end{lemma}
\begin{proof}
 It suffices to show that
 $\braket{~,~}$ satisfies
 (1)symmetry: $\braket{\ol{v},\ol{w}} = \braket{\ol{w},\ol{v}}$
 $~(\forall \cket{\ol{v}},\cket{\ol{w}} \in \Vec_A)$,
 (2)linearity in the second argument: $\braket{\ol{v},a_1 \cket{\ol{w_1}} + a_2 \cket{\ol{w_2}}}
 = a_1 \braket{\ol{v},\ol{w_1}} + a_2 \braket{\ol{v},\ol{w_2}}$
 $~(\forall a_1,a_2 \in \Real, \cket{\ol{v}},\cket{\ol{w_1}},\cket{\ol{w_2}} \in \Vec_A)$,
 and (3)positive-definiteness:
 $\braket{\ol{v},\ol{v}} > 0$
 $~(\forall \cket{\ol{v}} \in \Vec_A \setminus \{ \cket{\emptyset }\})$.

 (1):
 $\cket{\ol{v}} \in \Vec_A$ and $\cket{\ol{w}} \in \Vec_A$ can be expressed in the form
 $\cket{\ol{v}} = \sum_{i=1}^l c_i \cket{\phi_i}$
 and $\cket{\ol{w}} = \sum_{j=1}^t d_j \cket{\varphi_j}$,
 where $c_1,\ldots,c_l, d_1,\ldots,d_t \in \Real$ and
 $\cket{\phi_1},\ldots,\cket{\phi_l}, \cket{\varphi_1},\ldots,\cket{\varphi_t} \in \StNP_A$
 hold.
 We have
 \begin{alignat}{1}
  \includegraphics[scale=1.0]{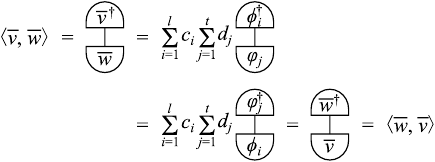} ~\raisebox{1.0em}{.}
 \end{alignat}

 (2): For any $\cket{\ol{v}} \in \Vec_A$,
 since $\cra{\ol{v}^\dagger} \in \Vec_A^*$ is a linear functional on $\Vec_A$,
 the map $\Vec_A \ni \cket{\ol{w}} \mapsto \braket{\ol{v},\ol{w}}
 = \craket{\ol{v}^\dagger|\ol{w}} \in \Real$ is obviously linear.

 (3): For any $\cket{\ol{v}} \in \Vec_A \setminus \{ \cket{\emptyset} \}$,
 which has a spectral decomposition of the form $\cket{\ol{v}} = \sum_{i=1}^\NA c_i \cket{\phi_i}$
 with $c_1,\ldots,c_\NA \in \Real$ and $\{ \cket{\phi_i} \}_{i=1}^\NA$, we have
 \begin{alignat}{1}
  \includegraphics[scale=1.0]{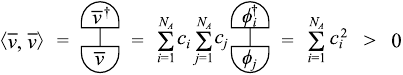} ~\raisebox{0.8em}{.}
 \end{alignat}
\end{proof}

\begin{ex}[quantum theory]
 $\braket{\ol{v},\ol{w}} \coloneqq \craket{\ol{v}^\dagger|\ol{w}}
 = \Tr~[\cket{\ol{v}} \cdot \cket{\ol{w}}]$
 $~(\cket{\ol{v}}, \cket{\ol{w}} \in \Vec_A)$ is an inner product.
\end{ex}

\begin{proposition} \label{pro:SelfDual} 
 For any system $A$, $\St_A$ is self-dual with respect to the inner product
 $\braket{~,~}$ of Eq.~\eqref{eq:derive_inner_prod}.
\end{proposition}
\begin{proof}
 Arbitrarily choose $\cket{\ol{v}} \in \Vec_A$.
 It suffices to show that
 $\cket{\ol{v}} \in \St_A$ holds if and only if
 $\braket{\ol{v},\rho} \ge 0$ holds for any $\cket{\rho} \in \St_A$.
 The ``only if'' part is obvious from $\cra{\ol{v}^\dagger} \in \Eff_A$.
 To prove the ``if'' part,
 we assume $\cket{\ol{v}} \not\in \St_A$ and show that
 there exists $\cket{\rho} \in \St_A$ such that $\braket{\ol{v},\rho} < 0$.
 $\cket{\ol{v}}$ has a spectral decomposition
 $\cket{\ol{v}} = \sum_{i=1}^\NA c_i \cket{\phi_i}$
 with $c_1,\ldots,c_\NA \in \Real$ and $\{ \cket{\phi_i} \}_{i=1}^\NA \in \MPDS_A$.
 From $\cket{\ol{v}} \not\in \St_A$,
 $c_i < 0$ holds for some $i \in \{1,\ldots,\NA\}$.
 For such $i$, we have $\braket{\ol{v},\phi_i} = \craket{\ol{v}^\dagger|\phi_i} = c_i < 0$.
 Thus, the ``if'' part is proved.
\end{proof}

\subsection{Homogeneity and symmetry}

We will derive that each state space is homogeneous and thus symmetric.
Homogeneity is easily derived from Proposition~\ref{pro:Spectral} and
the \poslink{pos:Filter}{filtering postulate}.

\begin{proposition} \label{pro:Homogeneous} 
 For any system $A$, $\St_A$ is homogeneous.
\end{proposition}
\begin{proof}
 It is sufficient to show that, for any completely mixed state $\cket{\rho} \in \St_A$,
 there exists a reversible process $f_\cket{\rho} \in \Proc_{A \to A}$
 satisfying $f_\cket{\rho} \circ \cket{\chi} = \cket{\rho}$.
 Indeed, in this case, let $\tf_\cket{\rho}$ be an inverse process of $f_\cket{\rho}$;
 then, for any completely mixed states $\cket{\rho_1},\cket{\rho_2} \in \St_A$,
 $g \coloneqq f_\cket{\rho_2} \circ \tf_\cket{\rho_1}$ is an automorphism on $\St_A$
 such that $g \circ \cket{\rho_1} = \cket{\rho_2}$.
 Thus, $\St_A$ is homogeneous.

 We arbitrarily choose a completely mixed state $\cket{\rho} \in \St_A$
 and show that there exists a reversible process $f \in \Proc_{A \to A}$
 satisfying $f \circ \cket{\chi} = \cket{\rho}$.
 $\cket{\rho}$ has a spectral decomposition
 $\cket{\rho} = \sum_{i=1}^\NA p_i \cket{\phi_i}$
 with $p_1,\ldots,p_\NA \in \Real_+$ and $\{ \cket{\phi_i} \}_{i=1}^\NA \in \MPDS_A$.
 Since $\cket{\rho}$ is completely mixed,
 $p_i > 0$ holds for any $i \in \{1,\ldots,\NA\}$.
 Here, let $f \coloneqq F_\cket{\phi_\NA}^{p_\NA} \circ \cdots \circ F_\cket{\phi_2}^{p_2}
 \circ F_\cket{\phi_1}^{p_1} \in \Proc_{A \to A}$.
 Let $\tf \coloneqq \tF_\cket{\phi_1}^{p_1} \circ \tF_\cket{\phi_2}^{p_2} \circ \cdots \circ
 \tF_\cket{\phi_\NA}^{p_\NA}$,
 where $\tF_\cket{\phi_i}^{p_i}$ is an inverse process of $F_\cket{\phi_i}^{p_i}$.
 Then, we can easily see $f \circ \tf \eqlocal \id_A \eqlocal \tf \circ f$,
 i.e., $f$ is reversible.
 Moreover, since $f \circ \cket{\phi_i} = p_i \cket{\phi_i}$ holds
 for any $i \in \{1,\ldots,\NA\}$,
 we have
 \begin{alignat}{1}
  \includegraphics[scale=1.0]{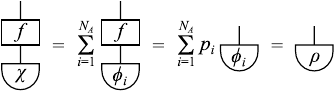} ~\raisebox{.8em}{.}
 \end{alignat}
\end{proof}

\begin{thm} \label{thm:SymmetricCone} 
 For any system $A$, $\St_A$ is a symmetric cone.
\end{thm}
\begin{proof}
 Combining Propositions~\ref{pro:SelfDual} and \ref{pro:Homogeneous},
 $\St_A$ is self-dual and homogeneous.
\end{proof}

The following theorem shows that a cone is symmetric if and only if
it is the cone of squares of some EJA.
\begin{thm}[\pagetarget{thm:KoecherVinberg}{Koecher-Vinberg theorem}\cite{Koe-1957,Vin-1960}]
 \label{thm:Koecher_Vinberg} 
 For any symmetric cone $\mC$ and any interior point $e \in \mC$,
 there exists an EJA $\EJA$ with the identity element $e$
 such that $\mC$ is the cone of squares of $\EJA$.
 Furthermore, $\EJA$ is the real vector space spanned by $\mC$.
 Conversely, for any EJA, its cone of squares is a symmetric cone.
\end{thm}
The definitions of an EJA and its cone of squares are given in Sec.~\ref{sec:derive_EJA}.
For any system $A$, since $\St_A$ is a symmetric cone,
the Koecher-Vinberg theorem immediately yields the following theorem:
\begin{thm} \label{thm:EJA} 
 For any system $A$, there exists an EJA $\EJA_A$ such that
 $\St_A$ is the cone of squares of $\EJA_A$.
 Furthermore, $\EJA_A$ and $\Vec_A$ are the same as real vector spaces.
\end{thm}

\subsection{Effect space is dual cone of state space}

\begin{thm} \label{thm:Eff} 
 For any system $A$, $\Eff_A$ is the dual cone of $\St_A$, i.e., $\Eff_A = \St_A^*$ holds.
\end{thm}
\begin{proof}
 Since $\craket{e'|\rho} \ge 0$ $~(\forall \cket{\rho} \in \St_A)$ holds
 for any $\cra{e'} \in \Eff_A$,
 $\Eff_A \subseteq \St_A^*$ is obvious.
 It remains to prove $\Eff_A \supseteq \St_A^*$.
 Arbitrarily choose $\cra{\ol{e}} \in \St_A^*$.
 Let $\braket{~,~}$ be the inner product defined by Eq.~\eqref{eq:derive_inner_prod}.
 Since $\Vec_A^*$ is the dual vector space of $\Vec_A$,
 there exists $\cket{\ol{e}^\#} \in \Vec_A$ such that
 $\craket{\ol{e}|\ol{x}} = \braket{\ol{e}^\#,\ol{x}}$ $~(\forall \cket{\ol{x}} \in \Vec_A)$.
 Since $\braket{\ol{e}^\#,\rho} \ge 0$ holds for any $\cket{\rho} \in \St_A$,
 from Proposition~\ref{pro:SelfDual}, we have $\cket{\ol{e}^\#} \in \St_A$.
 Consider $\cra{\ol{e}^{\#\dagger}} \coloneqq \cket{\ol{e}^\#}^\dagger \in \Eff_A$;
 then, from Eq.~\eqref{eq:derive_inner_prod},
 $\craket{\ol{e}|\rho} = \braket{\ol{e}^\#,\rho} = \craket{\ol{e}^{\#\dagger}|\rho}$
 holds for any $\cket{\rho} \in \St_A$,
 which yields $\cra{\ol{e}} = \cra{\ol{e}^{\#\dagger}} \in \Eff_A$.
 Therefore, $\Eff_A \supseteq \St_A^*$ holds.
\end{proof}
Theorems~\ref{thm:SymmetricCone} and \ref{thm:Eff}
immediately yield Property~(\hyperlink{property:diamondsuit}{$\diamondsuit$}).
It follows from $\St_A \cong \St_A^*$ and this theorem
that $\St_A$ and $\Eff_A$ are isomorphic as symmetric cones.
The restriction of $\dagger:\Vec_A \to \Vec_A^*$ to $\St_A$
is an isomorphism from $\St_A$ to $\Eff_A$,
whose inverse is the restriction of $\#:\Vec_A^* \to \Vec_A$ to $\Eff_A$.
We abuse notation by using the same symbol $\dagger$ for the map $\#$.
Let $\cket{\ol{e}^\dagger} \coloneqq \cra{\ol{e}}^\dagger \in \Vec_A$
for any $\cra{\ol{e}} \in \Vec_A^*$.
Clearly, $\craket{\ol{e}|\ol{x}} = \braket{\ol{e}^\dagger,\ol{x}}$ holds
for any $\cket{\ol{x}} \in \Vec_A$ and $\cra{\ol{e}} \in \Vec_A^*$.

\section{Euclidean Jordan Algebras (EJAs)} \label{sec:derive_EJA}

This section will be devoted to presenting some basic properties of EJAs
that are needed in the remainder of this paper.

\subsection{EJAs} \label{subsec:derive_EJA_EJA}

We first review the definitions of Jordan algebras, EJAs, and their cones of squares.
A vector space $\EJA$ over some field $\Field$ is called a \termdef{Jordan algebra} if
$\EJA$ is equipped with a commutative bilinear map $(x, y) \mapsto x \bullet y$
satisfying
\begin{alignat}{2}
 x^2 \bullet (y \bullet x) &= (x^2 \bullet y) \bullet x, &\quad \forall x, y &\in \EJA,
 \label{eq:EJA_times}
\end{alignat}
where $x^2 \coloneqq x \bullet x$.
The operator $\bullet$ is said to be the \termdef{Jordan product}.
Note that the Jordan product is not associative in general;
Eq.~\eqref{eq:EJA_times} is less restrictive than the associative condition
(i.e., $x \bullet (y \bullet z) = (x \bullet y) \bullet z$ for any $x,y,z \in \EJA$).
A finite-dimensional Jordan algebra $\EJA$ over $\Real$ is called an \termdef{EJA} if it is equipped
with an inner product $\braket{~,~}$ satisfying
\begin{alignat}{2}
 \braket{x \bullet y, z} &= \braket{y, x \bullet z}, &\quad \forall x, y, z &\in \EJA.
 \label{eq:EJA_inner}
\end{alignat}
Clearly, any EJA $\EJA$ is a real Hilbert space with the inner product $\braket{~,~}$.
We call
\begin{alignat}{1}
 \EJA^+ &\coloneqq \{ x^2 : x \in \EJA \}
 \label{eq:derive_EJA_EJAp}
\end{alignat}
the \termdef{cone of squares} of $\EJA$.

These are two typical examples of EJAs:
\begin{itemize}
 \item An $n$-dimensional real vector space, $\Real^n$,
       with the componentwise product as the Jordan product
       and the usual inner product.
       The Jordan product is commutative and associative, and thus Eq.~\eqref{eq:EJA_times} holds.
       The cone of squares of $\Real^n$ is the nonnegative orthant $\Real_+^n$.
       The state space, $\St_A$, of a classical system $A$ is isomorphic to $\Real_+^\NA$.
 \item The space, $\mS(\Complex^n)$, of all complex Hermitian matrices of order $n$
       equipped with the Jordan product $x \bullet y \coloneqq (x \cdot y + y \cdot x) / 2$
       and the inner product $\braket{x, y} \coloneqq \Tr(x \cdot y)$,
       where $\cdot$ denotes the matrix product.
       One can easily verify that the Jordan product is commutative and
       satisfies Eq.~\eqref{eq:EJA_times},
       but it is not associative if $n$ is larger than 1.
       The cone of squares of $\mS(\Complex^n)$ is $\mS_+(\Complex^n)$.
       The state space, $\St_A$, of a fully quantum system $A$ is isomorphic to $\mS_+(\Complex^\NA)$.
\end{itemize}

\subsection{Fundamental properties of EJAs} \label{subsec:derive_EJA_property}

We next present some fundamental properties of EJAs.
Many proofs are omitted since they can be found in, e.g., Ref.~\cite{Far-1994}
or can be easily obtained.
We will use the notation such as $\cket{\ol{v}}$ to denote an element of an EJA $\EJA$.
If $\cket{\ol{v}}$ is in $\EJA^+$, then we will often denote it by $\cket{v}$.
We will also use the simple notation
$\braket{\ol{x},\ol{y}}$ $~(\cket{\ol{x}},\cket{\ol{y}} \in \EJA)$
for $\braket{\cket{\ol{x}},\cket{\ol{y}}}$.

The element, denoted by $\cket{\chi}$, of $\EJA$ that satisfies
$\cket{\chi} \bullet \cket{\ol{v}} = \cket{\ol{v}}$
for any $\cket{\ol{v}} \in \EJA$ is called the \termdef{identity element} of $\EJA$.
$\cket{\chi} \in \EJA^+$ obviously holds.
The Koecher-Vinberg theorem states that the cone of squares $\EJA^+$ of
any EJA $\EJA$ is a symmetric cone;
conversely, for a given symmetric cone $\mC$ and its interior point $\cket{\chi} \in \EJA^+$,
there exists an EJA $\EJA$ with the identity element $\cket{\chi}$ that satisfies $\EJA^+ = \mC$.
Two elements $\cket{\rho},\cket{\sigma} \in \EJA^+$
are called \termdef{orthogonal} if $\braket{\rho,\sigma} = 0$ holds,
which is equivalent to $\cket{\rho} \bullet \cket{\sigma} = \cket{\emptyset}$.
$\cket{p} \in \EJA$ is called an \termdef{idempotent} if $\cket{p}^2 = \cket{p}$ holds.
One can obviously see that any idempotent is in $\EJA^+$
and that the zero element $\cket{\emptyset}$ and
the identity element $\cket{\chi}$ are idempotents.
We say that a nonzero idempotent is \termdef{primitive} if it cannot be expressed as
the sum of two nonzero idempotents.
Any idempotent can be decomposed into the sum of mutually orthogonal primitive idempotents.
A set of mutually orthogonal primitive idempotents, $\Phi \coloneqq \{ \cket{\phi_i} \}_{i=1}^n$,
is called a \termdef{Jordan frame} if it satisfies $\sum_{i=1}^n \cket{\phi_i} = \cket{\chi}$.
For any set of mutually orthogonal primitive idempotents
$\Phi \coloneqq \{ \cket{\phi_i} \}_{i=1}^k$,
there exists a Jordan frame $\Phi'$ with $\Phi' \supseteq \Phi$.
Each Jordan frame of $\EJA$ has the same number of elements, called the \termdef{rank} of $\EJA$ and
denoted by $\rank~\EJA$.
The dimension of $\EJA$ (as a real vector space) is denoted by $\dim~\EJA$.
For any $\cket{\ol{v}} \in \EJA$, there exist a Jordan frame
$\{ \cket{\phi_i} \}_{i=1}^n$ $~(n \coloneqq \rank~\EJA)$ and real numbers $c_1,\ldots,c_n \in \Real$,
called the \termdef{eigenvalues}, such that
\begin{alignat}{1}
 \cket{\ol{v}} &= \sum_{i=1}^n c_i \cket{\phi_i}.
 \label{eq:spectral_EJA}
\end{alignat}
Such a representation is called a \termdef{spectral decomposition} of $\cket{\ol{v}}$.
The eigenvalues are uniquely determined by $\cket{\ol{v}}$.
$\cket{\ol{v}}$ is in $\EJA^+$ if and only if the eigenvalues of $\cket{\ol{v}}$
are all nonnegative.
The sum of the eigenvalues of $\cket{\ol{v}}$, $\sum_{i=1}^n c_i$, is called the \termdef{trace} of
$\cket{\ol{v}}$ and denoted by $\tr~\cket{\ol{v}}$.
The number of nonzero eigenvalues of $\cket{\ol{v}}$
is called the \termdef{rank} of $\cket{\ol{v}}$
and denoted by $\rank~\cket{\ol{v}}$.
Since, from the Koecher-Vinberg theorem, $\EJA^+$ is a convex cone,
the properties of convex cones can be applied to $\EJA^+$.
For example, the partial ordering on $\EJA$ is defined as in Eq.~\eqref{eq:le}.
$\cket{\psi} \in \EJA^+$ is pure if and only if $\rank~\cket{\psi} \le 1$ holds.
$\cket{\rho} \in \EJA^+$ is completely mixed if and only if
$\rank~\cket{\rho} = \rank~\EJA$ holds.

A vector subspace of $\EJA$, $\EJA'$, is called a \termdef{subalgebra} of $\EJA$ if
$\EJA'$ is closed under the Jordan product $\bullet$,
i.e., $\cket{\ol{v}} \bullet \cket{\ol{w}} \in \EJA'$ holds
for any $\cket{\ol{v}}, \cket{\ol{w}} \in \EJA'$.
$\EJA'$ is itself an EJA.
The identity element, $\cket{\chi'}$, of $\EJA'$ is an idempotent of $\EJA$.
For any idempotent $\cket{p}$ of $\EJA$,
$\EJA_\cket{p} \coloneqq \{ \cket{\ol{v}} \in \EJA : \cket{p} \bullet \cket{\ol{v}} = \cket{\ol{v}} \}$
is the Peirce 1-subalgebra of $\EJA$ with the identity element $\cket{p}$, and
\begin{alignat}{1}
 \EJA_\cket{p}^{+} &= \EJA^+ \cap \EJA_\cket{p}
 = \{ \cket{\rho} \in \EJA^+ : \cket{p} \bullet \cket{\rho} = \cket{\rho} \}.
 \label{eq:EJA_subalgebra}
\end{alignat}
Equivalently, $\EJA_\cket{p}^{+}$ is the face of $\EJA^+$ generated by $\cket{p}$.
If $\EJA'$ is a subalgebra of $\EJA$ with identity element $\cket{\chi'}$, then
$\EJA' \subseteq \EJA_\cket{\chi'}$; equality need not hold in general.
The equalities $\EJA'=\EJA_\cket{\chi'}$ and $\EJA'{^+}=\EJA_\cket{\chi'}^+$ hold when
$\EJA' = \EJA_\cket{\chi'}$, in particular when $\EJA'$ is a direct-summand subalgebra.
Two subalgebras $\{ \cket{\emptyset} \}$ and $\EJA$ itself are called \termdef{trivial}.
Two subalgebras, $\EJA_1$ and $\EJA_2$, of $\EJA$
are said to be \termdef{orthogonal} if
$\cket{\ol{x}} \bullet \cket{\ol{y}} = \cket{\emptyset}$ holds
for any $\cket{\ol{x}} \in \EJA_1$ and $\cket{\ol{y}} \in \EJA_2$.
A necessary and sufficient condition for two subalgebras $\EJA_1$ and $\EJA_2$ to be orthogonal
is that $\cket{\rho} \bullet \cket{\sigma} = \cket{\emptyset}$
(i.e., $\braket{\rho,\sigma} = 0$) holds
for any $\cket{\rho} \in \EJA_1^+$ and $\cket{\sigma} \in \EJA_2^+$.

\subsection{Direct sum decomposition of EJAs} \label{subsec:derive_EJA_oplus}

We will introduce the direct sum decomposition of EJAs.
Let us consider $k$ ($k \ge 1$) mutually orthogonal non-trivial subalgebras,
$\EJA_1,\ldots,\EJA_k$, of $\EJA$.
We say that $\EJA$ is a \termdef{direct sum} of $\EJA_1,\ldots,\EJA_k$,
denoted by $\EJA = \bigoplus_{i=1}^k \EJA_i$,
if any $\cket{\ol{v}} \in \EJA$ can be expressed in the form
\begin{alignat}{1}
 \cket{\ol{v}} &= \sum_{i=1}^k \cket{\ol{v_i}}, \quad \cket{\ol{v_i}} \in \EJA_i.
 \label{eq:EJA_oplus_rho_vec}
\end{alignat}
In this case, we have
\begin{alignat}{2}
 \rank~\EJA &= \sum_{i=1}^k \rank~\EJA_i, &\quad
 \dim~\EJA &= \sum_{i=1}^k \dim~\EJA_i.
 \label{eq:derive_EJA_rank_dim2}
\end{alignat}
$\EJA = \bigoplus_{i=1}^k \EJA_i$ is called a
\termdef{direct sum decomposition} of $\EJA$.
$\cket{\ol{v_1}},\ldots,\cket{\ol{v_k}}$ of Eq.~\eqref{eq:EJA_oplus_rho_vec} are
uniquely determined by $\cket{\ol{v}}$;
indeed, by multiplying both sides of Eq.~\eqref{eq:EJA_oplus_rho_vec}
by the identity element of $\EJA_i$, denoted by $\cket{\chi_i}$,
we obtain $\cket{\ol{v_i}} = \cket{\chi_i} \bullet \cket{\ol{v}}$.
The identity element $\cket{\chi}$ of $\EJA$ is obviously decomposed as
$\cket{\chi} = \sum_{i=1}^k \cket{\chi_i}$.
$\cket{\ol{v}} \in \EJA^+$ holds if and only if $\cket{\ol{v_i}} \in \EJA_i^+$ holds for each $i$.
$\EJA^+ = \left( \bigoplus_{i=1}^k \EJA_i \right)^+$ is often rewritten as
$\bigoplus_{i=1}^k \EJA_i^+$,
which is referred to as a \termdef{direct sum decomposition} of $\EJA^+$.
For each nonzero pure element $\cket{\psi} \in \EJA^+$,
there exists one and only one $i \in \{1,\ldots,k\}$ such that $\cket{\psi} \in \EJA_i^+$.
The operator $\oplus$ is commutative and associative.

An EJA $\EJA$ with nonzero rank is called \termdef{simple} if
it cannot be expressed by a direct sum of two non-trivial subalgebras of $\EJA$%
\footnote{Another equivalent definition of a simple EJA is that it has no non-trivial ideals,
where a subalgebra $\EJA'$ of $\EJA$ is called an \termdef{ideal} if
$\cket{\ol{x}} \bullet \cket{\ol{y}} \in \EJA'$ holds
for any $\cket{\ol{x}} \in \EJA'$ and $\cket{\ol{y}} \in \EJA$.}.
In this case, its cone of squares $\EJA^+$ is called \termdef{irreducible}.
For a simple EJA $\EJA$, any inner product
$\braket{\ol{v},\ol{w}}$ $~(\cket{\ol{v}}, \cket{\ol{w}} \in \EJA)$
is proportional to $\tr[\cket{\ol{v}} \bullet \cket{\ol{w}}]$.
For each EJA $\EJA$, $\EJA = \bigoplus_{i=1}^k \EJA_i$
(or $\EJA^+ = \bigoplus_{i=1}^k \EJA_i^+$)
is called a \termdef{simple decomposition} if $\EJA_1,\ldots,\EJA_k$ are all simple.
In this case, the decomposition of $\cket{\ol{v}} \in \EJA$ in Eq.~\eqref{eq:EJA_oplus_rho_vec}
is also called the \termdef{simple decomposition}.
Any EJA with nonzero rank is expressed as a direct sum of simple EJAs in a unique way%
\footnote{It should be noted that the direct sum of spaces of complex Hermitian matrices,
$\tStMat \coloneqq \bigoplus_{i=1}^k \mS(\Complex^{n_i})$, is an EJA.
However, strictly speaking, according to our definition,
$\bigoplus_{i=1}^k \mS(\Complex^{n_i})$ is not a simple decomposition of $\tStMat$.
Indeed, while each element of $\tStMat$ must be a matrix of order $\sum_{i=1}^k n_i$,
each element of $\mS(\Complex^{n_i})$ is a matrix of order $n_i$,
and hence $\mS(\Complex^{n_i})$ is not a subset of $\tStMat$.
One can easily see that $\tStMat$ has the simple decomposition of the form
$\tStMat = \bigoplus_{i=1}^k \tStMat_i$,
where $\tStMat_i$ is a subalgebra of $\tStMat$ satisfying
$\tStMat_i \cong \mS(\Complex^{n_i})$.}.

It is well known that there are five types of simple EJAs:
\begin{thm}[Jordan-von Neumann-Wigner theorem \cite{Jor-Neu-Wig-1934}]
 \label{thm:Neumann} 
 Every simple EJA is isomorphic to one of the followings%
 \footnote{Since $\Spin_3 \cong \mS(\Real^2)$, $\Spin_4 \cong \mS(\Complex^2)$,
 and $\mS(\Oct^2) \cong \Spin_{10}$ hold,
 $\Spin_3, \Spin_4$, and $\mS(\Oct^2)$ were excluded to avoid overlapping.}:
 \begin{enumerate}
  \item $\mS(\Field^n)$ with $n \ge 1$ and
        $\Field \in \{\Real, \Complex, \Hermite\}$,
        where $\mS(\Field^n)$ is the set of all Hermitian matrices on the vector space $\Field^n$
        and $\Hermite$ is the set of all quaternions.
  \item Spin factors $\Spin_s$ with $s \ge 5$, where $s$ is the dimension.
  \item $\mS(\Oct^3)$, where $\Oct$ is the set of all octonions.
 \end{enumerate}
\end{thm}
Table~\ref{table:EJA} shows the ranks and dimensions of these simple EJAs.
It follows that they are classified by their ranks and dimensions.
\begin{table}[h]
 \centering
 \caption{Ranks and dimensions of simple EJAs}
 \label{table:EJA}
 \begin{tabular}{ccc}
  \hline
  Type & Rank & Dimension \\ \hline
  $\mS(\Real^n)$ & $n$ & $n(n+1)/2$ \\
  $\mS(\Complex^n)$ & $n$ & $n^2$ \\
  $\mS(\Hermite^n)$ & $n$ & $n(2n-1)$ \\
  $\Spin_s$ & $2$ & $s$ \\
  $\mS(\Oct^3)$ & $3$ & $27$ \\ \hline
 \end{tabular}
\end{table}

\subsection{State spaces as cones of squares of EJAs} \label{subsec:derive_EJA_St}

In Secs.~\ref{sec:derive_PDS} and \ref{sec:derive_symmetry},
we showed that an OPT with
the \poslink{pos:SymSharp}{symmetric sharpness},
\poslink{pos:CompletelyMixed}{complete mixing},
and \poslink{pos:Filter}{filtering} postulates
satisfies Property~(\hyperlink{property:diamondsuit}{$\diamondsuit$}).
In this subsection, we will show that the converse is almost true.
Specifically, we show that an OPT with Property~(\hyperlink{property:diamondsuit}{$\diamondsuit$})
satisfies all the properties (i.e., all lemmas, propositions, and theorems)
except Lemma~\ref{lemma:Projection} that we have derived
in Secs.~\ref{sec:derive_PDS} and \ref{sec:derive_symmetry}.
We here assume that Property~(\hyperlink{property:diamondsuit}{$\diamondsuit$}) holds and
do not assume that the three postulates hold.

We begin with some preliminaries.
\begin{lemma} \label{lemma:DiamondEJA}
 If Property~(\hyperlink{property:diamondsuit}{$\diamondsuit$}) holds, then,
 for any system $A$, there exists an EJA $\EJA_A$ with $\EJA_A^+ = \St_A$ such that
 \begin{alignat}{2}
  \craket{\gdis|\rho} &= \tr~\cket{\rho},
  &\quad \forall \cket{\rho} &\in \St_A.
  \label{eq:derive_EJA_normal_unittrace}
 \end{alignat}
\end{lemma}
\begin{proof}
 $\St_A$ is self-dual with respect to some inner product $\braket{~,~}'$.
 Since $\Eff_A$ is the dual cone of $\St_A$ and $\cra{\gdis}$ is completely mixed,
 there exists a completely mixed state $\cket{\chi'} \in \St_A$
 such that $\craket{\gdis|\ol{x}} = \braket{\chi',\ol{x}}'$
 for any $\cket{\ol{x}} \in \Vec_A$.

 First, we consider the EJA $\EJA'_A$ with the identity element $\cket{\chi'}$
 that satisfies $\EJA'_A{}^+ = \St_A$.
 From the Koecher-Vinberg theorem, such an EJA exists.
 Let $\bullet'$ and $\tr'$ be, respectively, the Jordan product and the trace in $\EJA'_A$.
 $\EJA'_A$ has the simple decomposition of the form $\EJA'_A = \bigoplus_{i=1}^k \EJA'_{(i)}$.
 Arbitrarily choose $\cket{\ol{x}},\cket{\ol{y}} \in \Vec_A$, which have
 the simple decompositions $\cket{\ol{x}} = \sum_{i=1}^k \cket{\ol{x_i}}$ and
 $\cket{\ol{y}} = \sum_{i=1}^k \cket{\ol{y_i}}$
 with $\cket{\ol{x_i}},\cket{\ol{y_i}} \in \EJA'_{(i)}$ $~(\forall i \in \{1,\ldots,k\})$.
 Since $\braket{\ol{x_i},\ol{y_j}}' = 0$ holds for any distinct $i$ and $j$,
 $\braket{\ol{x},\ol{y}}' = \sum_{i=1}^k \braket{\ol{x_i},\ol{y_i}}'$ holds.
 Since, for each $i$, $\braket{\ol{x_i},\ol{y_i}}'$ is proportional to
 $\tr'[\cket{\ol{x_i}} \bullet' \cket{\ol{y_i}}]$,
 \begin{alignat}{1}
  \braket{\ol{x},\ol{y}}' &= \sum_{i=1}^k c_i \cdot \tr'[\cket{\ol{x_i}} \bullet' \cket{\ol{y_i}}]
 \end{alignat}
 holds with some constants $c_1,\ldots,c_k \in \Realpp$.

 Next, we consider the EJA, $\EJA_A$, with the Jordan product $\bullet$
 defined as
 \begin{alignat}{2}
  \cket{\ol{x}} \bullet \cket{\ol{y}} &\coloneqq
  \sum_{i=1}^k c_i \cket{\ol{x_i}} \bullet' \cket{\ol{y_i}},
  &\quad &\forall \cket{\ol{x}}, \cket{\ol{y}} \in \EJA_A,
  \label{eq:derive_DiamondEJA_bullet}
 \end{alignat}
 where $\cket{\ol{x}}$ and $\cket{\ol{y}}$ have the simple decompositions
 $\cket{\ol{x}} = \sum_{i=1}^k \cket{\ol{x_i}}$ and
 $\cket{\ol{y}} = \sum_{i=1}^k \cket{\ol{y_i}}$
 with $\cket{\ol{x_i}},\cket{\ol{y_i}} \in \EJA'_{(i)}$ $~(i \in \{1,\ldots,k\})$.
 One can easily verify that $\EJA_A$ has the simple decomposition of the form
 $\EJA_A = \bigoplus_{i=1}^k \EJA_{(i)}$
 with $\EJA_{(i)}^+ = \EJA'_{(i)}{}^+$.
 $\EJA_A^+ = \St_A = \EJA'_A{}^+$ obviously holds.
 Note that since $\EJA_{(i)}$ and $\EJA'_{(i)}$ are equal as real Hilbert spaces,
 $\cket{\ol{v}}$ is in $\EJA_{(i)}$ if and only if it is in $\EJA'_{(i)}$.
 For each $i \in \{1,\ldots,k\}$, let $\cket{\chi'_i}$ denote the identity element of $\EJA'_{(i)}$;
 then, it follows from Eq.~\eqref{eq:derive_DiamondEJA_bullet} that
 $c_i^{-1} \cket{\chi'_i}$ is the identity element of $\EJA_{(i)}$.
 Thus, $\cket{\chi} \coloneqq \sum_{i=1}^k c_i^{-1} \cket{\chi'_i}$
 is the identity element of $\EJA_A$.
 Arbitrarily choose a Jordan frame $\{ \cket{\phi'_s} \}_{s=1}^n$ of $\EJA'_A$,
 where $n \coloneqq \rank~\EJA'_A$.
 $\cket{\phi'_s}$ is obviously a pure state.
 Let $\cket{\phi_s} \coloneqq c_{i_s}^{-1} \cket{\phi'_s}$,
 where $i_s \in \{1,\ldots,k\}$ is the value satisfying $\cket{\phi'_s} \in \EJA'_{(i_s)}$.
 Then, since $\sum_{s=1}^n \cket{\phi_s} = \sum_{i=1}^k c_i^{-1} \cket{\chi'_i} = \cket{\chi}$
 and $\cket{\phi_s} \bullet \cket{\phi_s}
 = c_{i_s}^{-2} \cket{\phi'_s} \bullet \cket{\phi'_s}
 = c_{i_s}^{-1} \cket{\phi'_s} \bullet' \cket{\phi'_s} = \cket{\phi_s}$ hold,
 $\{ \cket{\phi_s} \}_{s=1}^n$ is a Jordan frame of $\EJA_A$.
 Let $\tr$ denote the trace in $\EJA_A$; then,
 any primitive idempotent $\cket{\phi_s} \in \EJA_{(i_s)}$ of $\EJA_A$
 satisfies $\tr~\cket{\phi_s} = 1 = \tr'~\cket{\phi'_s} = c_{i_s} \cdot \tr'~\cket{\phi_s}$.
 Arbitrarily choose $\cket{\ol{x}} \in \EJA_A$,
 which has the simple decomposition $\cket{\ol{x}} = \sum_{i=1}^k \cket{\ol{x_i}}$
 with $\cket{\ol{x_i}} \in \EJA_{(i)}$.
 Then, we have
 \begin{alignat}{3}
  \tr~\cket{\ol{x}} &= \sum_{i=1}^k \tr~\cket{\ol{x_i}}
  &&= \sum_{i=1}^k c_i \cdot \tr'~\cket{\ol{x_i}}
  \label{eq:derive_EJAA_tr}
 \end{alignat}
 and thus
 \begin{alignat}{3}
  \craket{\gdis|\ol{x}} &= \braket{\chi',\ol{x}}'
  &&= \sum_{i=1}^k c_i \cdot \tr'[\cket{\chi'_i} \bullet' \cket{\ol{x_i}}]
  \nonumber \\
  & &&= \sum_{i=1}^k c_i \cdot \tr'~\cket{\ol{x_i}}
  = \tr~\cket{\ol{x}}.
 \end{alignat}
 Therefore, Eq.~\eqref{eq:derive_EJA_normal_unittrace} holds.
\end{proof}

For each system $A$, we will choose an EJA $\EJA_A$ satisfying $\EJA_A^+ = \St_A$ and
Eq.~\eqref{eq:derive_EJA_normal_unittrace}.
We here present some basic properties of the EJA $\EJA_A$.
From Eq.~\eqref{eq:derive_EJA_normal_unittrace}, for $\cket{\rho} \in \St_A$,
$\cket{\rho} \in \StN_A$ and $\tr~\cket{\rho} = 1$ are obviously equivalent.
Let $\braket{~,~}$ be the inner product of $\EJA_A$ defined as
$\braket{\ol{x},\ol{y}} \coloneqq \tr[\cket{\ol{x}} \bullet \cket{\ol{y}}]$;
then, we have
\begin{alignat}{3}
 \craket{\gdis|\ol{x}} &= \tr~\cket{\ol{x}} &&= \braket{\chi,\ol{x}},
 &\quad &\forall \cket{\ol{x}} \in \EJA_A.
 \label{eq:derive_EJAA_braket}
\end{alignat}
We will equip $\Vec_A$ with the Jordan product $\bullet$ and the inner product $\braket{~,~}$
to identify $\Vec_A$ with the EJA $\EJA_A$.
Any extended state in $\Vec_A$ has a spectral decomposition.
Since a state $\cket{\psi}$ is pure if and only if $\rank~\cket{\psi} \le 1$ holds,
$\cket{\psi} \in \StNP_A$ holds if and only if $\cket{\psi}$ is a primitive idempotent.
Since $\Eff_A$ is a symmetric cone, by equipping $\Vec_A^*$ with an appropriate Jordan product
(and an inner product), we can regard $\Vec_A^*$ as an EJA with
the identity element $\cra{\gdis}$ that satisfies $\Vec_A^*{}^+ = \Eff_A$.
We define the isomorphism
$\dagger:\Vec_A \ni \cket{\ol{x}} \to \cra{\ol{x}^\dagger} \in \Vec_A^*$ to satisfy
\begin{alignat}{2}
 \craket{\ol{x}^\dagger|\ol{y}} &= \braket{\ol{x},\ol{y}}
 = \tr[\cket{\ol{x}} \bullet \cket{\ol{y}}],
 &\quad &\forall \cket{\ol{x}},\cket{\ol{y}} \in \EJA_A.
\end{alignat}
Then, it follows from Eq.~\eqref{eq:derive_EJAA_braket} that
$\cra{\gdis} = \cra{\chi^\dagger}$ holds.
Also, we use the same notation, $\dagger:\Vec_A^* \to \Vec_A$,
for the inverse of $\dagger:\Vec_A \to \Vec_A^*$.

\begin{ex}[quantum theory]
 The identity element of $\Vec_A$ is $\cket{\chi_A} = \ident_\NA$.
 The Jordan product of $\cket{\ol{v}}, \cket{\ol{w}} \in \Vec_A$ is
 $\cket{\ol{v}} \bullet \cket{\ol{w}} =
 [\cket{\ol{v}} \cdot \cket{\ol{w}} + \cket{\ol{w}} \cdot \cket{\ol{v}}] / 2$.
 $\tr~\cket{\ol{v}}$ is equal to the trace of the matrix $\cket{\ol{v}}$,
 i.e., $\Tr~\cket{\ol{v}}$.
 $\braket{\ol{v},\ol{w}} = \Tr[\cket{\ol{v}} \bullet \cket{\ol{w}}]
 = \Tr[\cket{\ol{v}} \cdot \cket{\ol{w}}]$ holds.
\end{ex}

\begin{lemma} \label{lemma:DiamondFrame}
 Assume that Property~(\hyperlink{property:diamondsuit}{$\diamondsuit$}) holds.
 Let $\EJA_A$ be an EJA satisfying $\EJA_A^+ = \St_A$ and
 Eq.~\eqref{eq:derive_EJA_normal_unittrace}.
 Then, a set of normalized pure states is an MPDS
 if and only if it is a Jordan frame.
\end{lemma}
\begin{proof}
 Consider a set of normalized pure states, $\Phi \coloneqq \{ \cket{\phi_i} \}_{i=1}^k$.
 Let $n \coloneqq \rank~\EJA_A$.

 ``If'':
 Assume that $\Phi$ is a Jordan frame; in this case, $k = n$ holds.
 Let $\Pi \coloneqq \{ \cra{\phi_i^\dagger} \}_{i=1}^n$; then,
 from $\sum_{i=1}^n \cra{\phi_i^\dagger} = \left[ \sum_{i=1}^n \cket{\phi_i} \right]^\dagger
 = \cket{\chi}^\dagger = \cra{\gdis}$ and
 $\craket{\phi_i^\dagger|\phi_i} = 1 = \craket{\gdis|\phi_i}$ $~(\forall i \in \{1,\ldots,n\})$,
 $\Pi$ is a measurement that perfectly distinguishes between $\Phi$.
 Also, since $\sum_{i=1}^n \cket{\phi_i} = \cket{\chi}$ is completely mixed,
 any maximal effect, $\cra{e}$, satisfies $\craket{e|\chi} > 0$.
 Thus, there exists no normalized pure state that is
 perfectly distinguishable from every state in $\Phi$.
 Therefore, $\Phi$ is an MPDS.

 ``Only if'':
 Assume that $\Phi$ is an MPDS.
 Let $\{ \cra{e_i} \}_{i=1}^k$ be a maximal measurement that
 perfectly distinguishes between $\Phi$.
 For each $i$, $\cket{a_i} \coloneqq \cket{\chi} - \cket{e_i^\dagger}$ has
 a spectral decomposition of the form $\cket{a_i} = \sum_{j=1}^l c_j \cket{\psi_j}$
 where $c_1,\ldots,c_l$ are in $\Realpp$ and $\{ \cket{\psi_j} \}_{j=1}^l$ is
 a subset of some Jordan frame.
 From $\braket{a_i,\phi_i} = \craket{\gdis|\phi_i} - \craket{e_i|\phi_i} = 0$,
 $\cket{\psi_1},\ldots,\cket{\psi_l}$ are orthogonal to $\cket{\phi_i} \eqqcolon \cket{\psi_{l+1}}$.
 Since $\cket{\psi_1},\ldots,\cket{\psi_{l+1}}$ are mutually orthogonal primitive idempotents,
 there exists a Jordan frame $\{ \cket{\psi_i} \}_{i=1}^n$ that includes
 $\cket{\psi_1},\ldots,\cket{\psi_{l+1}}$.
 This yields
 \begin{alignat}{4}
  \cket{e_i^\dagger} &= \cket{\chi} - \cket{a_i} &&\ge \sum_{j=l+1}^n \cket{\psi_j}
  &&\ge \cket{\psi_{l+1}} &&= \cket{\phi_i}.
 \end{alignat}
 Since $\braket{\phi_i,\phi_j} \le \craket{e_i|\phi_j} = 0$ holds
 for any distinct $i,j \in \{1,\ldots,k\}$,
 $\Phi$ is a set of mutually orthogonal primitive idempotents.
 Assume, by contradiction, that $\Phi$ is not a Jordan frame; then,
 there exists a Jordan frame $\Phi_\mathrm{ex}$ that truly includes $\Phi$.
 From the proof of the ``if'' part, the Jordan frame $\Phi_\mathrm{ex}$ is an MPDS,
 which contradicts that $\Phi$ is an MPDS.
 Therefore, $\Phi$ is a Jordan frame.
\end{proof}

It is easily seen from the proof of this lemma that
$\cket{\rho} \in \St_A$ and $\cket{\sigma} \in \St_A$ are
perfectly distinguishable (i.e., $\cket{\rho} \in \ker_\cket{\sigma}$)
if and only if $\braket{\rho,\sigma} = 0$ holds.
$\NA = \rank~\EJA_A$ holds and any MPDS has $\NA$ elements.
It also follows that, for any MPDS $\Phi \coloneqq \{ \cket{\phi_i} \}_{i=1}^\NA$,
$\Pi \coloneqq \{ \cra{\phi_i^\dagger} \}_{i=1}^\NA$ is the unique measurement
that perfectly distinguishes between $\Phi$.
Eq.~\eqref{eq:spectral_vec} obviously holds from Eq.~\eqref{eq:spectral_EJA}.
The identity element $\cket{\chi}$ is the invariant state.
Any PDS $\Phi \coloneqq \{ \cket{\phi_i} \}_{i=1}^k$ satisfies
$\braket{\phi_i,\phi_j} = \delta_{i,j}$, and thus $\cket{\chi_\Phi}$ is an idempotent.
Any maximal effect is pure and expressed in the form $\cra{\psi^\dagger}$
with $\cket{\psi} \in \StNP_A$,
which satisfies $\craket{\psi^\dagger|\psi} = 1$.

With the above preliminaries, we can easily show that
an OPT with Property~(\hyperlink{property:diamondsuit}{$\diamondsuit$})
satisfies all the properties except Lemma~\ref{lemma:Projection} that we have derived
in Secs.~\ref{sec:derive_PDS} and \ref{sec:derive_symmetry}.
To show that all the properties derived in Sec.~\ref{sec:derive_PDS} hold,
it suffices to show that the \poslink{pos:SymSharp}{symmetric sharpness} and
\poslink{pos:CompletelyMixed}{complete mixing} postulates hold.
As for the \poslink{pos:SymSharp}{symmetric sharpness postulate},
it is obvious that, to each $\cket{\phi} \in \StNP_A$,
there corresponds one and only one maximal effect, $\cra{\phi^\dagger}$,
giving unit probability
and that $\craket{\varphi^\dagger|\psi} = \braket{\varphi,\psi} = \braket{\psi,\varphi}
= \craket{\psi^\dagger|\varphi}$ holds
for any $\cket{\varphi},\cket{\psi} \in \StNP_A$.
As for the \poslink{pos:CompletelyMixed}{complete mixing postulate},
since any $\cket{\rho} \in \St_A$ that is not completely mixed
has a spectral decomposition in the form $\cket{\rho} = \sum_{i=1}^\NA p_i \cket{\phi_i}$
with $p_1,\ldots,p_\NA \in \Real_+$, $p_\NA = 0$, and
$\{ \cket{\phi_i} \}_{i=1}^\NA \in \MPDS_A$,
the maximal effect $\cra{\phi_\NA^\dagger}$ satisfies
$\craket{\phi_\NA^\dagger|\rho} = \braket{\phi_\NA,\rho} = 0$.
All the properties except Lemma~\ref{lemma:Projection} in Sec.~\ref{sec:derive_symmetry}
are immediately obtained by
Property~(\hyperlink{property:diamondsuit}{$\diamondsuit$}).

\begin{lemma} \label{lemma:FaceBullet} 
 For any $\Phi \in \PDS_A$,
 let $\EJA_\Phi$ be the subalgebra of $\EJA_A$ with the identity element $\cket{\chi_\Phi}$,
 i.e., $\EJA_\Phi \coloneqq \{ \cket{\ol{v}} \in \EJA_A :
 \cket{\chi_\Phi} \bullet \cket{\ol{v}} = \cket{\ol{v}} \}$.
 Then, we have $\Face_\Phi = \EJA_\Phi^+$.
\end{lemma}
\begin{proof}
 Let $\Psi$ be a PDS complementary to $\Phi$; then, we have
 \begin{alignat}{1}
  \cket{\rho} \in \Face_\Phi
  &\quad\Leftrightarrow\quad \braket{\chi_\Psi,\rho} = 0 \nonumber \\
  &\quad\Leftrightarrow\quad \cket{\chi_\Psi} \bullet \cket{\rho} = \cket{\emptyset} \nonumber \\
  &\quad\Leftrightarrow\quad \cket{\chi_\Phi} \bullet \cket{\rho} = \cket{\rho} \nonumber \\
  &\quad\Leftrightarrow\quad \cket{\rho} \in \EJA_\Phi^+,
 \end{alignat}
 where the first line follows from $\Face_\Phi = \ker_\Psi$.
 The third line follows from $\cket{\chi_\Psi} + \cket{\chi_\Phi} = \cket{\chi}$ and
 $\cket{\chi} \bullet \cket{\rho} = \cket{\rho}$,
 and the last line follows from
 $\EJA_\Phi^+ = \{ \cket{\sigma} \in \EJA_A^+ : \cket{\chi_\Phi} \bullet \cket{\sigma} = \cket{\sigma} \}$.
 Thus, we have $\Face_\Phi = \EJA_\Phi^+$.
\end{proof}

\section{Derivation of quantum theory} \label{sec:derive_Quantum}

In this section, we consider an OPT having
Property~(\hyperlink{property:diamondsuit}{$\diamondsuit$})
and the \poslink{pos:ProcEq}{local equality postulate}.
We investigated the properties of individual systems in the previous sections;
we will investigate the structure of composite systems in this section.

We first derive the following properties:
\begin{enumerate}
 \item $D_{A \otimes B} = D_A D_B$ holds (Lemma~\ref{lemma:DAB}).
 \item $N_{A \otimes B} = \NA \NB$ holds (Lemma~\ref{lemma:NAB}).
 \item The operation $\dagger$ distributes over $\otimes$ (Lemma~\ref{lemma:DaggerOtimes}).
 \item When $\St_A$ and $\St_B$ have simple decompositions
       $\St_A = \bigoplus_{i=1}^{k_A} \EJA_{A,i}^+$ and
       $\St_B = \bigoplus_{j=1}^{k_B} \EJA_{B,j}^+$,
       $\St_{A \otimes B}$ has the simple decomposition
       $\St_{A \otimes B} = \bigoplus_{i=1}^{k_A} \bigoplus_{j=1}^{k_B}
       \EJA_{A,i}^+ \otimes \EJA_{B,j}^+$
       (Lemmas~\ref{lemma:OplusOtimes} and \ref{lemma:OtimesSimple}),
       where $\EJA_{A,i}^+ \otimes \EJA_{B,j}^+$ will be defined in
       Subsec.~\ref{subsec:derive_Quantum_otimes}.
\end{enumerate}
Subsequently, we derive
Properties~(\hyperlink{property:A}{A})--(\hyperlink{property:C}{C}) listed
in Subsec.~\ref{subsec:derive_abstract_stream}
based on the above stated properties.

\subsection{Results about local equality} \label{subsec:derive_Quantum_Process}

In this subsection,
we show that the following properties are derived only from
the \poslink{pos:ProcEq}{local equality postulate}
(without considering Property~(\hyperlink{property:diamondsuit}{$\diamondsuit$})):
\begin{enumerate}[label=(\roman*)]
 \item $D_{A \otimes B} = D_A D_B$ holds.
 \item The parallel composition of two pure states is pure,
       and the parallel composition of two pure effects is pure.
\end{enumerate}
%


\subsubsection{$D_{A \otimes B} = D_A D_B$}

Let us fix a basis of $\Vec_A$, $\{ \cket{\ol{w_i}} \}_{i=1}^{D_A}$.
Since $\Vec_A^*$ is the dual vector space of $\Vec_A$,
there exists a basis, $\{ \cra{\ol{v_i}} \}_{i=1}^{D_A}$, of $\Vec_A^*$
satisfying $\craket{\ol{v_i}|\ol{w_j}} = \delta_{i,j}$.
Also, let $\cket{\ol{\eta}_A} \in \Vec_{A \otimes A}$ and
$\cra{\ol{\varepsilon}_A} \in \Vec_{A \otimes A}^*$ be defined by
\begin{alignat}{2}
 \cket{\ol{\eta}_A} &\coloneqq \sum_{i=1}^{D_A} \cket{\ol{w_i}} \otimes \cket{\ol{w_i}},
 &\quad
 \cra{\ol{\varepsilon}_A} &\coloneqq \sum_{i=1}^{D_A} \cra{\ol{v_i}} \otimes \cra{\ol{v_i}}.
 \label{eq:derive_quant_etaepsilon}
\end{alignat}

\begin{ex}[fully quantum theory]
 Let $\{ \ket{s} \}_{s=1}^\NA$ be an ONB of the complex Hilbert space $\Complex^\NA$.
 Also, let $\cket{\ol{w_{s,t}}} \in \Vec_A$ be defined as
 \begin{alignat}{1}
  \cket{\ol{w_{s,t}}} &\coloneqq
  \begin{dcases}
   \ket{s}\bra{s}, & s = t, \\
   \frac{1}{\sqrt{2}}(\ket{s}\bra{t} + \ket{t}\bra{s}), & s < t, \\
   \frac{\i}{\sqrt{2}}(\ket{s}\bra{t} - \ket{t}\bra{s}), & s > t \\
  \end{dcases}
  \label{eq:derive_quant_cons}
 \end{alignat}
 and $\cra{\ol{w_{s,t}}^\dagger} \in \Vec_A^*$ be the same matrix as $\cket{\ol{w_{s,t}}}$,
 where $\i \coloneqq \sqrt{-1}$.
 Then, $\mB \coloneqq \{ \cket{\ol{w_{s,t}}} \}_{(s,t)=(1,1)}^{(\NA,\NA)}$
 and $\mB^* \coloneqq \{ \cra{\ol{w_{s,t}}^\dagger} \}_{(s,t)=(1,1)}^{(\NA,\NA)}$
 are, respectively, ONBs of $\Vec_A$ and $\Vec_A^*$
 $\craket{\ol{w_{s,t}}^\dagger|\ol{w_{s',t'}}} = \delta_{s,s'} \delta_{t,t'}$ obviously holds.
 Consider the case $\NA = 2$; in this case, $D_A = 4$ holds.
 Substituting $\mB$ and $\mB^*$, respectively, into
 $\{ \cket{\ol{w_i}} \}_{i=1}^4$ and $\{ \cra{\ol{v_i}} \}_{i=1}^4$
 in Eq.~\eqref{eq:derive_quant_etaepsilon} gives
 \begin{alignat}{2}
  \cket{\ol{\eta}_A} &=
  \begin{bmatrix}
   1 & 0 & 0 & 0 \\
   0 & 0 & 1 & 0 \\
   0 & 1 & 0 & 0 \\
   0 & 0 & 0 & 1 \\
  \end{bmatrix}
  &&= \cra{\ol{\varepsilon}_A}.
  \label{eq:derive_eta_quant}
 \end{alignat}
 Note that since this matrix is not positive semidefinite,
 $\cket{\ol{\eta}_A} \not\in \St_{A \otimes A}$ and
 $\cra{\ol{\varepsilon}_A} \not\in \Eff_{A \otimes A}$ hold.
\end{ex}

\begin{lemma} \label{lemma:EtaEpsilon} 
 \begin{alignat}{1}
  \includegraphics[scale=1.0]{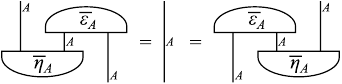} ~\raisebox{.5em}{.}
  \label{eq:derive_varepsilon_eta_id}
 \end{alignat}
\end{lemma}
\begin{proof}
 From the \poslink{pos:ProcEq}{local equality postulate} and Eq.~\eqref{eq:eqlocal_frho},
 it suffices to show, for any $\cket{\rho} \in \St_A$,
 \begin{alignat}{1}
  \includegraphics[scale=1.0]{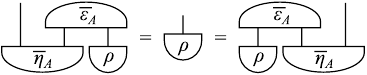} ~\raisebox{.5em}{.}
  \label{eq:derive_varepsilon_eta_id_proof}
 \end{alignat}
 Since $\cket{\rho}$ can be expressed in the form
 $\cket{\rho} = \sum_{i=1}^{D_A} \rho_i \cket{\ol{w_i}}$ with $\rho_i \in \Real$,
 we have
 \begin{alignat}{1}
  \includegraphics[scale=1.0]{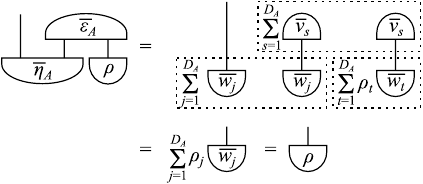} ~\raisebox{.5em}{,}
  \label{eq:derive_varepsilon_eta_id_proof2}
 \end{alignat}
 where the first and second equalities follow from Eq.~\eqref{eq:derive_quant_etaepsilon}
 and $\craket{\ol{v_i}|\ol{w_{i'}}} = \delta_{i,{i'}}$, respectively.
 This proves the first equality of Eq.~\eqref{eq:derive_varepsilon_eta_id_proof}.
 The second equality of Eq.~\eqref{eq:derive_varepsilon_eta_id_proof} can be proved
 in the same way.
\end{proof}

\begin{lemma} \label{lemma:DAB} 
 $D_{A \otimes B} = D_A D_B$ holds for any systems $A$ and $B$.
\end{lemma}
\begin{proof}
 From the \poslink{pos:ProcEq}{local equality postulate},
 each extended process $\ol{f} \in \Vec_{A \to B}$ is identified by
 a set of scalars $\{ \cra{e_j} \circ \ol{f} \circ \cket{\rho_i} \}_{(i,j)=(1,1)}^{(D_A,D_B)}$,
 where $\{ \cket{\rho_i} \in \St_A \}_{i=1}^{D_A}$ and $\{ \cra{e_j} \in \Eff_B \}_{j=1}^{D_B}$ are,
 respectively, sets of some fixed states and effects.
 Thus, we have $\dim~\Vec_{A \to B} \le D_A D_B$.
 One can also see $\dim~\Vec_{A \otimes B} \ge D_A D_B$.
 Indeed, if $\{ \cket{\ol{w_i}} \}_{i=1}^{D_A}$ and $\{ \cket{\ol{w'_j}} \}_{j=1}^{D_B}$
 are, respectively, bases of $\Vec_A$ and $\Vec_B$, then
 $\{ \cket{\ol{w_i}} \otimes \cket{\ol{w'_j}} \}_{(i,j)=(1,1)}^{(D_A,D_B)}$
 is a set of linearly independent extended states of $\Vec_{A \otimes B}$.
 We only need to show $\dim~\Vec_{A \otimes B} \le \dim~\Vec_{A \to B}$;
 indeed, in this case, we have
 $D_A D_B \le \dim~\Vec_{A \otimes B} \le \dim~\Vec_{A \to B} \le D_A D_B$
 and thus $D_{A \otimes B} = \dim~\Vec_{A \otimes B} = D_A D_B$.

 We will show $\dim~\Vec_{A \otimes B} \le \dim~\Vec_{A \to B}$.
 Consider the linear map $F:\Vec_{A \to B} \ni \ol{f} \mapsto
 (\id_A \otimes \ol{f}) \circ \cket{\ol{\eta}_A} \in \Vec_{A \otimes B}$,
 i.e.,
 \begin{alignat}{1}
  \includegraphics[scale=1.0]{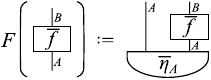} ~\raisebox{.5em}{,}
  \label{eq:derive_ProcSt_proc2st}
 \end{alignat}
 and the linear map $G:\Vec_{A \otimes B} \ni \cket{\ol{x}} \mapsto
 [\cra{\ol{\varepsilon}_A} \otimes \id_B] \circ [\id_A \otimes \cket{\ol{x}}] \in \Vec_{A \to B}$,
 i.e.,
 \begin{alignat}{1}
  \includegraphics[scale=1.0]{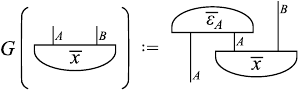} ~\raisebox{.5em}{.}
  \label{eq:derive_ProcSt_st2proc}
 \end{alignat}
 We have that, for any $\cket{\ol{x}} \in \Vec_{A \otimes B}$,
 \begin{alignat}{1}
  \includegraphics[scale=1.0]{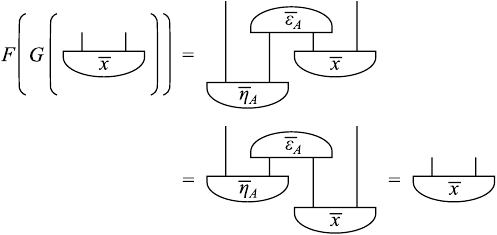} ~\raisebox{1.3em}{,}
  \label{eq:derive_ProcSt_FG}
 \end{alignat}
 where the second equality follows from
 $[\id_{A \otimes A} \circ \cket{\ol{\eta_A}}] \otimes [\cket{\ol{x}} \circ \id_I]
 = [\cket{\ol{\eta_A}} \circ \id_I] \otimes [\id_{A \otimes B} \circ \cket{\ol{x}}]$
 (see Eq.~\eqref{eq:process_fg})
 and the last equality follows from Lemma~\ref{lemma:EtaEpsilon}.
 Thus, $F[G(\Endash)]$ is the identity map on $\Vec_{A \otimes B}$,
 which gives that $G$ must be injective.
 Therefore, $\dim~\Vec_{A \otimes B} \le \dim~\Vec_{A \to B}$ holds.
\end{proof}

It follows from this lemma that $\Vec_{A \to B}$ and $\Vec_{A \otimes B}$
are isomorphic as real vector spaces for any systems $A$ and $B$.
One can easily see that $F$ and $G$ defined above are isomorphisms between these vector spaces.
Note that if $\Vec_{A \otimes B}$ is an EJA with a Jordan product $\bullet$,
then $\Vec_{A \to B}$ equipped with the Jordan product $\bullet_\mathrm{p}$
that is defined by $\ol{f} \bullet_\mathrm{p} \ol{g} \coloneqq G[F(\ol{f}) \bullet F(\ol{g})]$
$~(\ol{f},\ol{g} \in \Vec_{A \to B})$ is also an EJA.
In this case, $F$ and $G$ are also isomorphisms between these EJAs.

It follows from this lemma that,
for two bases $\{ \cket{\ol{w_i}} \}_{i=1}^{D_A}$ of $\Vec_A$
and $\{ \cket{\ol{w'_j}} \}_{j=1}^{D_B}$ of $\Vec_B$,
$\{ \cket{\ol{w_i}} \otimes \cket{\ol{w'_j}} \}_{(i,j)=(1,1)}^{(D_A,D_B)}$
is a basis of $\Vec_{A \otimes B}$.
Thus, any $\cket{\ol{x}} \in \Vec_{A \otimes B}$ is expressed in the form
\begin{alignat}{2}
 \cket{\ol{x}} &= \sum_{i=1}^{D_A} \sum_{j=1}^{D_B} c_{i,j} \cket{\ol{w_i}} \otimes \cket{\ol{w'_j}},
 &&\quad c_{i,j} \in \Real.
\end{alignat}
Also, one can easily see that,
for any $\cket{\rho},\cket{\rho'} \in \St_{A \otimes B}$,
we have
\begin{alignat}{1}
 \includegraphics[scale=1.0]{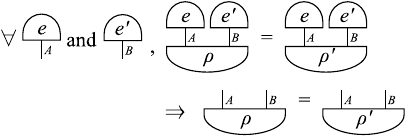} ~\raisebox{.5em}{,}
 \label{eq:derive_local_discriminability_St}
\end{alignat}
which indicates that any bipartite state can be identified from the statistics of
local measurements on the individual systems.
This property, often referred to as local tomography or local distinguishability,
has been discussed since at least the 1980s
(see, e.g., \cite{Ber-Can-Cor-Liv-1980,Woo-1990,Mer-1998}).
It is known that local tomography is equivalent to the relation $D_{A \otimes B} = D_A D_B$
and that, in any OPT that satisfies local tomography,
the \poslink{pos:ProcEq}{local equality postulate} holds \cite{Chi-Dar-Per-2010}.
Thus, from Lemma~\ref{lemma:DAB}, the \poslink{pos:ProcEq}{local equality postulate},
local tomography, and the relation $D_{A \otimes B} = D_A D_B$ are all equivalent.

\subsubsection{Composition of two pure states/effects is pure}

\begin{lemma} \label{lemma:PureOtimes} 
 The parallel composition of two pure states is pure.
 Also, the parallel composition of two pure effects is pure.
\end{lemma}
\begin{proof}
 Since the case of effects can be treated similarly,
 we only prove the case of states.
 Arbitrarily choose $\cket{\psi} \in \StP_A$ and $\cket{\psi'} \in \StP_B$.
 We can express $\cket{\psi} \otimes \cket{\psi'}$ in the following form:
 \begin{alignat}{2}
  \cket{\psi} \otimes \cket{\psi'} &= \sum_{i=1}^l \cket{\varphi_i},
  &\quad \cket{\varphi_i} &\in \StP_{A \otimes B}.
  \label{eq:derive_DAB_sigma}
 \end{alignat}
 To prove $\cket{\psi} \otimes \cket{\psi'} \in \StP_{A \otimes B}$,
 we only need to show that $\cket{\varphi_1} \propto \cket{\psi} \otimes \cket{\psi'}$
 always holds.

 The case $\cket{\psi'} = \cket{\emptyset}$ is obvious;
 assume $\cket{\psi'} \neq \cket{\emptyset}$.
 Arbitrarily choose $\cra{e} \in \Eff_A$.
 Applying $\cra{e} \otimes \id_B$ to Eq.~\eqref{eq:derive_DAB_sigma}
 and using $\cket{\psi'} \in \StP_B$, one can see that
 \begin{alignat}{1}
  \includegraphics[scale=1.0]{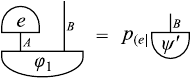}
  \label{eq:derive_pure_otimes_proof1}
 \end{alignat}
 holds for some $p_\cra{e} \in \Real_+$.
 In the same way, applying $\id_A \otimes \cra{\gdis_B}$ to Eq.~\eqref{eq:derive_DAB_sigma}
 and using $\cket{\psi} \in \StP_A$, one can see that
 \begin{alignat}{1}
  \includegraphics[scale=1.0]{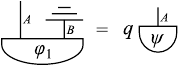}
  \label{eq:derive_pure_otimes_proof2}
 \end{alignat}
 holds for some $q \in \Real_+$.
 From Eqs.~\eqref{eq:derive_pure_otimes_proof1} and \eqref{eq:derive_pure_otimes_proof2},
 we have
 \begin{alignat}{2}
  \includegraphics[scale=1.0]{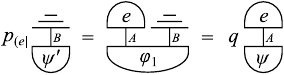} ~\raisebox{.5em}{.}
 \end{alignat}
 Thus, Eq.~\eqref{eq:derive_pure_otimes_proof1} gives
 \begin{alignat}{2}
  \includegraphics[scale=1.0]{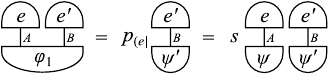}
 \end{alignat}
 for any $\cra{e'} \in \Eff_B$, where $s \coloneqq q \cdot \craket{\gdis_B|\psi'}^{-1}$.
 Since this equation holds for any effects $\cra{e}$ and $\cra{e'}$,
 $\cket{\varphi_1} = s \cket{\psi} \otimes \cket{\psi'}$ holds
 from Eq.~\eqref{eq:derive_local_discriminability_St}.
 Therefore, we have $\cket{\varphi_1} \propto \cket{\psi} \otimes \cket{\psi'}$.
\end{proof}

\subsection{Basic properties of composite systems} \label{subsec:derive_Quantum_basic}

In what follows, we consider an OPT
enjoying Property~(\hyperlink{property:diamondsuit}{$\diamondsuit$}) and
the \poslink{pos:ProcEq}{local equality postulate}.
For any system $A$, let $\EJA_A$ be an EJA satisfying $\EJA_A^+ = \St_A$ and
Eq.~\eqref{eq:derive_EJA_normal_unittrace}.

\begin{lemma} \label{lemma:CompositeMPDS} 
 For any $\{ \cket{\phi_i} \}_{i=1}^\NA \in \MPDS_A$ and
 $\{ \cket{\varphi_j} \}_{j=1}^\NB \in \MPDS_B$,
 we have $\Phi \coloneqq \{ \cket{\phi_i} \otimes \cket{\varphi_j} \}_{(i,j)=(1,1)}^{(\NA,\NB)}
 \in \MPDS_{A \otimes B}$.
 Furthermore, $\{ \cra{\phi_i^\dagger} \otimes \cra{\varphi_j^\dagger} \}_{(i,j)=(1,1)}^{(\NA,\NB)}$
 is the maximal measurement that perfectly distinguishes between $\Phi$.
\end{lemma}
\begin{proof}
 We have
 \begin{alignat}{1}
  \includegraphics[scale=1.0]{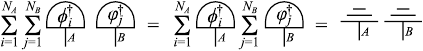} ~\raisebox{.5em}{,}
  \label{eq:derive_composite_MPDS_proof}
 \end{alignat}
 which indicates that $\Pi \coloneqq \{ \cra{e_{i,j}} \coloneqq
 \cra{\phi_i^\dagger} \otimes \cra{\varphi_j^\dagger} \}_{(i,j)=(1,1)}^{(\NA,\NB)}$
 is a measurement.
 Let $\cket{\tphi_{i,j}} \coloneqq \cket{\phi_i} \otimes \cket{\varphi_j}$; then,
 $\craket{e_{i,j}|\tphi_{i,j}} = 1$ holds, and thus $\Pi$ perfectly distinguishes between $\Phi$.
 Since $\cket{\tphi_{i,j}} \in \StNP_{A \otimes B}$ holds
 from Lemma~\ref{lemma:PureOtimes} and
 $\craket{\gdis_{A \otimes B}|\tphi_{i,j}} = \craket{\gdis_A|\phi_i} \craket{\gdis_B|\varphi_j} = 1$,
 $\Phi$ is a PDS.
 It remains to prove that $\Phi$ is an MPDS.
 Assume, by contradiction, that $\Phi$ is not an MPDS;
 then, there exists $\cket{\psi} \in \StNP_{A \otimes B}$
 that is perfectly distinguishable from $\cket{\chi_\Phi}$.
 $\cra{e_{i,j}}$ is a feasible pure effect that satisfies $\craket{e_{i,j}|\tphi_{i,j}} = 1$,
 and thus $\cra{e_{i,j}} = \cra{\tphi_{i,j}^\dagger}$ holds.
 From $\craket{e_{i,j}|\psi} = \braket{\tphi_{i,j},\psi} = 0$,
 we have $\craket{\gdis|\psi} = \sum_{i=1}^\NA \sum_{j=1}^\NB \craket{e_{i,j}|\psi} = 0$,
 which contradicts $\cket{\psi} \in \StNP_{A \otimes B}$.
 This completes our proof.
\end{proof}

From this lemma, we can easily see that,
for any two maximal measurements $\{ \cra{e_i} \}_{i=1}^\NA \in \Meas_A$ and
$\{ \cra{e'_j} \}_{j=1}^\NB \in \Meas_B$,
$\{ \cra{e_i} \otimes \cra{e'_j} \}_{(i,j)=(1,1)}^{(\NA,\NB)}$ is a maximal measurement
and that, for any two maximal effects $\cra{e_i} \in \EffM_A$ and $\cra{e'_j} \in \EffM_B$,
the effect $\cra{e_i} \otimes \cra{e'_j}$ is maximal.
This lemma also shows
\begin{alignat}{2}
 \cket{\chi_{A \otimes B}}
 &= \sum_{i=1}^\NA \sum_{j=1}^\NB \cket{\phi_i} \otimes \cket{\varphi_j}
 &&= \cket{\chi_A} \otimes \cket{\chi_B},
\end{alignat}
where $\{ \cket{\phi_i} \}_{i=1}^\NA \in \MPDS_A$ and
$\{ \cket{\varphi_j} \}_{j=1}^\NB \in \MPDS_B$ hold.

\begin{lemma} \label{lemma:NAB} 
 $N_{A \otimes B} = \NA \NB$ holds for any systems $A$ and $B$.
\end{lemma}
\begin{proof}
 From Lemma~\ref{lemma:CompositeMPDS},
 $\Phi \coloneqq \{ \cket{\phi_i} \otimes \cket{\varphi_j} \}_{(i,j)=(1,1)}^{(\NA,\NB)}$
 is an MPDS for any $\{ \cket{\phi_i} \}_{i=1}^\NA \in \MPDS_A$ and
 $\{ \cket{\varphi_j} \}_{j=1}^\NB \in \MPDS_B$.
 Thus, $N_{A \otimes B} = |\Phi| = \NA \NB$ holds.
\end{proof}

\begin{lemma} \label{lemma:DaggerOtimes} 
 For any $\cket{\rho} \in \St_A$ and $\cket{\sigma} \in \St_B$, we have
 \begin{alignat}{1}
  \includegraphics[scale=1.0]{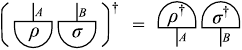} ~\raisebox{.5em}{.}
  \label{eq:derive_dagger_otimes}
 \end{alignat}
\end{lemma}
\begin{proof}
 $\cket{\rho}$ and $\cket{\sigma}$, respectively, have spectral decompositions
 of the form $\cket{\rho} = \sum_{i=1}^\NA p_i \cket{\phi_i}$ and
 $\cket{\sigma} = \sum_{j=1}^\NB q_j \cket{\varphi_j}$
 with $p_1,\ldots,p_\NA, q_1,\ldots,q_\NB \in \Real_+$,
 $\{ \cket{\phi_i} \}_{i=1}^\NA \in \MPDS_A$,
 and $\{ \cket{\varphi_j} \}_{j=1}^\NB \in \MPDS_B$.
 From the proof of Lemma~\ref{lemma:CompositeMPDS},
 $[\cket{\phi_i} \otimes \cket{\varphi_j}]^\dagger
 = \cra{\phi_i^\dagger} \otimes \cra{\varphi_j^\dagger}$ holds.
 Thus, since $\dagger$ is linear, we have
 \begin{alignat}{1}
  \includegraphics[scale=1.0]{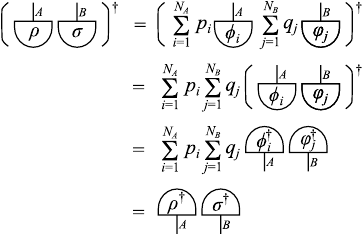} ~\raisebox{.5em}{.}
  \label{eq:derive_dagger_otimes_proof}
 \end{alignat}
\end{proof}

Lemma~\ref{lemma:DaggerOtimes} implies that $\dagger$ distributes over $\otimes$.
This result can be easily generalized to extended states and effects,
i.e., we have
\begin{alignat}{2}
 [\cket{\ol{x}} \otimes \cket{\ol{y}}]^\dagger
 &= \cra{\ol{x}^\dagger} \otimes \cra{\ol{y}^\dagger}
 &\quad &\forall \cket{\ol{x}} \in \Vec_A, \cket{\ol{y}} \in \Vec_B,
 \nonumber \\
 [\cra{\ol{v}} \otimes \cra{\ol{w}}]^\dagger
 &= \cket{\ol{v}^\dagger} \otimes \cket{\ol{w}^\dagger}
 &\quad &\forall \cra{\ol{v}} \in \Vec_A^*, \cra{\ol{w}} \in \Vec_B^*.
\end{alignat}

\subsection{Simple decompositions of state spaces of composite systems} \label{subsec:derive_Quantum_otimes}


Consider two subalgebras $\EJA'_A$ of $\Vec_A$ and $\EJA'_B$ of $\Vec_B$.
Let $\cket{\chi'_A}$ and $\cket{\chi'_B}$ denote
the identity elements of $\EJA'_A$ and $\EJA'_B$, respectively.
One can easily verify that
$\cket{\chi'} \coloneqq \cket{\chi'_A} \otimes \cket{\chi'_B}$ is an idempotent
of $\Vec_{A \otimes B}$.
Indeed, it follows that there exist
$\Phi_A \coloneqq \{ \cket{\phi_i} \}_{i=1}^k \in \PDS_A$
satisfying $\cket{\chi_{\Phi_A}} = \cket{\chi'_A}$
and $\Phi_B \coloneqq \{ \cket{\varphi_j} \}_{j=1}^l \in \PDS_B$ satisfying
$\cket{\chi_{\Phi_B}} = \cket{\chi'_B}$
and that $\Psi \coloneqq \{ \cket{\phi_i} \otimes \cket{\varphi_j} \}_{(i,j)=(1,1)}^{(k,l)}$
is a PDS of $A \otimes B$ that satisfies $\cket{\chi_\Psi} = \cket{\chi'}$.
Thus, $\cket{\chi'}$ is an idempotent.
Therefore, we can consider the Peirce 1-subalgebra of $\Vec_{A \otimes B}$
determined by the identity element $\cket{\chi'}$.
We will denote this subalgebra by $\EJA'_A \otimes \EJA'_B$.
By definition,
\begin{alignat}{1}
 \EJA'_A \otimes \EJA'_B &= \{ \cket{\ol{x}} \in \Vec_{A \otimes B} :
 \cket{\chi'} \bullet \cket{\ol{x}} = \cket{\ol{x}} \}
\end{alignat}
holds.
Let $\EJA'_A{}^+ \otimes \EJA'_B{}^+ \coloneqq (\EJA'_A \otimes \EJA'_B)^+$;
then, $\EJA'_A{}^+ \otimes \EJA'_B{}^+ = \Face_\cket{\chi'}$ holds
from Lemma~\ref{lemma:FaceBullet}.

\begin{lemma} \label{lemma:OplusOtimes} 
 When $\St_A$ and $\St_B$ have direct sum decompositions
 $\St_A = \bigoplus_{i=1}^{k_A} \EJA_{A,i}^+$
 and $\St_B = \bigoplus_{j=1}^{k_B} \EJA_{B,j}^+$,
 $\St_{A \otimes B}$ has a direct sum decomposition
 \begin{alignat}{1}
  \St_{A \otimes B} &= \bigoplus_{i=1}^{k_A} \bigoplus_{j=1}^{k_B} \EJA_{A,i}^+ \otimes \EJA_{B,j}^+.
  \label{eq:derive_StAB}
 \end{alignat}
\end{lemma}
\begin{proof}
 Equation~\eqref{eq:derive_StAB} holds if and only if
 \begin{alignat}{1}
  \Vec_{A \otimes B} &= \bigoplus_{i=1}^{k_A} \bigoplus_{j=1}^{k_B} \EJA_{(i,j)},
  \quad \EJA_{(i,j)} \coloneqq \EJA_{A,i} \otimes \EJA_{B,j}
  \label{eq:derive_VecAB}
 \end{alignat}
 holds.
 The case $k_A = k_B = 1$ is obvious since $\EJA_{(1,1)}$ is the EJA
 with the identity element $\cket{\chi_{A \otimes B}}$,
 i.e., $\EJA_{(1,1)} = \Vec_{A \otimes B}$ holds.
 Assume $k_A \ge 2$ or $k_B \ge 2$.
 To show Eq.~\eqref{eq:derive_VecAB}, it suffices to show the followings:
 (a) the EJAs $\{ \EJA_{(i,j)} \}_{(i,j)=(1,1)}^{(k_A,k_B)}$ are mutually orthogonal
 and
 (b) any $\cket{\ol{x}} \in \Vec_{A \otimes B}$ can be expressed in the form
 $\cket{\ol{x}} = \sum_{i=1}^{k_A} \sum_{j=1}^{k_B} \cket{\ol{x_{i,j}}}$
 with $\cket{\ol{x_{i,j}}} \in \EJA_{(i,j)}$.

 First, we prove (a).
 Arbitrarily choose $i,i' \in \{1,\ldots,k_A\}$ and $j,j' \in \{1,\ldots,k_B\}$
 such that $i \neq i'$ or $j \neq j'$ holds;
 then, it suffices to show that $\braket{\rho,\rho'} = 0$ holds
 for any $\cket{\rho} \in \EJA_{(i,j)}^+$ and $\cket{\rho'} \in \EJA_{(i',j')}^+$.
 Let $\cket{\chi_{A,i}}$, $\cket{\chi_{A,i'}}$, $\cket{\chi_{B,j}}$, and $\cket{\chi_{B,j'}}$ be,
 respectively, the identity elements of $\EJA_{A,i}$, $\EJA_{A,i'}$, $\EJA_{B,j}$, and $\EJA_{B,j'}$;
 then, $\craket{\chi_{A,i}^\dagger|\chi_{A,i'}} = 0$ and
 $\craket{\chi_{B,j}^\dagger|\chi_{B,j'}} = 0$ obviously hold.
 Arbitrarily choose $\cket{\rho} \in \EJA_{(i,j)}^+$ and $\cket{\rho'} \in \EJA_{(i',j')}^+$.
 Since $\cket{\chi_{A,i}} \otimes \cket{\chi_{B,j}} \in \EJA_{(i,j)}^+$ is completely mixed,
 there exists $p \in \Realpp$ such that
 $p \cket{\rho} \le \cket{\chi_{A,i}} \otimes \cket{\chi_{B,j}}$.
 Similarly, there exists $p' \in \Realpp$ such that
 $p' \cket{\rho'} \le \cket{\chi_{A,i'}} \otimes \cket{\chi_{B,j'}}$.
 Thus, we have
 %
 \begin{alignat}{3}
  pp' \braket{\rho,\rho'} &= pp' \craket{\rho^\dagger|\rho'} &&\le
  \craket{\chi_{A,i}^\dagger|\chi_{A,i'}} \craket{\chi_{B,j}^\dagger|\chi_{B,j'}}
  &&= 0,
 \end{alignat}
 i.e., $\braket{\rho,\rho'} = 0$.

 Next, we prove (b).
 For each $i$ and $j$, let $d_{A,i} \coloneqq \dim~\EJA_{A,i}$ and $d_{B,j} \coloneqq \dim~\EJA_{B,j}$.
 From Eq.~\eqref{eq:derive_EJA_rank_dim2},
 $\sum_{i=1}^{k_A} d_{A,i} = D_A$ and
 $\sum_{j=1}^{k_B} d_{B,j} = D_B$ hold.
 Choose a basis, $\{ \cket{\ol{w_{i,s}}} \}_{(i,s)={(1,1)}}^{(k_A,d_{A,i})}$, of $\Vec_A$
 satisfying $\cket{\ol{w_{i,s}}} \in \EJA_{A,i}$ for each $i$ and $s$,
 and a basis, $\{ \cket{\ol{w'_{j,t}}} \}_{(j,t)={(1,1)}}^{(k_B,d_{B,j})}$, of $\Vec_B$
 satisfying $\cket{\ol{w'_{j,t}}} \in \EJA_{B,j}$ for each $j$ and $t$.
 $\Psi \coloneqq \{ \cket{\ol{w_{i,j;s,t}}} \coloneqq
 \cket{\ol{w_{i,s}}} \otimes \cket{\ol{w'_{j,t}}} \}_{(i,j,s,t)=(1,1,1,1)}^{(k_A,k_B,d_{A,i},d_{B,j})}$
 is a set of linearly independent extended states.
 From $|\Psi| = \sum_{i=1}^{k_A} \sum_{j=1}^{k_B} d_{A,i} d_{B,j} = D_A D_B$ and
 Lemma~\ref{lemma:DAB}, $|\Psi| = D_{A \otimes B}$ holds,
 which implies that $\Psi$ is a basis of $\Vec_{A \otimes B}$.
 Thus, any $\cket{\ol{x}} \in \Vec_{A \otimes B}$ can be expressed in the form
 \begin{alignat}{2}
  \cket{\ol{x}} = \sum_{i=1}^{k_A} \sum_{j=1}^{k_B} \cket{\ol{x_{i,j}}},
  &\quad \cket{\ol{x_{i,j}}} &\coloneqq
  \sum_{s=1}^{d_{A,i}} \sum_{t=1}^{d_{B,j}} c^{(i,j)}_{s,t} \cket{\ol{w_{i,j;s,t}}},
 \end{alignat}
 where $c^{(i,j)}_{s,t} \in \Real$.
 The inclusion $\cket{\ol{w_{i,j;s,t}}} \in \EJA_{(i,j)}$ is justified as follows.
 Let $\cket{\chi_{ij}} \coloneqq \cket{\chi_{A,i}} \otimes \cket{\chi_{B,j}}$ be the identity element of $\EJA_{(i,j)}$.
 If $\cket{x} \in \EJA_{A,i}^+$ and $\cket{y} \in \EJA_{B,j}^+$, then for some $r,s>0$ we have
 $\cket{x} \le r\cket{\chi_{A,i}}$ and $\cket{y} \le s\cket{\chi_{B,j}}$,
 hence $\cket{x} \otimes \cket{y} \le rs\cket{\chi_{ij}}$ and
 $\cket{x} \otimes \cket{y} \in \EJA_{(i,j)}^+$.
 Since $\EJA_{A,i}$ and $\EJA_{B,j}$ are spanned by their positive cones,
 $\cket{\ol{w_{i,j;s,t}}} \in \EJA_{(i,j)}$ holds for any $i, j, s,$ and $t$.
 Thus, we have $\cket{\ol{x_{i,j}}} \in \EJA_{(i,j)}$.
\end{proof}

This proof tells us that the dimension of $\EJA_{(i,j)}$ is $d_{A,i} d_{B,j}$,
i.e.,
\begin{alignat}{1}
 \dim(\EJA_{A,i} \otimes \EJA_{B,j}) &= (\dim~\EJA_{A,i}) \cdot (\dim~\EJA_{B,j}).
\end{alignat}
One can also easily obtain
\begin{alignat}{1}
 \rank(\EJA_{A,i} \otimes \EJA_{B,j}) &= (\rank~\EJA_{A,i}) \cdot (\rank~\EJA_{B,j}).
 \label{eq:derive_rankEJAij}
\end{alignat}
Note that Eq.~\eqref{eq:derive_StAB} can also be expressed by
\begin{alignat}{1}
 \left( \bigoplus_{i=1}^{k_A} \EJA_{A,i}^+ \right) \otimes \left( \bigoplus_{j=1}^{k_B} \EJA_{B,j}^+ \right)
 &= \bigoplus_{i=1}^{k_A} \bigoplus_{j=1}^{k_B} \EJA_{A,i}^+ \otimes \EJA_{B,j}^+.
 \label{eq:derive_StAB2}
\end{alignat}
This indicates that the operation $\otimes$ on EJAs distributes over
the operation $\oplus$ on EJAs.

We close this subsection with an important result, which will be useful in the sequel.
\begin{lemma} \label{lemma:OtimesSimple} 
 When $\St_A$ and $\St_B$ have the simple decompositions
 $\St_A = \bigoplus_{i=1}^{k_A} \EJA_{A,i}^+$
 and $\St_B = \bigoplus_{j=1}^{k_B} \EJA_{B,j}^+$,
 $\St_{A \otimes B}$ has the simple decomposition given by Eq.~\eqref{eq:derive_StAB}.
\end{lemma}
\begin{proof}
 From Lemma~\ref{lemma:OplusOtimes},
 $\St_{A \otimes B}$ can be expressed in the form of Eq.~\eqref{eq:derive_StAB},
 it is sufficient to show that $\EJA_{(i,j)} \coloneqq \EJA_{A,i} \otimes \EJA_{B,j}$
 is simple for each $i \in \{1,\ldots,\NA\}$ and $j \in \{1,\ldots,\NB\}$.
 Consider arbitrary fixed $i$ and $j$.
 Assume, by contradiction, that $\EJA_{(i,j)}$ is not simple,
 i.e., $\EJA_{(i,j)}$ can be decomposed into a direct sum
 $\EJA_{(i,j)} = \EJA_0 \oplus \EJA_1$,
 where $\EJA_0$ and $\EJA_1$ are EJAs with nonzero ranks.
 Let $\{ \cket{\phi_l} \}_{l=1}^n$ $~(n \coloneqq \rank~\EJA_{A,i})$ and
 $\{ \cket{\phi'_{l'}} \}_{l'=1}^{n'}$ $~(n' \coloneqq \rank~\EJA_{B,j})$ be,
 respectively, Jordan frames of $\EJA_{A,i}$ and $\EJA_{B,j}$.
 Then, one can easily see that
 $\Phi \coloneqq \{ \cket{\tphi_{l,l'}} \coloneqq
 \cket{\phi_l} \otimes \cket{\phi'_{l'}} \}_{(l,l')=(1,1)}^{(n,n')}$
 is a Jordan frame of $\EJA_{(i,j)}$.
 Since both $\EJA_{A,i}$ and $\EJA_{B,j}$ are simple,
 there exist two normalized pure states $\cket{\varphi} \in \EJA_{A,i}^+$ and
 $\cket{\varphi'} \in \EJA_{B,j}^+$ such that
 \begin{alignat}{2}
  \craket{\varphi^\dagger|\phi_l} &> 0, &\quad \forall l &\in \{1,\ldots,n\}, \nonumber \\
  \craket{\varphi'{}^\dagger|\phi'_{l'}} &> 0, &\quad \forall l' &\in \{1,\ldots,n'\},
 \end{alignat}
 which can be verified for each simple EJA classified by
 the Jordan-von Neumann-Wigner theorem.
 Let $\cket{\tvarphi} \coloneqq \cket{\varphi} \otimes \cket{\varphi'}
 \in \EJA_{(i,j)}^+$; then,
 from Lemma~\ref{lemma:PureOtimes}, $\cket{\tvarphi}$ is pure.
 Since any nonzero pure state $\cket{\psi} \in \EJA_{(i,j)}^+$ satisfies
 either $\cket{\psi} \in \EJA_0^+$ or $\cket{\psi} \in \EJA_1^+$ (but not both),
 we can permute $\EJA_0$ and $\EJA_1$, if necessary, so that $\cket{\tvarphi} \in \EJA_0^+$ holds.
 Then, we have that, for any $l \in \{1,\ldots,n\}$ and $l' \in \{1,\ldots,n'\}$,
 \begin{alignat}{2}
  \craket{\tvarphi^\dagger|\tphi_{l,l'}}
  &= \craket{\varphi^\dagger|\phi_l} \craket{\varphi'{}^\dagger|\phi'_{l'}}
  &&> 0,
 \end{alignat}
 i.e., $\cket{\tphi_{l,l'}} \in \EJA_0^+$.
 Thus, $\Phi \subset \EJA_0^+$ holds,
 which gives $\rank~\EJA_0 = \rank~\EJA_{(i,j)}$
 and $\rank~\EJA_1 = \rank~\EJA_{(i,j)} - \rank~\EJA_0 = 0$.
 This contradicts $\rank~\EJA_1 \neq 0$,
 and hence we conclude that every $\EJA_{(i,j)}$ is simple.
\end{proof}

\subsection{State space is isomorphic to direct sum of spaces of complex positive semidefinite matrices}%
\label{subsec:derive_Quantum_hilbert}

In this subsection, we derive Properties~(\hyperlink{property:A}{A}) and (\hyperlink{property:B}{B})
given in Subsec.~\ref{subsec:derive_abstract_stream}.

\begin{thm} \label{thm:Complex} 
 For any system $A$,
 $\St_A \cong \bigoplus_{i=1}^k \mS_+(\Complex^{n_i})$
 holds for some natural numbers $k,n_1,\ldots,n_k$ with $\sum_{i=1}^k n_i = \NA$
 \footnote{\label{fn:Sel}%
 Theorem~\ref{thm:Complex} is somewhat similar to Theorem~4.14 in Ref.~\cite{Sel-Sca-Coe-2018}.
 We here point out that there seems to be a gap in the proof of the latter.
 This proof shows that if $\St_A$ is simple, then $\St_A \cong \mS_+(\Complex^\NA)$ holds;
 however, it does not rule out the possibility of, for example,
 $\St_A \cong \mS_+(\Hermite^2) \oplus \Spin_{10}$.
 (Note that $\mS_+(\Hermite^2) \oplus \Spin_{10}$ has the same rank and dimension as
 $\mS_+(\Complex^4)$.)}.
\end{thm}
\begin{proof}
 We will use the Jordan-von Neumann-Wigner theorem.
 $\St_A$ has the simple decomposition of the form $\St_A = \bigoplus_{i=1}^{k} \EJA_{(i)}^+$.
 We only need to prove $\EJA_{(i)} \cong \mS(\Complex^{n_i})$ for
 each $i \in \{1,\ldots,k\}$, where $n_i \coloneqq \rank~\EJA_{(i)}$.

 Consider arbitrary fixed $i$.
 Let $n \coloneqq n_i$.
 In the case of $n = 1$, $\EJA_{(i)} \cong \mS(\Complex)$ obviously holds from
 $\mS(\Real) \cong \mS(\Complex) \cong \mS(\Hermite)$,
 so assume $n \ge 2$.
 $\EJA_{(i,i)} \coloneqq \EJA_{(i)} \otimes \EJA_{(i)}$ is
 simple from Lemma~\ref{lemma:OtimesSimple}.
 Let $r \coloneqq \rank~\EJA_{(i,i)} = n^2$ and
 $d \coloneqq \dim~\EJA_{(i,i)} = (\dim~\EJA_{(i)})^2$.
 We prove, using Table~\ref{table:EJA}, that
 $\mS(\Oct^3)$, $\mS(\Real^n)$, $\mS(\Hermite^n)$, and $\Spin_s$ with $s \ge 5$
 are ruled out.
 Firstly, considering the case $\EJA_{(i)} \cong \mS(\Oct^3)$,
 one can see from $n = 3$ that the simple EJA $\EJA_{(i,i)}$ must have $r = 9$ and $d = 27^2$,
 which contradicts Table~\ref{table:EJA}.
 Secondly, we consider the case $\EJA_{(i)} \cong \mS(\Real^n)$,
 which leads to $r = n^2$ and $d = n^2(n+1)^2/4$,
 and easily see that this case is ruled out.
 In the same way, considering the case $\EJA_{(i)} \cong \mS(\Hermite^n)$,
 which leads to $r = n^2$ and $d = n^2(2n-1)^2$,
 one can see that this case is also ruled out.
 Finally, we consider the case $\EJA_{(i)} \cong \Spin_s$ with $s \ge 5$,
 which gives, from $n = 2$, $r = 4$ and $d = s^2$.
 Since $\EJA_{(i,i)}$ is not isomorphic to $\mS(\Hermite^4)$,
 this case is ruled out.
 Thus, we conclude that $\EJA_{(i)} \cong \mS(\Complex^n)$ must hold.
\end{proof}

Combining Lemma~\ref{lemma:OtimesSimple} and Theorem~\ref{thm:Complex}
gives the following theorem.
\begin{thm} \label{thm:ComplexOtimes} 
 For any systems $A$ and $B$
 with $\St_A \cong \bigoplus_{i=1}^{k_A} \mS_+(\Complex^{m_i})$ and
 $\St_B \cong \bigoplus_{j=1}^{k_B} \mS_+(\Complex^{n_j})$,
 we have $\St_{A \otimes B} \cong \bigoplus_{i=1}^{k_A} \bigoplus_{j=1}^{k_B}
 \mS_+(\Complex^{m_in_j})$.
\end{thm}
\begin{proof}
 $\St_A$ and $\St_B$ have the simple decomposition of the form
 $\St_A = \bigoplus_{i=1}^{k_A} \EJA_{A,i}^+$ with $\EJA_{A,i}^+ \cong \mS_+(\Complex^{m_i})$
 and $\St_B = \bigoplus_{j=1}^{k_B} \EJA_{B,j}^+$ with $\EJA_{B,j}^+ \cong \mS_+(\Complex^{n_j})$.
 From Lemma~\ref{lemma:OtimesSimple},
 $\St_{A \otimes B}$ has the simple decomposition
 $\St_{A \otimes B} = \bigoplus_{i=1}^{k_A} \bigoplus_{j=1}^{k_B}
 \EJA_{A,i}^+ \otimes \EJA_{B,j}^+$.
 Since $\rank(\EJA_{A,i} \otimes \EJA_{B,j}) = (\rank~\EJA_{A,i}) \cdot (\rank~\EJA_{B,j}) = m_in_j$
 holds, applying Theorem~\ref{thm:Complex}, we have
 $\EJA_{A,i}^+ \otimes \EJA_{B,j}^+ \cong \mS_+(\Complex^{m_in_j})$.
\end{proof}

As will be shown here, Theorems~\ref{thm:Complex} and \ref{thm:ComplexOtimes} allow us to
obtain a simple expression for a state (resp. extended state)
in terms of a positive semidefinite matrix (resp. Hermitian matrix).
For each system $A$, $\St_A$ has the simple decomposition of the form
$\St_A = \bigoplus_{i=1}^{k_A} \EJA_{A,i}^+$,
where $\EJA_{A,i}^+ \cong \mS_+(\Complex^{n_i})$ holds for some natural number $n_i$.
Let us choose a set of $k_A$ mutually orthogonal projection matrices of order $\NA$,
denoted by $\mP_A \coloneqq \{ P_{A,i} \}_{i=1}^{k_A}$,
with $\rank~P_{A,i} = n_i$.
$\sum_{i=1}^{k_A} P_{A,i} = \ident_\NA$ obviously holds.
$\StMat_A$ is defined as
\begin{alignat}{1}
 \StMat_A &\coloneqq \left\{ \sum_{i=1}^{k_A} H_i : H_i \in \StMat_{A,i} \right\}, \nonumber \\
 \StMat_{A,i} &\coloneqq \left\{ H \in \mS(\Complex^\NA) : P_{A,i} H P_{A,i} = H \right\}.
 \label{eq:derive_StMat}
\end{alignat}
$\StMat_A$ is an EJA with the Jordan product
$H \bullet H' = (H H' + H' H) / 2$ $~(H,H' \in \StMat_A)$.
Clearly, $\ident_\NA$ is the identity element of $\StMat_A$.
$\tr~H = \Tr~H$ also holds for any $H \in \StMat_A$.
It follows that the cone of squares, $\StMat_A^+$, of $\StMat_A$ is
the set of all positive semidefinite matrices in $\StMat_A$.
Each $\StMat_{A,i}$ is the subalgebra of $\StMat_A$ with
the identity element $P_{A,i}$.
One can easily check $\StMat_{A,i} \cong \mS(\Complex^{n_i}) \cong \EJA_{A,i}$,
$\StMat_A = \bigoplus_{i=1}^{k_A} \StMat_{A,i} \cong \Vec_A$,
and $\StMat_A^+ \cong \St_A$.
Using an isomorphism $\Mat^A$ from $\Vec_A$ to $\StMat_A$,
we can fully and faithfully represent any $\cket{\ol{x}} \in \Vec_A$ as
the corresponding Hermitian matrix $\Mat^A_\cket{\ol{x}} \in \StMat_A$.
In particular, $\Mat^A_\cket{\rho} \in \StMat_A^+$ holds
if and only if $\cket{\rho} \in \St_A$ holds.
$\Mat^A_\cket{\chi_A} = \ident_\NA$ and
$\craket{\gdis|\ol{x}} = \tr~\cket{\ol{x}} = \Tr~\Mat^A_\cket{\ol{x}}$
$~(\cket{\ol{x}} \in \Vec_A)$ obviously hold.
We often denote $\Mat^A$ briefly by $\Mat$.

Now, let us consider composite systems.
For any systems $A$ and $B$ with
$\mP_A \coloneqq \{ P_{A,i} \}_{i=1}^{k_A}$ and
$\mP_B \coloneqq \{ P_{B,j} \}_{j=1}^{k_B}$,
we can choose $\mP_{A \otimes B}$ as
$\mP_{A \otimes B} \coloneqq \{ P_{A \otimes B,(i-1)k_B+j} \coloneqq P_{A,i} \otimes P_{B,j}
\}_{(i,j)=(1,1)}^{(k_A,k_B)}$,
where the Kronecker product of two matrices $X_1$ and $X_2$ is denoted by
$X_1 \otimes X_2$.
Indeed, $\mP_{A \otimes B}$ is a set of $k_A k_B$ mutually orthogonal projection matrices
of order $\NA\NB$.
Let $\{ \cket{\ol{v_i}} \}_{i=1}^{D_A}$ and $\{ \cket{\ol{w_j}} \}_{j=1}^{D_B}$ be,
respectively, ONBs of $\Vec_A$ and $\Vec_B$.
Since $\{ \cket{\ol{v_i}} \otimes \cket{\ol{w_j}} \}_{(i,j)=(1,1)}^{(D_A,D_B)}$ is
an ONB of $\Vec_{A \otimes B}$,
$\Vec_{A \otimes B}$ is a tensor product space of real Hilbert spaces $\Vec_A$ and $\Vec_B$.
In contrast, $\StMat_{A \otimes B}$ is also a tensor product space of
real Hilbert spaces $\StMat_A$ and $\StMat_B$.
Thus, for given two isomorphisms $\Mat^A:\Vec_A \to \StMat_A$ and
$\Mat^B:\Vec_B \to \StMat_B$,
we can consider the isomorphism (as real Hilbert spaces)
$\Mat^{A \otimes B}:\Vec_{A \otimes B} \to \StMat_{A \otimes B}$ such that
\begin{alignat}{2}
 \Mat^{A \otimes B}_{\cket{\ol{x}} \otimes \cket{\ol{y}}}
 &= \Mat^A_\cket{\ol{x}} \otimes \Mat^B_\cket{\ol{y}}, &\quad
 \forall \cket{\ol{x}} \in \Vec_A, \cket{\ol{y}} \in \Vec_B.
 \label{eq:derive_StateTensor}
\end{alignat}
Note that $\Mat^{A \otimes B}$ is uniquely determined by Eq.~\eqref{eq:derive_StateTensor};
indeed, any $\cket{\ol{z}} \in \Vec_{A \otimes B}$ can be expressed in the form
$\cket{\ol{z}} = \sum_{i=1}^l \cket{\ol{x_i}} \otimes \cket{\ol{y_i}}$
with $\cket{\ol{x_i}} \in \Vec_A$ and $\cket{\ol{y_i}} \in \Vec_B$,
and thus $\Mat^{A \otimes B}_\cket{\ol{z}} =
\sum_{i=1}^l \Mat^A_\cket{\ol{x_i}} \otimes \Mat^B_\cket{\ol{y_i}}$ holds.
Without loss of generality, we can think of $\Mat^{A \otimes B}$ as
an isomorphism as EJAs.
We will choose such an isomorphism $\Mat^{A \otimes B}$ for any systems $A$ and $B$.
It is easily seen that, for any systems $A$, $B$, and $C$,
\begin{alignat}{1}
 \Mat^A_\cket{\ol{x}} \otimes \Mat^{B \otimes C}_{\cket{\ol{y}} \otimes \cket{\ol{z}}}
 &= \Mat^{A \otimes B \otimes C}_{\cket{\ol{x}} \otimes \cket{\ol{y}} \otimes \cket{\ol{z}}}
 = \Mat^{A \otimes B}_{\cket{\ol{x}} \otimes \cket{\ol{y}}} \otimes \Mat^C_\cket{\ol{z}},
 \nonumber \\
 &\quad \forall \cket{\ol{x}} \in \Vec_A, \cket{\ol{y}} \in \Vec_B, \cket{\ol{z}} \in \Vec_C
 \label{eq:derive_StateTensor2}
\end{alignat}
holds.

\subsection{Correspondence between processes and CP maps} \label{subsec:derive_Quantum_process}

We here derive Property~(\hyperlink{property:C}{C}).
Let us begin with some preliminaries.
For any system $A$, the simple decomposition of $\St_A$ can be expressed by
$\St_A = \bigoplus_{l=1}^k \EJA_{(l)}^+$
with $\EJA_{(l)} \cong \mS(\Complex^{n_l})$.
For each $l$, let $\Mat^{(l)}$ be an isomorphism from $\EJA_{(l)}$ to $\mS(\Complex^{n_l})$
and $\{ \cket{\ol{w_{l;s,t}}} \}_{(s,t)=(1,1)}^{(n_l,n_l)}$ be
the ONB of $\EJA_{(l)}$ satisfying
\begin{alignat}{1}
 \Mat^{(l)}_\cket{\ol{w_{l;s,t}}} &\coloneqq
 \begin{dcases}
  \ket{s}\bra{s}, & s = t, \\
  \frac{1}{\sqrt{2}}(\ket{s}\bra{t} + \ket{t}\bra{s}), & s < t, \\
  \frac{\i}{\sqrt{2}}(\ket{s}\bra{t} - \ket{t}\bra{s}), & s > t, \\
 \end{dcases}
 \label{eq:derive_Mat_w}
\end{alignat}
where $\{ \ket{s} \}_{s=1}^{n_l}$ is an ONB of the complex Hilbert space $\Complex^{n_l}$.
Consider the following extended state of system $A \otimes A$:
\begin{alignat}{2}
 \cket{\cup_A} &\coloneqq \sum_{l=1}^k \cket{\cup^{(l)}_A}, \nonumber \\
 \cket{\cup^{(l)}_A} &\coloneqq
 \sum_{s=1}^{n_l} \sum_{t=1}^{n_l} \gamma^{(l)}_{s,t} \cket{\ol{w_{l;s,t}}} \otimes \cket{\ol{w_{l;s,t}}}
 \in \EJA_{(l)} \otimes \EJA_{(l)},
 \label{eq:derive_quant_Bell}
\end{alignat}
where
\begin{alignat}{1}
 \gamma^{(l)}_{s,t} &\coloneqq
 \begin{dcases}
  1, & s \le t, \\
  -1, & s > t. \\
 \end{dcases}
 \label{eq:derive_gamma_lst}
\end{alignat}
Let $\Mat^{(l,l)}:\EJA_{(l)} \otimes \EJA_{(l)} \to \mS(\Complex^{n_l^2})$
be the isomorphism that satisfies $\Mat^{(l,l)}_{\cket{\ol{x}} \otimes \cket{\ol{y}}}
= \Mat^{(l)}_\cket{\ol{x}} \otimes \Mat^{(l)}_\cket{\ol{y}}$
for any $\cket{\ol{x}},\cket{\ol{y}} \in \EJA_{(l)}$; then,
from Eqs.~\eqref{eq:derive_Mat_w} and \eqref{eq:derive_quant_Bell}, we have
\begin{alignat}{2}
 \Mat^{(l,l)}_\cket{\cup^{(l)}_A}
 &= \sum_{s=1}^{n_l} \sum_{t=1}^{n_l} (\ket{s} \otimes \ket{s})(\bra{t} \otimes \bra{t})
 &&= \ket{\Gamma}\bra{\Gamma}, \nonumber \\
 \ket{\Gamma} &\coloneqq \sum_{i=1}^{n_l} \ket{i} \otimes \ket{i}.
\end{alignat}
It follows that $\Mat^{(l,l)}_\cket{\cup^{(l)}_A}$ is positive semidefinite,
and thus $\cket{\cup_A} \in \St_{A \otimes A}$ holds.
$\cket{\cup_A}$ can be expressed by
\begin{alignat}{2}
 \cket{\cup_A} &= \sum_{i=1}^{D_A} \gamma_i \cket{\ol{w_i}} \otimes \cket{\ol{w_i}},
 \label{eq:derive_quant_Bell_i}
\end{alignat}
where
\begin{alignat}{1}
 \mu(l,s,t) &\coloneqq \sum_{j=1}^{l-1} n_j^2 + (s - 1) n_l + t \in \{1,\ldots,D_A\},
 \nonumber \\
 \gamma_{\mu(l,s,t)} &\coloneqq \gamma^{(l)}_{s,t}, \nonumber \\
 \cket{\ol{w_{\mu(l,s,t)}}} &\coloneqq \cket{\ol{w_{l;s,t}}}.
 \label{eq:derive_mulst}
\end{alignat}
It follows that $\{ \cket{\ol{w_i}} \}_{i=1}^{D_A}$ is an ONB of $\Vec_A$.
$\gamma_i \in \{ 1, -1 \}$ holds from Eq.~\eqref{eq:derive_gamma_lst}.
It is noteworthy that $\cket{\cup_A}$ is somewhat similar to
$\cket{\ol{\eta}_A}$ of Eq.~\eqref{eq:derive_quant_etaepsilon};
however, $\cket{\cup_A}$ is a state, while $\cket{\ol{\eta}_A}$ is not in general a state.
Let $\cra{\cap_A} \coloneqq \cket{\cup_A}^\dagger$.

\begin{ex}[quantum theory]
 Consider a system $A$ with $\NA = 2$.
 If $A$ is classical, i.e., $\St_A \cong \mS_+(\Complex) \oplus \mS_+(\Complex)$,
 then $\cket{\cup_A}$ is expressed by
 \begin{alignat}{1}
  \Mat^{A \otimes A}_\cket{\cup_A}
  &=
  \begin{bmatrix}
   1 & 0 & 0 & 0 \\
   0 & 0 & 0 & 0 \\
   0 & 0 & 0 & 0 \\
   0 & 0 & 0 & 1 \\
  \end{bmatrix}.
 \end{alignat}
 In this case, $\cket{\ol{\eta}_A}$ is a state;
 indeed, $\cket{\ol{\eta}_A} = \cket{\cup_A}$ holds.
 If $A$ is fully quantum, i.e., $\St_A \cong \mS_+(\Complex^2)$,
 then $\cket{\cup_A}$ is expressed by
 \begin{alignat}{1}
  \Mat^{A \otimes A}_\cket{\cup_A}
  &=
  \begin{bmatrix}
   1 & 0 & 0 & 1 \\
   0 & 0 & 0 & 0 \\
   0 & 0 & 0 & 0 \\
   1 & 0 & 0 & 1 \\
  \end{bmatrix}.
 \end{alignat}
\end{ex}

\begin{lemma} \label{lemma:ZigZag} 
 \begin{alignat}{1}
  \includegraphics[scale=1.0]{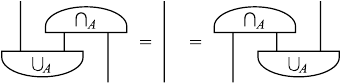} ~\raisebox{.5em}{.}
  \label{eq:derive_cup_cap_id}
 \end{alignat}
\end{lemma}
\begin{proof}
 Since $\{ \cket{\ol{w_i}} \}_{i=1}^{D_A}$ of Eq.~\eqref{eq:derive_mulst} is an ONB of $\Vec_A$,
 any $\cket{\rho} \in \St_A$ can be expressed in the form
 $\cket{\rho} = \sum_{i=1}^{D_A} \rho_i \cket{\ol{w_i}}$ with $\rho_i \in \Real$.
 Thus, from Eq.~\eqref{eq:derive_quant_Bell_i}, we have
 \begin{alignat}{1}
  \includegraphics[scale=1.0]{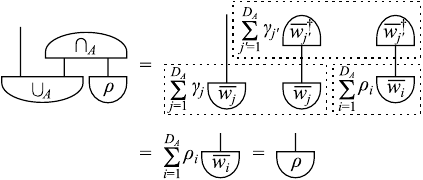} ~\raisebox{.5em}{,}
  \label{eq:derive_cup_cap_id_proof}
 \end{alignat}
 where the second equality follows from
 $\craket{\ol{w_i}^\dagger|\ol{w_{i'}}} = \delta_{i,i'}$ and $\gamma_i^2 = 1$.
 This proves, from the \poslink{pos:ProcEq}{local equality postulate},
 the first equality of Eq.~\eqref{eq:derive_cup_cap_id}.
 The same is true for the second equality of Eq.~\eqref{eq:derive_cup_cap_id}.
\end{proof}

For any systems $A$ and $B$,
let $\mathbf{CP}_{A \to B}$ be the set of all CP maps from $\StMat_A$ to $\StMat_B$
and $\StMat_{A \to B}$ be the set of all linear maps from $\StMat_A$ to $\StMat_B$.
It is easily seen that $\mathbf{CP}_{A \to B}$ is a convex cone in $\StMat_{A \to B}$.
\begin{thm} \label{thm:Proc} 
 For any systems $A$ and $B$, $\Proc_{A \to B}$ and $\mathbf{CP}_{A \to B}$ are
 isomorphic as convex cones.
 Also, there exists an isomorphism
 $\mL:\Proc_{A \to B} \ni f \mapsto \mL_f \in \mathbf{CP}_{A \to B}$
 such that
 \begin{alignat}{1}
  \Mat^B_{f \circ \cket{\rho}} = \mL_f[\Mat^A_\cket{\rho}], &\quad
  \forall f \in \Proc_{A \to B}, \cket{\rho} \in \St_A.
  \label{eq:derive_Proc}
 \end{alignat}
\end{thm}
\begin{proof}
 Let $\mL^{A \to B}:\Vec_{A \to B} \ni \ol{f} \mapsto \mL_{\ol{f}}
 \coloneqq \Mat^B \circ \ol{f} \circ \tilde{\Mat}^A \in \StMat_{A \to B}$,
 where $\tilde{\Mat}^A:\StMat_A \to \Vec_A$ is the inverse of
 the isomorphism $\Mat^A:\Vec_A \to \StMat_A$.
 We simply denote $\mL^{A \to B}$ by $\mL$.
 One can easily see that $\mL$ is linear and satisfies
 \begin{alignat}{2}
  \mL_{\ol{f}}[\Mat^A_\cket{\ol{x}}] &= \Mat^B_{\ol{f} \circ \cket{\ol{x}}}, &\quad
  \forall \ol{f} \in \Vec_{A \to B}, \cket{\ol{x}} &\in \Vec_A.
  \label{eq:derive_mL}
 \end{alignat}
 It follows that $\Vec_{A \to B}$ and $\StMat_{A \to B}$ are isomorphic as vector spaces
 and that $\mL:\Vec_{A \to B} \to \StMat_{A \to B}$ is an isomorphism,
 whose inverse is $\mL^{-1}:\StMat_{A \to B} \ni h \mapsto \mL^{-1}_h \coloneqq
 \tilde{\Mat}^B \circ h \circ \Mat^A \in \Vec_{A \to B}$.
 Indeed, we have that, for any $\ol{f} \in \Vec_{A \to B}$ and
 $h \in \StMat_{A \to B}$,
 \begin{alignat}{2}
  (\mL^{-1} \circ \mL)(\ol{f})
  &= \tilde{\Mat}^B \circ (\Mat^B \circ \ol{f} \circ \tilde{\Mat}^A) \circ \Mat^A
  &&= \ol{f}, \nonumber \\
  (\mL \circ \mL^{-1})(h)
  &= \Mat^B \circ (\tilde{\Mat}^B \circ h \circ \Mat^A) \circ \tilde{\Mat}^A
  &&= h.
 \end{alignat}
 What is left is to show that the restriction of $\mL$ to $\Proc_{A \to B}$,
 denoted by the same notation $\mL$,
 is an isomorphism from $\Proc_{A \to B}$ to $\mathbf{CP}_{A \to B}$.
 It suffices to show
 $\mL_f \in \mathbf{CP}_{A \to B}$ for any $f \in \Proc_{A \to B}$ and
 $\mL^{-1}_c \in \Proc_{A \to B}$ for any $c \in \mathbf{CP}_{A \to B}$.

 First, we show $\mL_f \in \mathbf{CP}_{A \to B}$ for any $f \in \Proc_{A \to B}$.
 $\mL^{A \to A}_{\id_A} \in \StMat_{A \to A}$ is the identity operator on $\StMat_A$.
 Also, we have that, for any
 $\ol{f} \in \Vec_{A \to B}$, $\ol{g} \in \Vec_{C \to D}$,
 $\cket{\ol{x}} \in \Vec_A$, and $\cket{\ol{y}} \in \Vec_C$,
 \begin{alignat}{2}
  \mL_{\ol{f} \otimes \ol{g}}[\Mat_{\cket{\ol{x}} \otimes \cket{\ol{y}}}]
  &= \Mat_{(\ol{f} \otimes \ol{g}) \circ [\cket{\ol{x}} \otimes \cket{\ol{y}}]}
  \nonumber \\
  &= \Mat_{[\ol{f} \circ \cket{\ol{x}}] \otimes [\ol{g} \circ \cket{\ol{y}}]}
  \nonumber \\
  &= \Mat_{\ol{f} \circ \cket{\ol{x}}} \otimes \Mat_{\ol{g} \circ \cket{\ol{y}}}
  \nonumber \\
  &= \mL_{\ol{f}}[\Mat_\cket{\ol{x}}] \otimes \mL_{\ol{g}}[\Mat_\cket{\ol{y}}]
  \nonumber \\
  &= (\mL_{\ol{f}} \otimes \mL_{\ol{g}})[\Mat_{\cket{\ol{x}} \otimes \cket{\ol{y}}}],
 \end{alignat}
 where the third and last lines follow from Eq.~\eqref{eq:derive_StateTensor}.
 Thus, we have
 \begin{alignat}{2}
  \mL_{\ol{f} \otimes \ol{g}} &= \mL_{\ol{f}} \otimes \mL_{\ol{g}},
  &\quad & \forall \ol{f} \in \Vec_{A \to B}, \ol{g} \in \Vec_{C \to D}.
  \label{eq:derive_mL_otimes}
 \end{alignat}
 Moreover, we have that, for any $f \in \Proc_{A \to B}$ and $E \in \System$,
 \begin{alignat}{2}
  (f \otimes \id_E) \circ \cket{\sigma} &\in \St_{B \otimes E},
  &\quad &\forall \cket{\sigma} \in \St_{A \otimes E},
 \end{alignat}
 which implies that
 $\mL_{f \otimes \id_E} = \mL_f \otimes \mL_{\id_E}$ is positive,
 and thus $\mL_f$ is CP (i.e., $\mL_f \in \mathbf{CP}_{A \to B}$).

 Next, we show $\mL^{-1}_c \in \Proc_{A \to B}$ for any $c \in \mathbf{CP}_{A \to B}$.
 Arbitrarily choose $c \in \mathbf{CP}_{A \to B}$.
 Let $\cket{\sigma_c} \coloneqq (\id_A \otimes \mL^{-1}_c) \circ \cket{\cup_A}
 \in \Vec_{A \otimes B}$, i.e.,
 \begin{alignat}{1}
  \includegraphics[scale=1.0]{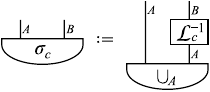} ~\raisebox{.5em}{.}
  \label{eq:derive_proc_sigma}
 \end{alignat}
 From Eqs.~\eqref{eq:derive_mL} and \eqref{eq:derive_mL_otimes},
 we have
 \begin{alignat}{2}
  \Mat_\cket{\sigma_c} &= \Mat_{(\id_A \otimes \mL^{-1}_c) \circ \cket{\cup_A}}
  \nonumber \\
  &= \mL_{\id_A \otimes \mL^{-1}_c}[\Mat_\cket{\cup_A}]
  &&= (\mL_{\id_A} \otimes c)[\Mat_\cket{\cup_A}].
 \end{alignat}
 From $\Mat_\cket{\cup_A} \in \StMat_{A \otimes A}^+$ and $c \in \mathbf{CP}_{A \to B}$,
 we have $\Mat_\cket{\sigma_c} \in \StMat_{A \otimes B}^+$,
 i.e., $\cket{\sigma_c} \in \St_{A \otimes B}$.
 Let
 \begin{alignat}{1}
  \includegraphics[scale=1.0]{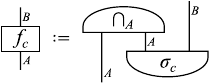} ~\raisebox{.5em}{.}
  \label{eq:derive_proc_fc}
 \end{alignat}
 Then, since $\cket{\sigma_c}$ and $\cra{\cap_A}$
 are, respectively, a state and an effect, $f_c$ is a process.
 Substituting Eq.~\eqref{eq:derive_proc_sigma} into Eq.~\eqref{eq:derive_proc_fc}
 and using Eq.~\eqref{eq:derive_cup_cap_id} yields $\mL^{-1}_c = f_c$.
 Therefore, $\mL^{-1}_c \in \Proc_{A \to B}$ holds.
\end{proof}

\begin{thm} \label{thm:ProcFeasible}
 Let $\mL:\Proc_{A \to B} \ni f \mapsto \mL_f \in \mathbf{CP}_{A \to B}$
 be an isomorphism satisfying Eq.~\eqref{eq:derive_Proc};
 then, $f \in \Proc_{A \to B}$ is deterministic if and only if
 $\mL_f$ is a TP-CP map.
 Also, $f \in \Proc_{A \to B}$ is feasible if and only if
 $\mL_f$ is a trace non-increasing CP map.
\end{thm}
\begin{proof}
 From Eq.~\eqref{eq:process_test_nas},
 $f \in \ProcD_{A \to B}$ holds if and only if
 $\cra{\gdis_B} \circ f \circ \cket{\rho} = \craket{\gdis_A|\rho}$ holds
 for any $\cket{\rho} \in \St_A$.
 Since $\cra{\gdis_B} \circ f \circ \cket{\rho} = \Tr~\Mat_{f \circ \cket{\rho}}
 = \Tr~\mL_f[\Mat_\cket{\rho}]$ and
 $\craket{\gdis_A|\rho} = \Tr~\Mat_\cket{\rho}$ hold,
 $f \in \ProcD_{A \to B}$ holds if and only if $\mL_f$ is a TP-CP map.

 Assume $f \in \ProcF_{A \to B}$; then,
 from Eq.~\eqref{eq:process_feasible_Pr},
 $\Tr~\mL_f[\Mat_\cket{\rho}] = \cra{\gdis} \circ f \circ \cket{\rho} \le \craket{\gdis|\rho} =
 \Tr~\Mat_\cket{\rho}$ holds for any $\cket{\rho} \in \St_A$,
 which implies that $\mL_f$ is a trace non-increasing CP map.
 Conversely, assume that $\mL_f$ is a trace non-increasing CP map.
 $\mL_f$ can be written in the Kraus representation as
 \begin{alignat}{3}
  \mL_f(H) &= \sum_{i=1}^k E_i H E_i^\dagger,
  &\quad \forall H \in \StMat_A,
  &&\quad \sum_{i=1}^k E_i^\dagger E_i &\le \ident_\NA,
 \end{alignat}
 where $E_1,\ldots,E_k$ are $\NB \times \NA$ matrices.
 Consider $g \in \Proc_{A \to B}$ satisfying
 \begin{alignat}{3}
  \mL_g(H) &= \sum_{i=k+1}^l E_i H E_i^\dagger,
  &\quad \forall H \in \StMat_A,
  &&\quad \sum_{i=1}^l E_i^\dagger E_i &= \ident_\NA,
 \end{alignat}
 where $l \ge k$ holds and $E_{k+1},\ldots,E_l$ are $\NB \times \NA$ matrices;
 then, one can easily see that $\mL_{f + g} = \mL_f + \mL_g$ is a TP-CP map.
 Thus, $f + g \in \ProcD_{A \to B}$
 (i.e., $\{ f, g \} \in \Test_{A \to B}$) holds.
 Therefore, $f$ is feasible.
\end{proof}

\subsection{Classical and fully quantum theories} \label{subsec:derive_Quantum_ClassicalFullQuantum}

Up to now, we showed that an OPT with
the four postulates introduced in this paper is quantum theory,
which includes classical theory and fully quantum theory as special cases.
One can single out classical theory and fully quantum theory by
introducing appropriate additional postulates.
Here, to single out these theories,
we introduce postulates about perfectly distinguishable normalized pure states.
We say that a system $A$ satisfies \termdef{perfect distinguishability}
if any distinct $\cket{\psi}, \cket{\phi} \in \StNP_A$ are
perfectly distinguishable.
Also, we say that a system $A$ satisfies \termdef{indistinguishability}
if, for any $\cket{\psi}, \cket{\phi} \in \StNP_A$,
there exists $\cket{\varphi} \in \StNP_A$
such that neither $\cket{\psi}$ nor $\cket{\phi}$ is perfectly distinguishable from
$\cket{\varphi}$.
It is easily seen that, in quantum theory, a system $A$ is classical
if $A$ satisfies perfect distinguishability.
Thus, the following postulate singles out classical theory.
\begin{postulate}[\pagetarget{pos:Classical}{Perfect distinguishability}][\NumClassical]
 \label{postulate:Classical}
 Any system satisfies perfect distinguishability.
\end{postulate}
Also, in quantum theory, a system $A$ is fully quantum
if $A$ satisfies indistinguishability,
and thus the following postulate singles out fully quantum theory.
\begin{postulate}[\pagetarget{pos:FullQuantum}{Indistinguishability}][\NumFullQuantum]
 \label{postulate:FullQuantum}
 Any system satisfies indistinguishability.
\end{postulate}

\section{Conclusion}

In this study, we imposed four purely operational postulates on the framework of OPTs.
We showed that these postulates are sufficient to single out finite-dimensional quantum theory
with superselection rules.
In an OPT satisfying the first three postulates
--- the \poslink{pos:SymSharp}{symmetric sharpness},
\poslink{pos:CompletelyMixed}{complete mixing}, and \poslink{pos:Filter}{filtering} ---,
we discovered that each state space is a symmetric cone and
the corresponding effect space is its dual cone.
We showed that such an OPT is quantum theory
if it further satisfies the \poslink{pos:ProcEq}{local equality postulate}.
Moreover, classical and fully quantum theories can be, respectively, singled out from
the \poslink{pos:Classical}{perfect distinguishability} and
\poslink{pos:FullQuantum}{indistinguishability} postulates.

\begin{acknowledgments}
 I am grateful to O. Hirota, M. Sohma, T. S. Usuda, and K. Kato for valuable comments.
 I would like to thank Enago (www.enago.jp) for the English language review.
 This work was supported by JSPS KAKENHI Grant Number JP19K03658.
\end{acknowledgments}

\bibliographystyle{apsrev4-1}

\end{document}